\documentclass[journal,onecolumn,11pt,twoside]{IEEEtran}
\usepackage{epsfig,graphics,subfigure}
\usepackage{amsmath,amssymb,amsthm,multirow,balance,dsfont,hyperref}
\usepackage{cite,cleveref}
\usepackage[table]{xcolor}
\usepackage{colortbl}
\usepackage{arydshln}
\usepackage{array}
\usepackage[ruled,vlined]{algorithm2e}
\usepackage{mathrsfs}
\usepackage{empheq}
\usepackage[mathscr]{euscript}
\usepackage{rotating}
\usepackage{diagbox}

\DeclareFontFamily{U}{mathx}{\hyphenchar\font45}
\DeclareFontShape{U}{mathx}{m}{n}{
      <5> <6> <7> <8> <9> <10>
      <10.95> <12> <14.4> <17.28> <20.74> <24.88>
      mathx10
      }{}
\DeclareSymbolFont{mathx}{U}{mathx}{m}{n}
\DeclareFontSubstitution{U}{mathx}{m}{n}
\DeclareMathAccent{\widecheck}{0}{mathx}{"71}

\renewcommand{\scriptsize}{\fontsize{7.2pt}{9pt}\selectfont}
\newtheorem{theorem}{Theorem}
\newtheorem{lemma}{Lemma}

\newtheorem{example}{Example}

\newtheorem{assumption}{Assumption}
\newtheorem{definition}{Definition}

\usepackage{color}

\newcommand{\boldh}{\boldsymbol{h}}

\newcommand{\bd}{\boldsymbol{d}}
\newcommand{\bs}{\boldsymbol{s}}
\newcommand{\bsb}{\overline{\bs}}
\newcommand{\bsc}{\widecheck{\bs}}

\newcommand{\bp}{\boldsymbol{p}}
\newcommand{\bu}{\boldsymbol{u}}
\newcommand{\bv}{\boldsymbol{v}}
\newcommand{\bw}{\boldsymbol{w}}
\newcommand{\bz}{\boldsymbol{z}}
\newcommand{\bzb}{\overline{\bz}}
\newcommand{\bzc}{\widecheck{\bz}}
\newcommand{\bx}{\boldsymbol{x}}
\newcommand{\by}{\boldsymbol{y}}

\newcommand{\bpsi}{\boldsymbol{\psi}}
\newcommand{\bphi}{\boldsymbol{\phi}}
\newcommand{\bphib}{\overline{\bphi}}
\newcommand{\bphic}{\widecheck{\bphi}}

\newcommand{\bdelta}{\boldsymbol{\delta}}
\newcommand{\bchi}{\boldsymbol{\chi}}

\newcommand{\bH}{\boldsymbol{H}}

\newcommand{\cA}{\mathcal{A}}

\newcommand{\cC}{\mathcal{C}}

\newcommand{\cE}{\mathcal{E}}

\newcommand{\cS}{\mathcal{S}}

\newcommand{\cG}{\mathcal{G}}

\newcommand{\cI}{\mathcal{I}}
\newcommand{\cJ}{\mathcal{J}}

\newcommand{\cN}{\mathcal{N}}
\newcommand{\cP}{\mathcal{P}}

\newcommand{\cU}{\mathcal{U}}

\newcommand{\cV}{\mathcal{V}}

\newcommand{\cw}{{\scriptstyle\mathcal{W}}}

\newcommand{\ccw}{{\scriptscriptstyle\mathcal{W}}}

\newcommand{\ccu}{{\scriptscriptstyle\mathcal{U}}}

\newcommand{\bcB}{\boldsymbol{\cal{B}}}
\newcommand{\bcH}{\boldsymbol{\cal{H}}}
\newcommand{\bcQ}{\boldsymbol{\cal{Q}}}
\newcommand{\bcC}{\boldsymbol{\cal{C}}}
\newcommand{\bcD}{\boldsymbol{\cal{D}}}

\newcommand{\bcw}{\boldsymbol{\cw}}

\newcommand{\bwt}{\widetilde \bw}

\newcommand{\bpsit}{\widetilde \bpsi}
\newcommand{\bphit}{\widetilde \bphi}

\newcommand{\bcwt}{\widetilde\bcw}

\newcommand{\expec}{\mathbb{E}}

\newcommand{\col}{\text{col}}

\newcommand{\diag}{\text{diag}}
\newcommand{\sign}{\text{sign}}

\DeclareMathOperator*{\st}{subject~to}

\usepackage{etoolbox}
\AtBeginEnvironment{tabular}{\scriptsize}
\newcolumntype{C}[1]{>{\centering\arraybackslash}m{#1}}

\SetKwInput{KwInput}{Input}
\SetKwInput{KwDo}{Do}

\begin{document}
\title{Differential error feedback for communication-efficient  decentralized learning}
\author{Roula Nassif, \IEEEmembership{Member, IEEE},  Stefan Vlaski, \IEEEmembership{Member, IEEE}, Marco Carpentiero,  \IEEEmembership{Student Member, IEEE}\\
 Vincenzo Matta, \IEEEmembership{Senior Member,~IEEE},  Ali H. Sayed, \IEEEmembership{Fellow, IEEE}\vspace{-10mm}
\thanks{A short conference version of this work {appears in~\cite{nassif2024differential}. This extended version includes proofs, derivations, and new results.}

R. Nassif is with Universit\'e C\^ote d'Azur, I3S Laboratory, CNRS,  France (email: roula.nassif@unice.fr).  S. Vlaski is with Imperial College London, UK (e-mail: s.vlaski@imperial.ac.uk). M. Carpentiero and V.  Matta are with University of Salerno, Italy (e-mail: $\{$mcarpentiero,vmatta$\}$@unisa.it).  A. H. Sayed is with the Institute of Electrical and Micro Engineering, EPFL, Switzerland (e-mail: ali.sayed@epfl.ch).}}

\maketitle

\begin{abstract}
Communication-constrained algorithms for decentralized learning and optimization rely on {local updates coupled with the exchange of compressed signals}. In this context, \emph{differential quantization} is an effective technique to mitigate the negative impact of compression by leveraging correlations between {successive} iterates. In addition, the use of \emph{error feedback}, which consists of incorporating the compression error into subsequent steps, is a powerful mechanism to compensate for the bias caused by the compression. Under error feedback, performance guarantees in the literature have so far focused on algorithms employing a fusion center or a special class of contractive compressors that cannot be implemented with a finite number of bits. In this work, we propose a new \emph{decentralized} communication-efficient learning approach that blends differential quantization with error feedback.  The approach is specifically tailored for decentralized learning problems where agents have individual {risk} functions to minimize subject to subspace constraints that require the minimizers across the network to lie in low-dimensional subspaces. This constrained formulation includes consensus or single-task optimization as special  cases, and allows for more general task relatedness models such as multitask smoothness and  coupled optimization. We show that, under some general conditions on the compression noise, and for sufficiently small step-sizes~$\mu$, the resulting communication-efficient strategy is stable both in  terms of  mean-square error and  average bit rate: by reducing  $\mu$, it is possible to keep the \emph{estimation errors small (on the order of $\mu$) without increasing indefinitely the bit rate as $\mu\rightarrow 0$.}  The results {establish} that, in the \emph{small step-size regime} and with a  \emph{finite number of bits}, it is possible to attain the performance achievable in the absence of compression.
\end{abstract}


\newpage
\section{Introduction}
Data is increasingly being collected in a distributed and streaming manner, in an environment where communication and data privacy are becoming major concerns. In this context, centralized learning schemes with fusion centers tend to be replaced by new paradigms, such as federated and decentralized learning~\cite{mcmahan2017communication,li2020federated
,alistarch2017qsgd
,smith2017federated,sayed2014adaptation,sayed2014adaptive,sayed2013diffusion,nassif2020multitask}. In these approaches, each participating device (which is referred to as \emph{agent} or \emph{node}) has a local training dataset, which is never uploaded to the server. The training data is kept locally on users' devices, and the devices act as agents performing local computations to learn global models of interest. In applications where communication with a server becomes a bottleneck, \emph{decentralized} topologies (where agents only communicate with their neighbors) become attractive alternatives to federated topologies (where a server connects with all remote devices). These decentralized implementations reduce the communication burden since model updates are exchanged locally between agents without relying on a central coordinator~\cite{sayed2014adaptation,sayed2014adaptive,sayed2013diffusion,nassif2020multitask,kovalev2021linearly,koloskova2019decentralized}. Studies  have shown that decentralized approaches can be as efficient as the centralized schemes when considering, for instance, steady-state mean-square-error performance~\cite{sayed2014adaptation,nassif2020adaptation2}.

In traditional decentralized implementations,  agents need to exchange (possibly \emph{high-dimensional} and \emph{dense}) parameter vectors at every iteration of the learning algorithm
, leading to high communication costs. In modern distributed networks comprising a massive number of devices (e.g., thousands of participating smartphones), communication can be slower than local computation by many orders of magnitude (due to limited resources such as energy and bandwidth). Designers are typically limited by an upload bandwidth of 1MB/s or less~\cite{mcmahan2017communication}.  Therefore, in practice,  if not addressed adequately, the scarcity of the communication resources may limit the application of decentralized learning~\cite{cao2023communication}. A variety of methods have been proposed to reduce the communication overhead of decentralized learning.  These methods can be divided into two main categories. In the first one, communication is reduced by skipping communication rounds while performing a certain number of local updates in between~\cite{mcmahan2017communication,koloskova2020unified,liu2023decentralized}, thus trading-off communication overhead, computation, and learning performance. In the second one, information is compressed by employing either quantization (e.g., employing dithered quantization~\cite{aysal2008distributed}), sparsification (e.g., employing top-$k$ or rand-$k$ sparsifiers~\cite{kovalev2021linearly}), or both (e.g., employing top-$k$ combined with dithering~\cite{beznosikov2023biased}), before being exchanged.  Compression operators and  learning algorithms are then jointly designed to prevent the compression error from accumulating during the learning process and from significantly deteriorating the performance of the decentralized approach~\cite{nedic2008distributed,zhao2012diffusion,thanou2013distributed,carpentiero2022adaptive,carpentiero2024compressed,kovalev2021linearly,koloskova2019decentralized,reisizadeh2019exact,taheri2020quantized,lee2021finite,nassif2023quantization,zhao2022beer,singh2024decentralized}. 
Other works propose to combine the aforementioned two categories to further reduce the communication overhead~\cite{sparq2023sparq,singh2021squarm}. 

In this work, we introduce a new  communication-efficient approach for decentralized learning. The approach exploits  \emph{differential quantization} and \emph{error feedback} to mitigate the negative impact of compressed communications on the learning performance. \emph{Differential quantization} is a common  technique for mitigating the impact of compression by leveraging correlations between {successive} iterates.  In this case, instead of communicating compressed versions of  the iterates, agents communicate compressed versions of the differences between current estimates and their predictions based on previous iterations. Several recent works have focused on studying the benefits of differential quantization in the context of decentralized learning.  For instance, the work~\cite{koloskova2019decentralized} shows that, in a diminishing  step-size regime, differential quantization can reduce communication overhead without degrading much the learning rate. 
In~\cite{tang2018communication}, it is shown that decentralized learning can achieve the same convergence rate as centralized learning in non-convex settings, under very high accuracy constraints on the compression operators. The  constraints are relaxed in the study~\cite{zhao2022beer}, which also assumes decentralized non-convex optimization. The work~\cite{sparq2023sparq} studies the benefits of differential quantization and event-triggered communications. The analysis shows that  compression affects slightly the convergence rate  when gradients are bounded. Similar results are established in~\cite{singh2021squarm} under a weaker bounded gradient dissimilarity assumption, and in~\cite{taheri2020quantized} in the context of decentralized learning over directed graphs. The works~\cite{carpentiero2022adaptive,carpentiero2024compressed,nassif2023quantization} study the benefits of differential quantization without imposing any assumptions on the gradients, and by allowing for the use of combination matrices\footnote{As we will see, combination matrices in decentralized learning are used to control the exchange of information between neighboring agents.}  that are not symmetric~\cite{carpentiero2022adaptive,carpentiero2024compressed} or that have matrix valued entries~\cite{nassif2023quantization}. While the works~\cite{koloskova2019decentralized,tang2018communication,carpentiero2022adaptive,carpentiero2024compressed,nassif2023quantization,zhao2022beer,sparq2023sparq,singh2021squarm,taheri2020quantized} focus on studying primal stochastic optimization techniques (that are based on propagating and estimating primal variables), the works~\cite{kovalev2021linearly,singh2024decentralized} consider primal-dual  techniques and the work~\cite{lee2021finite} considers deterministic optimization. 

\emph{Error feedback}, on the other hand, consists of locally storing the compression error (i.e., the difference between the input and output of the compression operator), and incorporating it  into the next iteration. This technique has been previously employed for stochastic gradient descent (SGD) algorithms. Specifically, it has been applied to the SignSGD algorithm in the single-agent context under $1$-bit quantization~\cite{karimireddy2019error}, and to the distributed SGD to handle biased compression operators~\cite{beznosikov2023biased}. In the context of decentralized learning, the DeepSqueeze approach in~\cite{tang2019deepsqueeze} uses error feedback without differential quantization.

In the current work, we show how to blend differential quantization and error feedback in order to obtain a communication-efficient decentralized learning algorithm. First, we describe in Sec.~\ref{sec: Problem setup} the decentralized learning framework and the class of compression operators considered in the study. While most existing works on decentralized learning with compressed communications are focused on single-task or consensus algorithm design, the design in the current work goes beyond this traditional focus by allowing for both \emph{single-task} and \emph{multitask} implementations. In single-task learning, nodes collaborate to reach an agreement on a single parameter vector (also referred to as \emph{task} or \emph{objective}) despite having different local data distributions. Multitask learning, on the other hand, involves training multiple  tasks simultaneously and exploiting their intrinsic relationship. This approach offers several advantages, including improved network performance, especially when the tasks share commonalities in their underlying features~\cite{nassif2020multitask}. Compared with previous works, another {contribution} in the current study {is} the consideration of a general class of compression operators. {Specifically, rather than} being confined to the set of  probabilistic unbiased operators {as in~\cite{nassif2023quantization,carpentiero2022adaptive,carpentiero2024compressed,kovalev2021linearly,tang2018communication}}, we allow {for the use} of biased  (possibly deterministic) compression operators. Moreover, while most existing works assume that some quantities (e.g., the norm or some components of the vector to be quantized) are represented with very high precision (e.g., machine precision) and neglect the associated quantization error~\cite{kovalev2021linearly,koloskova2019decentralized,singh2024decentralized,tang2018communication,carpentiero2022adaptive,carpentiero2024compressed,zhao2022beer,sparq2023sparq,singh2021squarm,taheri2020quantized,tang2019deepsqueeze}, the current work incorporates realistic quantization models into the compression process and shows how to effectively manage the errors and minimize their impact on the learning performance. In Secs.~\ref{sec: Decentralized algorithmic framework: compressed communications} and~\ref{sec: Stochastic performance analysis},  we present and analyze the proposed learning strategy. While there exist several theoretical works investigating communication-efficient learning, the analysis in the current work is more general  in the following sense. {First, it considers a general class of compression schemes.} Moreover, unlike the studies in~\cite{kovalev2021linearly,koloskova2019decentralized,singh2024decentralized,lee2021finite,tang2018communication,zhao2022beer,sparq2023sparq,singh2021squarm,tang2019deepsqueeze}, it does not require the combination matrices to be symmetric. 
We further allow the entries of the combination matrix to be matrix-valued (as opposed to scalar valued as in traditional implementations) in order to solve general multitask optimization problems. 
{Finally, we do} not assume bounded gradients as in~\cite{koloskova2019decentralized,sparq2023sparq,taheri2020quantized,singh2024decentralized,tang2019deepsqueeze}. For ease of reference, the modeling conditions from this and related works are summarized in Table~\ref{table: modeling conditions}. 
We establish in Sec.~\ref{sec: Stochastic performance analysis} the mean-square-error stability of the proposed decentralized communication-efficient approach. In addition to investigating the mean-square-error stability, we characterize the steady-state average bit rate of the proposed approach  when  variable-rate quantizers are used. The analysis  shows that, by properly designing the quantization operators, the iterates generated by the proposed approach lead to small estimation errors on the order of the step-size $\mu$ (as it happens in the {\em ideal case without compression}), while concurrently guaranteeing a bounded average bit rate as $\mu\rightarrow 0$. Our theoretical findings {show} that, in the \emph{small step-size regime}, the proposed strategy attains the \emph{performance achievable in the absence of compression, despite the use of a finite number of bits}. This demonstrates the effectiveness of the approach in maintaining performance while reducing communication overhead. Finally, we present  in Sec.~\ref{sec: simulation results} experimental results illustrating the theoretical findings and showing that blending differential {compression} and error feedback can achieve superior  performance compared to state-of-the-art baselines.

{\scriptsize
\begin{table}{
\caption{Comparison of modeling assumptions for decentralized stochastic optimization studies. All works employ differential quantization {without error feedback}, except the one marked with $^{\star}$ and {our} work. While the work with $^{\star}$ uses error feedback, {our} work employs error feedback with differential quantization. {All works assume that some quantities are exchanged with very high precision, except the one marked with $^{\dagger}$ and this work. We use the symbol -- in the last column for works that do not have bounded gradient assumptions.} }
\label{table: modeling conditions}
\vspace{-3mm}
\begin{center}{
\begin{tabular}{||c ||c|c |c |c|c||} 
 \hline \hline
\cellcolor[gray]{0.7} {\scriptsize Reference}&\cellcolor[gray]{0.7} {\scriptsize Stochastic optimization context} &\cellcolor[gray]{0.7}{ \scriptsize Combination matrix}&\cellcolor[gray]{0.7}{\scriptsize Compression operator}&\cellcolor[gray]{0.7} {\scriptsize Step-size}&\cellcolor[gray]{0.7} {\scriptsize Gradient assumption}
\\ [0.5ex] 
 \hline\hline
  \multirow{1}{*}{\scriptsize \cite{koloskova2019decentralized}}& \multirow{1}{*}{\scriptsize Primal, consensus-type}&  \multirow{1}{*}{\scriptsize Symmetric} & \multirow{1}{*}{\scriptsize Can be deterministic \& biased}& \multirow{1}{*}{\scriptsize Diminishing}& \multirow{1}{*}{\scriptsize Bounded} 
  \\ \hline
  \multirow{1}{*}{\scriptsize \cite{kovalev2021linearly}}& \multirow{1}{*}{\scriptsize Primal-dual, consensus-type}&  \multirow{1}{*}{\scriptsize Symmetric} & \multirow{1}{*}{\scriptsize Probabilistic \& unbiased}& \multirow{1}{*}{\scriptsize Constant}& \multirow{1}{*}{--} 
  \\ \hline
  \multirow{1}{*}{\scriptsize \cite{sparq2023sparq}}& \multirow{1}{*}{\scriptsize Primal, consensus-type}&  \multirow{1}{*}{\scriptsize Symmetric} & \multirow{1}{*}{\scriptsize Can be deterministic \& biased}& \multirow{1}{*}{\scriptsize Diminishing}& \multirow{1}{*}{\scriptsize Bounded}
  \\ \hline
  \multirow{1}{*}{\scriptsize \cite{singh2021squarm}}& \multirow{1}{*}{\scriptsize Primal, consensus-type}&  \multirow{1}{*}{\scriptsize Symmetric} & \multirow{1}{*}{\scriptsize Can be deterministic \& biased}& \multirow{1}{*}{\scriptsize Constant}& \multirow{1}{*}{\scriptsize Bounded dissimilarities} 
  \\ \hline
  \multirow{1}{*}{\scriptsize \cite{singh2024decentralized}}& \multirow{1}{*}{\scriptsize Primal-dual, multitask}&  \multirow{1}{*}{\scriptsize Symmetric} & \multirow{1}{*}{\scriptsize Can be deterministic \& biased}& \multirow{1}{*}{\scriptsize Constant}& \multirow{1}{*}{\scriptsize Bounded}
  \\ \hline
  \multirow{1}{*}{\scriptsize \cite{tang2018communication}}& \multirow{1}{*}{\scriptsize Primal, consensus-type}&  \multirow{1}{*}{\scriptsize Symmetric} & \multirow{1}{*}{\scriptsize Probabilistic \& unbiased}& \multirow{1}{*}{\scriptsize Constant}& \multirow{1}{*}{--} 
\\ \hline
  \multirow{1}{*}{\scriptsize \cite{zhao2022beer}}& \multirow{1}{*}{\scriptsize Primal, consensus-type}&  \multirow{1}{*}{\scriptsize Symmetric} & \multirow{1}{*}{\scriptsize Can be deterministic \& biased}& \multirow{1}{*}{\scriptsize Constant}& \multirow{1}{*}{--} 
  \\ \hline
  \multirow{1}{*}{\scriptsize \cite{taheri2020quantized}}& \multirow{1}{*}{\scriptsize Primal, consensus-type}&  \multirow{1}{*}{\scriptsize Can be non-symmetric} & \multirow{1}{*}{\scriptsize Probabilistic \& unbiased}& \multirow{1}{*}{\scriptsize Constant}& \multirow{1}{*}{\scriptsize Bounded} 
  \\ \hline
 \multirow{1}{*}{\scriptsize \cite{tang2019deepsqueeze}$^{\star}$}& \multirow{1}{*}{\scriptsize Primal, consensus-type}&  \multirow{1}{*}{\scriptsize Symmetric} & \multirow{1}{*}{\scriptsize Can be deterministic \& biased}& \multirow{1}{*}{\scriptsize Constant}& \multirow{1}{*}{\scriptsize Bounded dissimilarities} 
 \\ \hline

  \multirow{1}{*}{\scriptsize \cite{carpentiero2022adaptive}}& \multirow{1}{*}{\scriptsize Primal, consensus-type}&  \multirow{1}{*}{\scriptsize Can be non-symmetric} & \multirow{1}{*}{\scriptsize Probabilistic \& unbiased}& \multirow{1}{*}{\scriptsize Constant}& \multirow{1}{*}{--} 
  \\ \hline
 \multirow{2}{*}{\scriptsize \cite{nassif2023quantization}$^{\dagger}$ }&  \multirow{2}{*}{\scriptsize Primal, consensus-type \& multitask}& {\scriptsize Can be non-symmetric} & {\scriptsize Probabilistic \& unbiased}& \multirow{2}{*}{\scriptsize Constant}& \multirow{2}{*}{--} 
 \\
\multirow{1}{*}{}&&{\scriptsize with matrix-valued block entries}  &{\scriptsize no high precision representations}&& 
\\ \hline
\multirow{2}{*}{\scriptsize This work}&   \multirow{2}{*}{\scriptsize Primal, consensus-type \& multitask}& {\scriptsize Can be non-symmetric} &{\scriptsize Can be deterministic \& biased}& \multirow{2}{*}{ \scriptsize Constant}& \multirow{2}{*}{--} 
\\
\multirow{1}{*}{}& &{ \scriptsize with matrix-valued block entries}  &{ \scriptsize no high precision representations}&& 
\\ \hline
 \hline
\end{tabular}}
\end{center}}
\end{table}}

\noindent\textbf{Notation:} All vectors are column vectors. Random quantities are denoted in boldface. Matrices are denoted in uppercase letters while vectors and scalars are denoted in {lowercase} letters. The symbol $(\cdot)^\top$ denotes matrix transposition. The operator $\col\{\cdot\}$ stacks the column vector entries on top of each other. The operator $\diag\{\cdot\}$ 
 forms a matrix from block arguments by placing each block immediately below and to the right of its predecessor. The symbol $\otimes $ denotes the Kronecker product. {The ceiling and floor functions are denoted by $\lceil \cdot \rceil$ and $\lfloor \cdot \rfloor$, respectively.} The $M\times M$ identity matrix is denoted by $I_M$. The abbreviation ``w.p.'' is used for ``with probability''.  
 The notation $\alpha=O(\mu)$ signifies that there exist two positive constants $c$ and $\mu_0$ such that $|\alpha|\leq c\mu$ for all $\mu\leq\mu_0$. {Vectors or matrices} of all zeros {are} denoted by 0.

\section{Problem setup}
\label{sec: Problem setup}

In this section, we formally state the decentralized optimization problem and introduce strategies, quantities, and assumptions that will be used in subsequent sections.

\subsection{Decentralized optimization under subspace constraints}
\label{subsec: Decentralized optimization under subspace constraints}

 We consider a connected  graph (or network) $\cG(\cV,\cE)$, where $\cV$ and $\cE$ denote the set of $K$ agents or nodes (labeled $k=1,\ldots, K$) and  the set of possible communication links or edges, respectively.  Let $w_k\in\mathbb{R}^{M_k}$ denote some parameter vector at agent $k$ and let $\cw=\col\{w_1,\ldots,w_K\}$ denote the {$M$-dimensional} vector (where $M=\sum_{k=1}^KM_k$) collecting the parameter vectors from across the network. We associate with each agent $k$ a differentiable convex {risk} $J_k(w_k):\mathbb{R}^{M_k}\rightarrow\mathbb{R}$
 , expressed as the expectation of some loss function $L_k(\cdot)$ and written as $J_k(w_k)=\expec L_k(w_k;\by_k)$, where $\by_k$ denotes the random data at agent $k$. The expectation is computed relative to the distribution of the local data. In the stochastic setting, when the data distribution is unknown, the {risks} $J_k(\cdot)$ and their gradients $\nabla_{w_k}J_k(\cdot)$ are unknown. In this case, instead of using the true gradient, it is common to use approximate gradient {vectors based on the loss functions such as} {$\widehat{\boldsymbol{\nabla}_{w_k}J_k}(w_k)=\nabla_{w_k}L_k(w_k;\by_{k,i})$} where $\by_{k,i}$ represents the data observed at iteration $i$~\cite{sayed2014adaptation,sayed2022inference}.

In traditional \emph{single-task} or \emph{consensus} problems, agents need to agree on a common parameter {vector (also called \emph{model} or \emph{task})} corresponding to the minimizer of the {following weighted sum of individual {risks}:}
\begin{equation}
\label{eq: consensus formulation}
w^o=\arg\min_{w\in\mathbb{R}^{{M_c}}}\frac{1}{K}\sum_{k=1}^KJ_k(w),
\end{equation}
where {$M_c$} represents a \emph{common} vector length, i.e., {in this case, the dimensions $M_k$ for all agents are identical and equal to $M_c$. Moreover,} $w$ is the global parameter vector, which all agents need to agree upon. Each agent seeks to estimate $w^o$ through local computations and communications among neighboring agents without the need to know any of the {risks or losses} besides {its} own. Among many useful strategies that have been proposed in the literature~\cite{sayed2013diffusion,sayed2014adaptation,sayed2014adaptive,nedic2009distributed,bertsekas1997new,dimakis2010gossip}, diffusion strategies~\cite{sayed2013diffusion,sayed2014adaptation,sayed2014adaptive} are particularly attractive since they are scalable, robust, and enable continuous learning and adaptation in response to drifts in the location of the minimizer.

In this work, instead of considering the single-task formulation~\eqref{eq: consensus formulation}, we consider a generalization that allows the network to solve {\emph{multitask}} optimization problems. {Multitask learning} is suitable for network applications {where differences in the  data distributions} 
require more complex models and  more flexible algorithms than single-task implementations. In multitask networks, agents generally  need to estimate and track multiple distinct, though related, {models or objectives}. For instance, in distributed power system state estimation, the local state vectors to be estimated at neighboring control centers may overlap partially since the areas in a power system are interconnected~\cite{kekatos2013distributed}. Likewise, in weather forecasting applications, regional differences in the collected data distributions require agents to exploit the correlation profile in the data for enhanced decision rules~\cite{nassif2020learning}. Existing strategies to address multitask problems generally exploit prior knowledge on how the tasks across the network relate to each other~\cite{nassif2020multitask}.   For example, one way to model relationships among tasks is to formulate convex optimization problems with appropriate co-regularizers between neighboring agents~\cite{nassif2020learning,nassif2020multitask}. Another way to leverage the relationships among tasks  is to constrain the model parameters to lie within certain subspaces that can represent for instance shared latent patterns that are common across tasks~\cite{nassif2020adaptation,nassif2020multitask,plata2017heterogeneous}. The choice between these techniques depends in general on the specific characteristics of the tasks and the desired trade-offs between model complexity and performance~\cite{nassif2020multitask,plata2017heterogeneous}.

In this work, we study \emph{decentralized learning under subspace constraints in the presence of compressed communications}. {Specifically,} we consider inference problems of the form: 
\begin{equation}
\label{eq: network constrained problem}
\begin{split}
\cw^o=&\arg\min_{\ccw}~\sum_{k=1}^KJ_k(w_k),\\
&\st ~\cw\in\text{Range}(\cU)
\end{split}
\end{equation}
where {the matrix $\cU$ is an $M\times P$ full-column rank matrix (with $P\ll M$) assumed to be semi-unitary, i.e., its columns are orthonormal ($\cU^\top\cU=I_P$). By using the {stochastic gradient} of the individual {risk} $J_k(w_k)$, agent~$k$ seeks to estimate the $k-$th $M_k\times 1$ subvector $w^o_k$ of the network vector $\cw^o=\col\{w^o_1,\ldots,w^o_K\}$, which is required to lie in {a} low-dimensional subspace. While the objective in~\eqref{eq: network constrained problem} is  additively separable, the subspace constraint couples the models across the agents.} Constrained formulations of the form~\eqref{eq: network constrained problem} have  been studied previously in decentralized settings where communication constraints are absent~\cite{nassif2020adaptation,dilorenzo2020distributed}. As explained in~\cite{nassif2020multitask,dilorenzo2020distributed},~\cite[Sec.~II]{nassif2020adaptation}, by properly selecting the matrix $\cU$, formulation~\eqref{eq: network constrained problem} can be tailored to address a wide range of optimization problems encountered in network applications. Examples include $i)$ consensus or single-task optimization (where the agents' objective is to reach consensus on the minimizer $w^o$ in~\eqref{eq: consensus formulation}), $ii)$ decentralized coupled optimization (where the parameter vectors to be estimated at neighboring agents are partially overlapping)~\cite{mota2015distributed,plata2017heterogeneous,alghunaim2020distributed,kekatos2013distributed}, and $iii)$ multitask inference under smoothness (where the network parameter vector $\cw$ to be estimated is required to be smooth w.r.t. the underlying network topology)~\cite{nassif2020adaptation,nassif2020learning}. For instance, setting in~\eqref{eq: network constrained problem} 
$\cU=\frac{1}{\sqrt{K}}(\mathds{1}_K\otimes I_{{M_c}})$ where $\mathds{1}_K$ is the $K\times 1$ vector of all ones, we obtain an optimization problem equivalent to the consensus problem~\eqref{eq: consensus formulation}. While projecting onto the space spanned by the vector of all ones allows to enforce consensus across the network, graph smoothness can in general be promoted by projecting onto the space spanned by the eigenvectors of the graph Laplacian corresponding to small eigenvalues~\cite{nassif2020multitask,nassif2020adaptation,dilorenzo2020distributed}.
\subsection{Decentralized diffusion-based approach}
\label{subsec: Decentralized diffusion-based approach}

To solve problem~\eqref{eq: network constrained problem} in a decentralized manner, we consider the primal approach proposed and studied in~\cite{nassif2020adaptation,nassif2020adaptation2}, namely, 
\begin{subequations}
\label{eq: decentralized learning approach}
 \begin{empheq}[left={\empheqlbrace\,}]{align}
\bpsi_{k,i}&=\bw_{k,i-1}-\mu\widehat{\nabla_{w_k}J_k}(\bw_{k,i-1})\label{eq: step1}\\
\bw_{k,i}&=\sum_{\ell\in\cN_k}A_{k\ell}\bpsi_{\ell,i}\label{eq: step2}
 \end{empheq}
\end{subequations}
where $\cN_k$ {denotes} the set of nodes connected to agent $k$ by a communication link (including node $k$ itself) and $\mu>0$ is a small step-size parameter. Note that the information sharing across agents in~\eqref{eq: decentralized learning approach} is implemented by means of a $K\times K$ block  combination matrix $\cA=[A_{k\ell}]$ that has a zero $M_k\times M_{\ell}$ block element $(k,\ell)$ if nodes $k$ and $\ell$ are not neighbors, i.e., $A_{k\ell}=0$ if $\ell\notin\cN_k$, and satisfies the following conditions~\cite{nassif2020adaptation}:
\begin{equation}
\label{eq: condition A}
\cA\,\cU=\cU,\quad \cU^\top\cA=\cU^\top,~~ \text{and~ }\rho(\cA-\cP_{\ccu})<1,
\end{equation}
where $\rho(\cdot)$ denotes the spectral radius of its matrix argument, and $\cP_{\ccu}=\cU\cU^\top$ is the orthogonal projection matrix onto $\text{Range}(\cU)$.   {It is shown~\cite{nassif2020adaptation2,dilorenzo2020distributed}  how combination matrices $\cA$ satisfying~\eqref{eq: condition A} can be constructed.} If strategy~\eqref{eq: decentralized learning approach} is employed to solve the single-task problem~\eqref{eq: consensus formulation}, we can select the combination matrix $\cA$ in the form $A\otimes I_{{M_c}}$, where $A=[a_{k\ell}]$ is a $K\times K$ doubly-stochastic matrix satisfying:
\begin{equation}
\label{eq: doubly stochastic condition}
a_{k\ell}\geq 0,\quad A\mathds{1}_K=\mathds{1}_K,\quad\mathds{1}_K^\top A= \mathds{1}_K^\top,\quad a_{k\ell}=0\text{ if }\ell\notin\cN_k,
\end{equation}
In this case, conditions~\eqref{eq: condition A} {are satisfied for  $\cU=\frac{1}{\sqrt{K}}(\mathds{1}_K\otimes I_{{M_c}})$~\cite{nassif2020adaptation}} and strategy~\eqref{eq: decentralized learning approach} reduces to the standard \emph{diffusion Adapt-Then-Combine (ATC)} approach~\cite{sayed2014adaptation,sayed2014adaptive,sayed2013diffusion}:
\begin{subequations}
\label{eq: diffusion strategy}
 \begin{empheq}[left={\empheqlbrace\,}]{align}
\bpsi_{k,i}&=\bw_{k,i-1}-\mu\widehat{\nabla_{w_k}J_k}(\bw_{k,i-1})\label{eq: diffusion strategy step1}\\
\bw_{k,i}&=\sum_{\ell\in\cN_k}a_{k\ell}\bpsi_{\ell,i}\label{eq: diffusion strategy step2}
 \end{empheq}
\end{subequations}
This fact motivates the use of the terminology ``diffusion Adapt-Then-Combine (ATC) approach'' in the sequel when referring to the decentralized strategy~\eqref{eq: decentralized learning approach} for solving general constrained optimization problems of the form~\eqref{eq: network constrained problem}.

The first step~\eqref{eq: step1} in the ATC approach~\eqref{eq: decentralized learning approach} is the \emph{self-learning} step corresponding to the stochastic gradient descent step on the individual {risk} $J_k(\cdot)$. This step is followed by the \emph{social learning} step~\eqref{eq: step2} where agent $k$ receives the intermediate estimates $\{\bpsi_{\ell,i}\}$ from its neighbors $\ell\in\cN_k$ and combines them through $\{A_{k\ell}\}$ to form $\bw_{k,i}$, which corresponds to the estimate of $w^o_k$ at agent $k$ and iteration $i$. To alleviate the communication bottleneck resulting from the exchange of the intermediate estimates  among agents over many iterations, compressed communication must be considered. Before presenting the communication-efficient variant of the ATC approach~\eqref{eq: decentralized learning approach}, we describe in the following the class of compression operators considered in this study.

\subsection{Compression operator}
\label{subsec: Compression operators}


For the sake of clarity, we first introduce the following formal definitions for key concepts relating to data compression.
\begin{definition}{\emph{\textbf{(Compression operator)}.}} {Let $L$ represent a generic vector length. A compression} operator $\bcC:\mathbb{R}^L\rightarrow\mathbb{R}^L$ associates to {every input  $x\in\mathbb{R}^L$ a  random quantity $\bcC(x)\in\mathbb{R}^L$ that is} governed by the conditional probability measure~$\mathbb{P}(\cdot|x)$.

\end{definition}
Note that the above family of compression operators includes deterministic mappings as a particular case.

\begin{definition}{\emph{\textbf{(Bounded-distortion compression operator)}.}} 
\label{def: bounded distortion}A bounded-distortion compression operator is a compression operator that {fulfills} the property:
\begin{equation}
\label{eq: bounded quantization noise variance condition}
\expec\|x-\bcC(x)\|^2\leq \beta^2_c\|x\|^2+\sigma^2_c,
\end{equation}
for some $\beta^2_c\geq 0$ and $\sigma^2_c\geq 0$, and where the expectation is taken over the conditional probability  measure~$\mathbb{P}(\cdot|x)$.

\end{definition}

\begin{definition}{\emph{\textbf{(Unbiased compression operator)}.}} An unbiased compression operator is a compression operator that {fulfills} the property:
\begin{equation}
\label{eq: unbiasedness condition}
\expec[\bcC(x)]=x.
\end{equation}

\end{definition}
In this study, we consider \emph{bounded-distortion compression operators}.  Table~\ref{table: examples of quantizers} provides a list of bounded-distortion operators commonly used in decentralized learning, with the corresponding compression noise parameters $\beta^2_c$ and $\sigma^2_c$, and the bit-budget required to encode an input vector $x\in\mathbb{R}^L$. By comparing the reported schemes, we observe that the ``rand-$c$'', ``randomized Gossip'', 
 and ``top-$c$'' can be considered as \emph{sparsifiers} that map a full vector into a sparse version thereof. For instance, the rand-$c$ scheme selects randomly $c$ components of the input vector and encodes them with very high precision ($32$ or $64$ {bits} are typical values for encoding a scalar). These bits are then communicated over the links in addition to the bits encoding the locations of the selected entries. On the other hand, the QSGD scheme encodes the norm of the input vector with very high precision. In addition to encoding the norm, $L$-bits are used to encode the signs of the input vector components and $L\lceil\log_2(s)\rceil$ to encode the levels. In comparison, the  probabilistic uniform  and  probabilistic ANQ \emph{quantizers} do not make any assumptions on the high-precision representation of specific variables. 
{In the following, we highlight the key facts regarding the compression operators considered in this study.}
\begin{table*}
\caption{Examples of bounded-distortion compression operators. {For each scheme, we report the compression rule, the parameters $\beta_c^2$ and $\sigma^2_c$ in~\eqref{eq: bounded quantization noise variance condition}, and the bit-budget. $B_{\text{HP}}$ denotes the number of bits required to encode a  scalar with high precision (Typical values for $B_{\text{HP}}$ are $32$ or $64$). The operator marked with $^\dagger$ is deterministic. 
The operators marked with $^\star$ are unbiased.} }
\label{table: examples of quantizers}
\vspace{-7mm}
\begin{center}
\begin{tabular}{||>{\centering\arraybackslash}m{0.905in} ||c ||>{\centering\arraybackslash}m{0.62in} |>{\centering\arraybackslash}m{0.17in} ||>{\centering\arraybackslash}m{1.17in}||} 
 \hline \hline
\cellcolor[gray]{0.8} {\tiny{Name}} &\cellcolor[gray]{0.8} {\tiny Rule} &\cellcolor[gray]{0.8}{\tiny $\beta^2_{c}$} &\cellcolor[gray]{0.8} {\tiny $\sigma^2_{c}$}&\cellcolor[gray]{0.8} {\tiny Bit-budget} \\ [0.5ex] 
 \hline\hline
 \multirow{1}{*}{{\tiny No compression~\cite{nassif2020adaptation}} }& {\tiny $\bcC(x)=x$} &\multirow{1}{*}{\tiny $0$} &\multirow{1}{*}{\tiny$0$}&{\tiny$LB_{\text{HP}}$}\\ \hline  
 \multirow{3}{*}{\tiny Probabilistic uniform}&  {\tiny$\left[\bcC(x)\right]_j=\Delta\cdot\boldsymbol{n}(x_j) $}&\multirow{3}{*}{\tiny 0} &\multirow{3}{*}{\tiny$L\frac{\Delta^2}{4}$}&\multirow{3}{*}{
\tiny$\boldsymbol{r}(x)$ defined in~\eqref{eq: bit rate general formula}}\\ 
 \multirow{2}{*}{\tiny or dithered} &  {\tiny $\boldsymbol{n}(x_j)=\left\lbrace\begin{array}{ll}
m,&\text{w.p. }
\frac{(m+1)\Delta-x_j}{\Delta},\\
m+1,&\text{w.p. }\frac{x_j-m\Delta}{\Delta},
 \end{array}\right.\qquad m=\left\lfloor
\frac{x_j}{\Delta}
\right\rfloor$} &&&\\ 
{\tiny quantizer$^\star$~\cite{aysal2008distributed,reisizadeh2019exact}} & {\tiny $\Delta$ is the quantization step}  &&&\\ 
 \hline
 \multirow{5}{*}{\tiny Probabilistic ANQ$^\star$~\cite{lee2021finite} }& {\tiny $\left[\bcC(x)\right]_j=y_{\boldsymbol{n}(x_j)}$} &\multirow{5}{*}{\tiny $2\omega^2$} &\multirow{5}{*}{\tiny $2L\eta^2$}&\multirow{5}{*}{\tiny $\boldsymbol{r}(x)$ defined in~\eqref{eq: bit rate general formula}}\\ 
 & {\tiny $\boldsymbol{n}(x_j)= \left\lbrace\begin{array}{ll}
m,&\text{w.p. } \frac{y_{m+1}-x_j}{y_{m+1}-y_{m}},\\
m+1,&\text{w.p. } \frac{x_j-y_{m}}{y_{m+1}-y_{m}},
 \end{array}\right.\quad m=\left\lfloor
\sign(x_j)\frac{\ln\left(1+\frac{\omega}{\eta}|x_j|\right)}
{2\ln \left( \omega + \sqrt{1+\omega^2} \right)}
\right\rfloor$} & &&\\ 
&{\tiny $y_{m}=\sign(m)\frac{\eta}{\omega}\left[
\left(\omega + \sqrt{1+\omega^2}\right)^{2 |m|}
- 1
\right]$}
 & &&\\ 
&{\tiny $\omega$ and $\eta$ are two non-negative design parameters }& &&\\ \hline
 \multirow{2}{*}{\tiny Rand-$c$~\cite{kovalev2021linearly,koloskova2019decentralized} } & {\tiny$\left[\bcC(x)\right]_j=\alpha\cdot\left\lbrace\begin{array}{ll}
x_j,&\text{if  }x_j\in\Omega_c\\
0,&\text{otherwise }
 \end{array}\right.\quad\text{where }\alpha=\left\lbrace\begin{array}{ll}
1,&\text{biased version}\\
\frac{L}{c},&\text{unbiased version}
 \end{array}\right. $}&\multirow{2}{*}{\tiny$\alpha\left(1-\frac{c}{L}\right)$} &\multirow{2}{*}{\tiny $0$}&\multirow{2}{*}{\tiny {$cB_{\text{HP}}+c\lceil\log_2(L)\rceil$}}\\ 
&{\tiny $\Omega_c$ is a set of $c$ randomly selected coordinates, $c\in\{1,\ldots,L\}$}& &&\\ 
  \hline
\multirow{2}{*}{\tiny Randomized Gossip~\cite{koloskova2019decentralized} }& {\tiny$\bcC(x)=\alpha\cdot\left\lbrace\begin{array}{ll}
x,&\text{w.p. }q\\
0,&\text{w.p. }1-q
 \end{array}\right.\quad\text{where }\alpha=\left\lbrace\begin{array}{ll}
1,&\text{biased version}\\
\frac{1}{q},&\text{unbiased version}
 \end{array}\right.$}&\multirow{2}{*}{\tiny $\alpha(1-q)$} &\multirow{2}{*}{\tiny $0$}&\multirow{2}{*}{{\tiny $LB_{\text{HP}}q\qquad$ (on average)}}\\
&{\tiny $q\in(0,1]$ is the transmission probability}&&&\\ 
 \hline
  \multirow{3}{*}{\tiny QSGD$^\star$~\cite{alistarch2017qsgd} }& {\tiny $\left[\bcC(x)\right]_j=\|x\|\cdot\sign(x_j)\cdot\frac{\boldsymbol{n}(x_j,x)}{s}$ }&\multirow{3}{*}{\tiny$\min\left(\frac{L}{s^2},\frac{\sqrt{L}}{s}\right)$} &\multirow{3}{*}{\tiny$0$}&\multirow{3}{*}{\tiny{$B_{\text{HP}}+L+L\lceil\log_2(s)\rceil$}}\\ 
 &{\tiny $\boldsymbol{n}(x_j,x)=\left\lbrace\begin{array}{ll}
m,&\text{w.p. }(m+1)-\frac{|x_j|}{\|x\|}s,\\
 m+1,&\text{w.p. }\frac{|x_j|}{\|x\|}s-m,
 \end{array}\right.\qquad m=\left\lfloor s \frac{|x_j|}{\|x\|}\right\rfloor$ }& &&\\ 
&{\tiny$s$ is the number of quantization levels }
 & &&\\ 
 \hline
  \multirow{2}{*}{\tiny Top-$c$$^{\dagger}$ sparsifier~\cite{koloskova2019decentralized} }& {\tiny $\left[\cC(x)\right]_j=\left\lbrace\begin{array}{ll}
x_j,&\text{if  }x_j\in\Omega_c\\
0,&\text{otherwise }
 \end{array}\right.$ $\qquad\qquad(c\in\{1,\ldots,L\})$}&\multirow{2}{*}{\tiny$1-\frac{c}{L}$} &\multirow{2}{*}{\tiny $0$}&\multirow{2}{*}{\tiny {$cB_{\text{HP}}+c\lceil\log_2(L)\rceil$}}\\ 
&{\tiny $\Omega_c$ is a set of the $c$ coordinates with highest magnitude }& &&\\ 
  \hline
 \hline
\end{tabular}
\end{center}
\end{table*}


\subsubsection{Allowing for an absolute compression noise term} Many existing works focus on studying decentralized learning in the presence of bounded-distortion compression operators that satisfy condition~\eqref{eq: bounded quantization noise variance condition} with the \emph{absolute noise} term $\sigma^2_{c}=0$~\cite{kovalev2021linearly,koloskova2019decentralized,singh2024decentralized,tang2018communication,carpentiero2022adaptive,carpentiero2024compressed,zhao2022beer,sparq2023sparq,singh2021squarm,taheri2020quantized,tang2019deepsqueeze}. In contrast, the analysis in the current work is conducted in the presence of both the \emph{relative} (captured through $\beta^2_c$) and the {absolute} compression noise terms. 
 As explained in~\cite{lee2021finite}, neglecting the effect of $\sigma^2_{c}$ requires that some quantities (e.g., the norm of the vector  in QSGD) are represented with no quantization error, in practice at the machine precision. We refer the reader to~\cite{nassif2023quantization} for a description of a  framework for designing randomized compression operators that do not require high-precision quantization of specific variables. Particularly, Sec.~II in~\cite{nassif2023quantization} describes the design of the \emph{probabilistic uniform and ANQ} quantizers endowed with a \emph{variable-rate} coding scheme from~\cite{lee2021finite} to adapt the bit rate based on the quantizer input. When these rules, which are listed in Table~\ref{table: examples of quantizers} (rows $2$--$3$), are applied entrywise to a vector $x\in\mathbb{R}^L$, the overall (random) bit budget will be equal to~\cite{nassif2023quantization}:
\begin{equation}
\label{eq: bit rate general formula}
\boldsymbol{r}(x)=\log_2(3)\sum_{j=1}^L(1+\lceil\log_2(|\boldsymbol{n}(x_j)|+1)\rceil).
\end{equation}

\subsubsection{Allowing for the use of biased compression operators} Although the communication-efficient  approach developed in~\cite{nassif2023quantization} can be used for solving decentralized learning under subspace constraints and makes no assumptions about encoding quantities with high precision, it is not designed to handle biased compression operators, i.e., operators that do not satisfy the unbiasedness condition~\eqref{eq: unbiasedness condition}. As we will see, by incorporating explicitly the error feedback mechanism into the differential quantization approach proposed in~\cite{nassif2023quantization}, we can address biased compression by filtering the compression error over time. In general, biased compression operators tend to outperform their unbiased counterparts~\cite{beznosikov2023biased}. 

 {While the list of compression operators in Table~\ref{table: examples of quantizers} provides several examples of interest, it is not exhaustive. As we will see, through concatenation, it is possible to achieve other meaningful schemes.}
 
\begin{example}{\textbf{\emph{{(Concatenation of compression schemes).}}}}
\label{example: top-c quantizer}
In this example, we investigate a specific type of compression operator consisting of concatenating two distinct bounded-distortion compression operators. The first operation, known as the top-$c$ sparsifier, entails retaining only the $c$ {largest-magnitude} components of the input vector. The second operation involves quantizing the output of the top-$c$ sparsifier. This concatenation is particularly noteworthy as it tends to require the fewest number of bits for representation by exploiting the inherent sparsity induced by the sparsifier and by concentrating the quantization process on the most significant components of the input. Numerical results illustrating {the benefits of the concatenation} are provided in Sec.~\ref{sec: simulation results}.

\begin{definition}{\emph{\textbf{(Top-$c$ quantizer)}.}} 
\label{def: definition of the top-c quantizer}
Let $\bcQ(\cdot)$ be a bounded-distortion compression operator with parameters $\beta^2_q$ and $\sigma^2_q$. Let $\cS(\cdot)$ be the top-$c$ sparsifier, i.e., the deterministic compression operator listed in Table~\ref{table: examples of quantizers} (row 7), which can also be defined as~\cite{stich2018sparsified,basu2019qsparse}:
\begin{equation}
\left[\cS(x)\right]_j=\left\lbrace\begin{array}{ll}
x_j,&\emph{if } j\in\cI_c\\
0,&\emph{otherwise}
\end{array}
\right.
\end{equation}
where $\cI_c$ is the set containing the indices of the $c$ largest-magnitude components of $x$. In case of ties (i.e., when the first $c$ components are not uniquely determined), any tie-break rule is permitted. The top-$c$ quantizer operator is defined as:
\begin{equation}
\label{eq: top-c quantizer}
\bcC(x)=\alpha\cdot\bcQ(\cS(x)),\qquad \emph{where }\alpha=\left\lbrace\begin{array}{ll}
\frac{1}{1+\beta^2_q},&\emph{if }\bcQ(\cdot) \emph{ is an unbiased scheme}\\
1,&\emph{if }\bcQ(\cdot) \emph{ is a biased scheme with $\beta^2_q\leq 1$.}
\end{array}
\right.
\end{equation}
\end{definition}

\begin{lemma}{\emph{\textbf{(Property of the top-$c$ quantizer)}.}} 
\label{lemma: Property of the top-$c$ quantizer}
The compression operator defined by~\eqref{eq: top-c quantizer} is a bounded-distortion compression operator with parameters:
\begin{equation}
\label{eq: bounded distortion compression operator parameters}
\beta^2_c=1-\frac{c}{L}\left(1-\left((1-\alpha)^2+\alpha^2\beta^2_q\right)\right)
\qquad \emph{and} \qquad\sigma^2_c=
\alpha^2\sigma^2_q,
\end{equation}
 where the expectation in~\eqref{eq: bounded quantization noise variance condition} is taken over the conditional probability measure~$\mathbb{P}(\cdot|x)$ that governs the random behavior of the compression operator $\bcQ(\cdot)$. By choosing $\alpha$ according to~\eqref{eq: top-c quantizer}, we find that the parameters $\{\beta^2_c,\sigma^2_c\}$ reduce to $\left\{1-\frac{c}{L(1+\beta^2_q)},\frac{\sigma^2_q}{(1+\beta^2_q)^2}\right\}$ for unbiased $\bcQ(\cdot) $ and to   $\left\{1-(1-\beta^2_q)\frac{c}{L},\sigma^2_q\right\}$ for biased $\bcQ(\cdot) $.
\end{lemma}
\begin{proof}
See Appendix~\ref{app: Property of the top-c quantizer}.
\end{proof}
\vspace{-2mm}
\hfill $\blacksquare$
\end{example}


\subsection{Contributions: Significant reduction in communication, with almost no effect on steady-state performance}
In summary, we provide the following main contributions.
\begin{itemize}
\item We propose a communication-efficient variant of the ATC {diffusion approach~\eqref{eq: decentralized learning approach}} for solving decentralized learning under subspace constraints. The strategy blends \emph{differential  quantization} and \emph{error feedback}.
\item We provide a detailed characterization of the proposed approach for a general class of bounded-distortion compression operators satisfying~\eqref{eq: bounded quantization noise variance condition}, both in terms of mean-square stability and communication resources.
\item \emph{In terms of steady-state performance:} {We show} that, in the small step-size regime, i.e., when $\mu \rightarrow 0$ (so that higher order terms of the step-size can be neglected), the iterates $\bw_{k,i}$ generated by the communication-efficient approach resulting from incorporating  differential quantization and error feedback into the ATC {diffusion approach~\eqref{eq: decentralized learning approach}} satisfy:
\begin{equation}
\label{eq: main result}
{\limsup_{i\rightarrow\infty}\expec\|w^o_k-\bw_{k,i}\|^2\approx\kappa\mu},
\end{equation}
where $\kappa$ is a constant that depends mainly on the gradient noise {(i.e., the difference between the true gradient and its approximation)} variance, and does not depend on the compression noise terms $\{\beta^2_c,\sigma^2_c\}$. The same result holds when studying the ATC {diffusion approach~\eqref{eq: decentralized learning approach}}  in the absence of compression~\cite[Theorem~1]{nassif2020adaptation}. {As we will show later, this asymptotic equivalence is achievable because the compression error is contained in higher-order terms that vanish faster than $\mu$ as $\mu\rightarrow 0$.}
\item  \emph{In terms of bit rate:} {While result~\eqref{eq: main result} provides important reassurance about the accuracy of the compressed approach, it does not address the communication efficiency directly, which is often quantified in terms of bit-rate.} In the absence of the absolute quantization noise term ($\sigma^2_c=0$), result~\eqref{eq: main result} is  achieved at the expense of communicating some quantities with high precision (as previously explained). In the presence of the absolute noise term, the analysis reveals that, to guarantee~\eqref{eq: main result}, {the parameters of the compression schemes (which are chosen by the designer) should be set so that the absolute noise term converges to zero as $\mu\rightarrow 0$.} 
We prove that this result can be achieved with a bit rate that remains bounded as $\mu\rightarrow 0$, despite the fact that we are requiring an increasing precision as the step-size decreases.
\end{itemize}

Thus, our theoretical findings  reveal that, in the small step-size regime, the proposed strategy attains the performance achievable in the absence of compression, despite the use of a finite number of bits. This demonstrates the effectiveness of the approach in maintaining performance while reducing communication overheads. While the theoretical findings show the optimality of the strategy in the small step-size regime, the experimental results in Sec.~\ref{sec: simulation results} illustrate its practical effectiveness in terms of achieving superior or competitive performance against state-of-the-art baselines in various scenarios, including those beyond the small step-size regime.


\section{Decentralized algorithmic framework: compressed communications}
\label{sec: Decentralized algorithmic framework: compressed communications}
In this work, we propose the DEF-ATC (differential error feedback - adapt then combine) diffusion strategy listed in Algorithm~\ref{alg: Quantized decentralized approach} {and in~\eqref{eq: step 1}--\eqref{eq: step 3}} for solving problem~\eqref{eq: network constrained problem} in a decentralized and communication-efficient manner. At each iteration $i$, each agent $k$ in the network performs three steps. The first step, which corresponds to the \emph{adaptation} step, is identical to the adaptation step~\eqref{eq: step1}, except that the step-size $\mu$ in~\eqref{eq: step1} is replaced by  $\mu/\zeta$ in~\eqref{eq: adaptation step}, {where $\zeta\in(0,1]$ is a \emph{damping} parameter appearing  in the compression step~\eqref{eq: reconstruction}. This parameter is used to {counteract} the instability induced by the compression errors. 
The second step is the \emph{compression} step. To update $\{\bphi_{\ell,i}\}_{\ell\in\cN_k}$, each agent $k$ first encodes the error compensated difference $\bpsi_{k,i}-\bphi_{k,i-1}+\bz_{k,i-1}$ (using {a bounded-distortion compression operator}\footnote{Since the compression scheme characteristics can vary {across agents}, the compression operator becomes $\bcC_k(\cdot)$ instead of $\bcC(\cdot)$ with a subscript $k$ added to $\bcC$.} $\bcC_k(\cdot)$), and broadcasts the result 
to its neighbors. Then, agent $k$ updates the \emph{local} compression error vector $\bz_{k,i}$ according to~\eqref{eq: quantization error 1} and  performs the reconstruction on each received vector by first decoding it to obtain $\bdelta_{\ell,i}$, and then computing the predictor $\bphi_{\ell,i}$ according to~\eqref{eq: reconstruction} {where $\bdelta_{\ell,i}$ is scaled by the aforementioned damping parameter $\zeta$}. Observe that implementing the compression step in Algorithm~\ref{alg: Quantized decentralized approach} requires storing the previous compression error $\bz_{k,i-1}$ and the previous predictors $\{\bphi_{\ell,i-1}\}_{\ell\in\cN_k}$ by agent $k$.  The compression step is followed by the \emph{combination} step~\eqref{eq: combination} where agent $k$ combines the reconstructed vectors $\{\bphi_{\ell,i}\}$ using the combination coefficients $\{A_{k\ell}\}$ and a \emph{mixing} parameter $\gamma\in(0,1]$. The resulting vector $\bw_{k,i}$ is the estimate of $w^o_k$, the $k$-th subvector of $\cw^o$ in~\eqref{eq: network constrained problem}, at agent~$k$ and iteration~$i$. As we will see in Sec.~\ref{sec: Stochastic performance analysis}, and as for the damping coefficient $\zeta$, the mixing parameter~$\gamma$ in the combination step~\eqref{eq: combination}  can also be used to control the algorithm stability.  A block diagram  illustrating the implementation of   the DEF-ATC {diffusion approach} at agent $k$ is provided in Fig.~\ref{fig: data settings}.

\begin{figure}
\begin{center}
\includegraphics[scale=0.45]{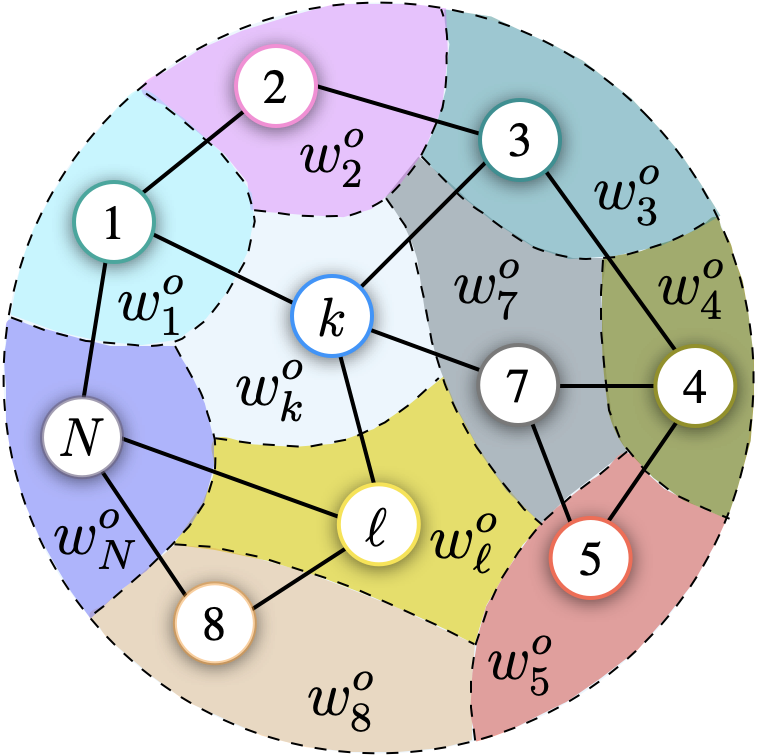}\qquad
\includegraphics[scale=0.26]{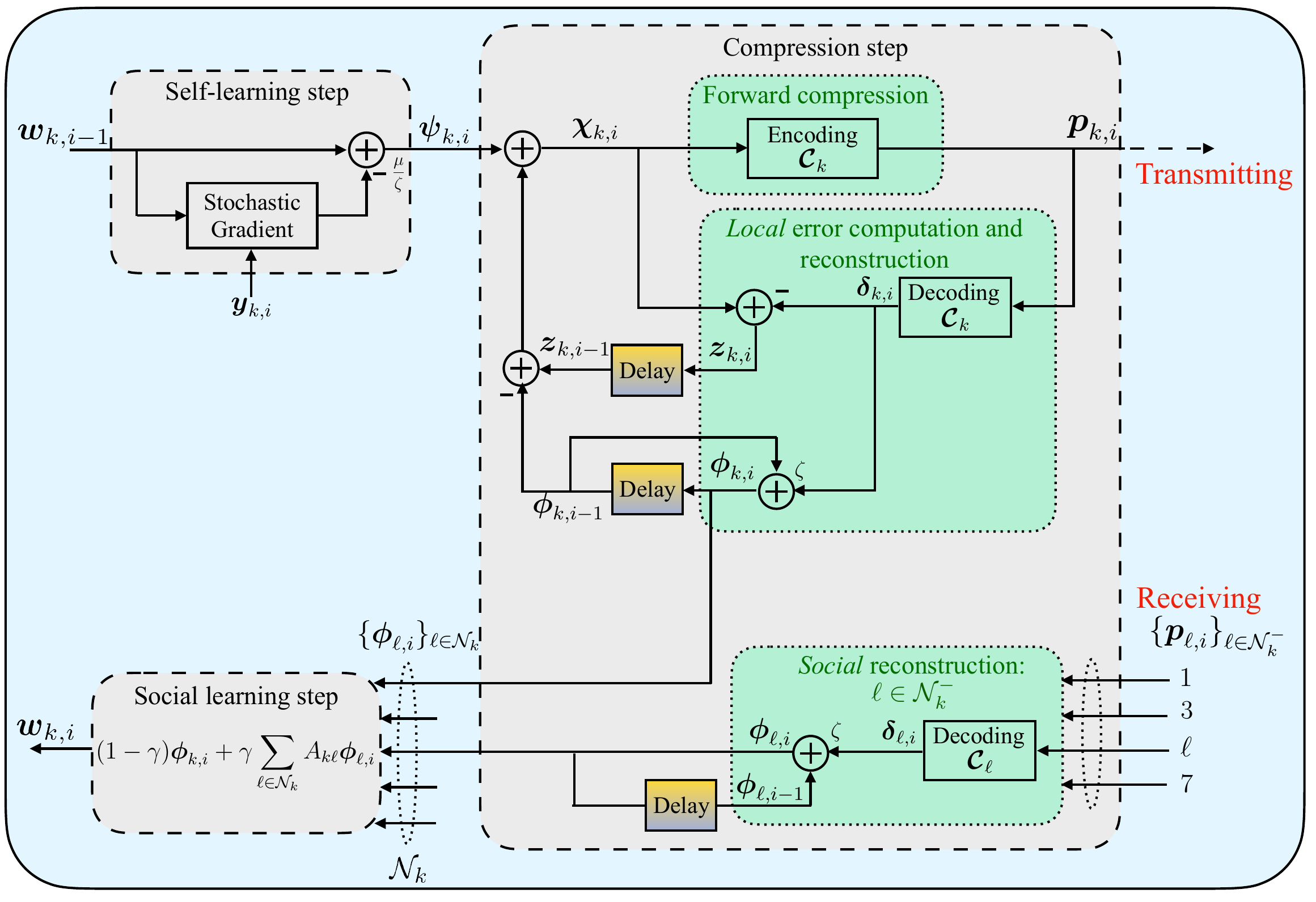}
\caption{\textit{(Left)} An illustration of a multitask network~\cite{nassif2020multitask,nassif2023quantization}. The objective at agent $k$ is to estimate $w^o_k$ (of dimension $M_k\times 1$), the $k$-th subvector of $\cw^o$ in~\eqref{eq: network constrained problem}. In this example, the neighborhood set of agent $k$ is given by $\cN_k=\{1,k,3,\ell,7\}$. \textit{(Right)} The implementation of the DEF-ATC {diffusion approach} listed in Alg.~\ref{alg: Quantized decentralized approach} at agent $k$. The set $\cN_k^{-}$ is the  neighborhood set of agent $k$, excluding $k$ itself. The compression step consists of three sub-steps: $i)$ the \emph{forward compression} step where agent $k$ encodes the error-compensated difference $\boldsymbol{\chi}_{k,i}=\bpsi_{k,i}-\bphi_{k,i-1}+\bz_{k,i-1}$ and sends the resulting vector ${\bp}_{k,i}$ (sequence of symbols or bits) to its neighbors; $ii)$  the \emph{local error computation and reconstruction} step where agent $k$ decodes the local vector  ${\bp}_{k,i}$ to obtain $\bdelta_{k,i}=\bcC_k(\bpsi_{k,i}-\bphi_{k,i-1}+\bz_{k,i-1})$, updates the local compression error vector $\bz_{k,i}$ according to~\eqref{eq: quantization error 1} and the local predictor $\bphi_{k,i}$ according to~\eqref{eq: step 2}; and $iii)$ the \emph{(social) reconstruction} step where agent $k$ receives the encoded vectors  $\{{\bp}_{\ell,i}\}_{\ell\in\cN_k^{-}}$  from its neighbors, decodes them to obtain $\{\bdelta_{\ell,i}\}_{\ell\in\cN_k^{-}}$, and then updates the predictors $\{\bphi_{\ell,i}\}_{\ell\in\cN_k^{-}}$ according to~\eqref{eq: reconstruction}. The resulting vectors {$\bphi_{k,i}$ and $\{\bphi_{\ell,i}\}_{\ell\in\cN_k^{-}}$} are then used in the social learning step~\eqref{eq: combination}. Observe that implementing the compression step requires storing the previous compression error $\bz_{k,i-1}$ and the previous estimates $\{\bphi_{\ell,i-1}\}_{\ell\in\cN_k}$ by agent $k$.}
\label{fig: data settings}
\end{center}
\end{figure}

\begin{algorithm}[t]
 \KwInput{initializations $\bw_{k,0}=0$, $\bphi_{k,0}=0$, and $\bz_{k,0}=0$, small step-size $\mu$, damping coefficient $\zeta\in(0,1]$, mixing parameter $\gamma\in(0,1]$, combination matrix $\cA$ satisfying~\eqref{eq: condition A}.}
\For{$i=1,2,\ldots,$ on the $k$-th node}
{   
\textbf{Adapt:} update $\bw_{k,i-1}$ according to: \begin{equation}\label{eq: adaptation step}
\bpsi_{k,i}=\bw_{k,i-1}-\frac{\mu}{\zeta}\widehat{\nabla_{w_k}J_k}(\bw_{k,i-1})
\end{equation}

\medskip

\textbf{Compress and broadcast:} \\
 $\bullet$ encode the error compensated difference $\bpsi_{k,i}-\bphi_{k,i-1}+\bz_{k,i-1}$ using {a bounded-distortion compression operator $\bcC_k(\cdot)$} and broadcast the result {$\bp_{k,i}$} to {the} neighbors $\cN_k$\\
  $\bullet$ upon receiving the compressed messages {$\{\bp_{\ell,i}\}$}  from neighbors $\ell\in\cN_k$, {first decode} them to obtain $\{\bdelta_{\ell,i}=\bcC_{\ell}(\bpsi_{\ell,i}-\bphi_{\ell,i-1}+\bz_{\ell,i-1})\}_{\ell\in\cN_k}$, and then {compute} $\{\bphi_{\ell,i}\}_{\ell\in\cN_k}$ according~to:
 \begin{equation}\label{eq: reconstruction}\bphi_{\ell,i}=\bphi_{\ell,i-1}+\zeta\bdelta_{\ell,i},\qquad \ell\in\cN_k\end{equation}
 $\bullet$  update the local compression error: \begin{equation}\label{eq: quantization error 1}{\bz_{k,i}=(\bpsi_{k,i}-\bphi_{k,i-1}+\bz_{k,i-1})-\bdelta_{k,i}}\end{equation}\\
  
\medskip

\textbf{Combine:} Update local model according to:\begin{equation}\label{eq: combination}\bw_{k,i}=(1-\gamma)\bphi_{k,i}+\gamma\sum_{\ell\in\cN_k}A_{k\ell}\bphi_{\ell,i}\end{equation}
}
 \caption{DEF-ATC (differential error feedback - adapt then combine) {diffusion strategy for solving~\eqref{eq: network constrained problem}}}
 \label{alg: Quantized decentralized approach}
\end{algorithm}

{For} the sake of convenience, we {rewrite} Algorithm~\ref{alg: Quantized decentralized approach} in the following compact form:
\begin{subequations}
\label{eq: compact algorithm form}
\begin{empheq}[left=\empheqlbrace]{align}
\bpsi_{k,i}&=\bw_{k,i-1}-\frac{\mu}{\zeta}\widehat{\nabla_{w_k}J_k}(\bw_{k,i-1})\label{eq: step 1}\\
\bphi_{k,i}&=\bphi_{k,i-1}+\zeta\bcC_{k}(\bpsi_{k,i}-\bphi_{k,i-1}+\bz_{k,i-1})\label{eq: step 2}\\
\bw_{k,i}&=(1-\gamma)\bphi_{k,i}+\gamma\sum_{\ell\in\cN_k}A_{k\ell}\bphi_{\ell,i}\label{eq: step 3}
\end{empheq}
\end{subequations}
where the compression error $\bz_{k,i}$ is updated according to:
\begin{equation}
\label{eq: quantization error}
\bz_{k,i}=(\bpsi_{k,i}-\bphi_{k,i-1}+\bz_{k,i-1})-\bcC_{k}(\bpsi_{k,i}-\bphi_{k,i-1}+\bz_{k,i-1}).
\end{equation}
 Observe that, in the absence of compression (i.e., when the operator $\bcC_k(\cdot)$ in~\eqref{eq: step 2} and~\eqref{eq: quantization error} is replaced by the identity operator and the parameters $\zeta$ and $\gamma$ in~\eqref{eq: step 2} and~\eqref{eq: step 3}, respectively, are set to 1) we {recover} the diffusion ATC approach~\eqref{eq: decentralized learning approach}. Therefore, Algorithm~\ref{alg: Quantized decentralized approach} can be seen as a communication-efficient variant of the Adapt-Then-Combine (ATC) approach. To mitigate the negative impact of compression, the {DEF-ATC approach} uses \emph{differential quantization} and \emph{error-feedback} in step~\eqref{eq: step 2}. Differential quantization consists of compressing  differences of the form $\bpsi_{k,i}-\bphi_{k,i-1}$ and  transmitting them,  instead of communicating compressed versions of the estimates $\bpsi_{k,i}$. Error feedback, on the other hand, consists  of locally storing the compression error $\bz_{k,i}$ (i.e., the difference between the input and output of the compression operator), and incorporating it back into the next iteration. {In Remark~1 further ahead, we explain the role of the compression error $\bz_{k,i}$ and how its introduction helps mitigate the accumulation of errors over time.}
 
\section{Mean-square-error and bit rate stability analysis}
\label{sec: Stochastic performance analysis}

\subsection{Modeling assumptions}
\label{subsec: Modeling assumptions}
In this section, we analyze strategy~\eqref{eq: compact algorithm form} with a matrix $\cA$ satisfying~\eqref{eq: condition A}  by examining the average squared distance between $\bw_{k,i}$ and $w^o_k$, namely, $\expec\|w^o_k-\bw_{k,i}\|^2$, under the following assumptions on the {risks} $\{J_k(\cdot)\}$, the gradient noise processes $\{\bs_{k,i}(\cdot)\}$ defined by~\cite{sayed2014adaptation}:
\begin{equation}
\label{eq: gradient noise process}
\bs_{k,i}(w)\triangleq \nabla_{w_k}J_k(w)-\widehat{\nabla_{w_k}J_k}(w),
\end{equation}
and the compression operators $\{\bcC_k(\cdot)\}$.

\begin{assumption}{\emph{\textbf{(Conditions on individual and aggregate {risks})}}.}
\label{assump: assumption of the individual costs}
The individual {risks} $J_k(w_k)$ are assumed to be twice differentiable {and convex} such that:
\begin{equation}
\label{eq: individual costs}
\lambda_{k,\min}I_{M_k}\leq \nabla^2_{w_k}J_k(w_k)\leq \lambda_{k,\max}I_{M_k},
\end{equation}
{where $\lambda_{k,\min}\geq 0$} for $k=1,\ldots,K$. It is further assumed that, for any $\{w_k\in\mathbb{R}^{M_k}\}$, the  individual {risks}  satisfy:
\begin{equation}
\label{eq: hessian}
0<\lambda_{\min} I_{P}\leq \cU^\top\emph{\diag}\left\{\nabla^2_{w_k}J_k(w_k)\right\}_{k=1}^K\cU\leq \lambda_{\max} I_{P},
\end{equation}
for some positive parameters $\lambda_{\min}\leq \lambda_{\max}$.\hfill\qed
\end{assumption}
\noindent As explained in~\cite{nassif2020adaptation}, condition~\eqref{eq: hessian} ensures that problem~\eqref{eq: network constrained problem} has a unique minimizer $\cw^o$.

\begin{assumption}{\emph{\textbf{(Conditions on gradient noise)}}.}
\label{assump: gradient noise}
The gradient noise process defined in~\eqref{eq: gradient noise process} satisfies for $k=1,\ldots,K$:
\begingroup
\allowdisplaybreaks\begin{align}
\expec\left[\bs_{k,i}(\bw_{k,i-1})|\{{\bphi_{\ell,i-1},\bz_{\ell,i-1}}\}_{\ell=1}^K\right]&=0,\label{eq: expectation of the gradient noise}\\
\expec\left[\|\bs_{k,i}(\bw_{k,i-1})\|^2|\{{\bphi_{\ell,i-1},\bz_{\ell,i-1}}\}_{\ell=1}^K\right]&\leq\beta^2_{s,k}\|{w^o_k-\bw_{k,i-1}}\|^2+\sigma^2_{s,k},\label{eq: expectation squared of the gradient noise}
\end{align}
\endgroup
for some $\beta^2_{s,k}\geq 0$ and $\sigma^2_{s,k}\geq 0$.\hfill\qed
\end{assumption}
\noindent As explained in~\cite{sayed2014adaptation,sayed2014adaptive,sayed2013diffusion}, these conditions are satisfied  by many  {risk}  functions of interest in learning and adaptation such as quadratic and  regularized logistic {costs}.  Condition~\eqref{eq: expectation of the gradient noise} states that the gradient vector approximation should be unbiased conditioned on the {iterates generated at the previous time instant}. Condition~\eqref{eq: expectation squared of the gradient noise} states that the second-order moment of the gradient noise should get smaller for better estimates, since it is bounded by the squared norm of the iterate.

\begin{assumption}{\emph{\textbf{(Conditions on compression operators)}.}}
\label{assump: quantization noise}
In step~\eqref{eq: step 2} of the {DEF-ATC strategy}, each agent $k$ at time $i$ applies to the error compensated difference $\bchi_{k,i}=\bpsi_{k,i}-\bphi_{k,i-1}+\bz_{k,i-1}$ a bounded-distortion compression operator $\bcC_k(\cdot)$ {(see Definition~\ref{def: bounded distortion})} with compression noise parameters $\beta^2_{c,k}$ and  $\sigma^2_{c,k}$. It is assumed that given the past history,   the randomized compression mechanism depends only on the quantizer input $\bchi_{k,i}$. Consequently, from~\eqref{eq: bounded quantization noise variance condition}, we get:
\begin{align}
\expec\left[\|\bchi_{k,i}-\bcC_k(\bchi_{k,i})\|^2|\boldh_i\right]&=\expec\left[\|\bchi_{k,i}-\bcC_k(\bchi_{k,i})\|^2|\bchi_{k,i}\right]\leq \beta^2_{c,k}\|\bchi_{k,i}\|^2+\sigma^2_{c,k},\label{eq: expectation squared of the quantization noise random}
\end{align}
where  $\boldh_i$ is the  vector collecting all iterates generated by~\eqref{eq: compact algorithm form} before the quantizer is applied to $\bchi_{k,i}$, namely, 
$\Big\{
\{\bphi_{\ell,j}\}_{j=1}^{i-1},\{\bpsi_{\ell,j}\}_{j=1}^{i},\{\bz_{\ell,j}\}_{j=1}^{i-1}\Big\}_{\ell=1}^K$
.\hfill\qed
\end{assumption}

\subsection{Network error vector recursion}
 In the following, we derive a useful recursion that allows to examine the time-evolution across the network of the error dynamics relative to the reference vector  $\cw^o=\col\{w^o_k\}_{k=1}^K$ defined in~\eqref{eq: network constrained problem}.  Let $\bwt_{k,i}=w^o_k-\bw_{k,i}$, $\bpsit_{k,i}=w^o_k-\bpsi_{k,i}$, and $\bphit_{k,i}=w^o_k-\bphi_{k,i}$. Using~\eqref{eq: gradient noise process}  and the mean-value theorem~\cite[pp.~24]{polyak1987introduction},~\cite[Appendix~D]{sayed2014adaptation}, we can express the stochastic gradient vector appearing in~\eqref{eq: step 1} as follows:
\begin{equation}
\label{eq: stochastic gradient vector-mean-value}
\widehat{\nabla_{w_k}J_k}(\bw_{k,i-1})=-\bH_{k,i-1}\bwt_{k,i-1}+b_k-\bs_{k,i}(\bw_{k,i-1}){,}
\end{equation}
where  $\bH_{k,i-1}\triangleq \int_{0}^1\nabla^2_{w_k}J_k(w^o_k-t\bwt_{k,i-1})dt$ and $b_k\triangleq\nabla_{w_k}J_k(w^o_k)$.  By subtracting $w^o_k$ from both sides of~\eqref{eq: step 1}, {by using~\eqref{eq: stochastic gradient vector-mean-value},} and by introducing the following network quantities:
\begingroup
\allowdisplaybreaks\begin{align}
b&\triangleq\col\left\{b_1,\ldots,b_K\right\},\label{eq: equation for b}\\
\bs_i&\triangleq\col\left\{\bs_{1,i}(\bw_{1,i-1}),\ldots,\bs_{K,i}(\bw_{K,i-1})\right\},\label{eq: equation for s}\\
\bcH_{i-1}&\triangleq\diag\left\{\bH_{1,i-1},\ldots,\bH_{K,i-1}\right\},\label{eq: equation for H}\\
\bcwt_{i-1}&\triangleq\col\{\bwt_{1,i-1},\ldots,\bwt_{K,i-1}\},\label{eq: collection of the network error vector}
\end{align}
\endgroup
we can show that the network error vector $\bpsit_i=\col\{\bpsit_{k,i}\}_{k=1}^K$ evolves according to:
\begin{equation}
\bpsit_i=\left(I_M-\frac{\mu}{\zeta}\bcH_{i-1}\right)\bcwt_{i-1}-\frac{\mu}{\zeta}\bs_i+\frac{\mu}{\zeta} b.\label{eq: network error vector psi}
\end{equation}
By subtracting $w^o_k$  from both sides of~\eqref{eq: step 3}, by replacing $w^o_k$ by $(1-\gamma) w^o_k+\gamma w^o_k$, and by using  $w^o_k=\sum_{\ell\in\cN_k}A_{k\ell}w^o_{\ell}$~\cite[Sec. III-B]{nassif2020adaptation}, we obtain:
\begin{equation}
\label{eq: evolution of wtk i}
\begin{split}
\bwt_{k,i}&=(1-\gamma) \bphit_{k,i}+\gamma\sum_{\ell\in\cN_k}A_{k\ell}\bphit_{\ell,i}.
\end{split}
\end{equation}
From~\eqref{eq: evolution of wtk i}, we {find} {that the} network error vector $\bcwt_{i-1}$ in~\eqref{eq: collection of the network error vector} evolves according to:
\begin{equation}
\bcwt_{i-1}=(1-\gamma)\bphit_{i-1}+\gamma\cA\bphit_{i-1}=\cA'\bphit_{i-1},\label{eq: wt in terms of phi}
\end{equation} 
where 
\begin{align}
\bphit_{i}&\triangleq\col\{\bphit_{1,i},\ldots,\bphit_{K,i}\}{,}\label{eq: collection of phi tilde}\\
\cA'&\triangleq(1-\gamma) I_M+\gamma\cA{.}\label{eq: equation for A'}
\end{align}  
By  subtracting $w^o_k$ from both sides of~\eqref{eq: step 2} and by adding and subtracting $w^o_k$ to the difference $\bpsi_{k,i}-\bphi_{k,i-1}$, we can write:
\begin{equation}
\label{eq: bphit_k}
\bphit_{k,i}=\bphit_{k,i-1}- \zeta\bcC_k(\bphit_{k,i-1}-\bpsit_{k,i}+\bz_{k,i-1}).
\end{equation}
 Now, by adding and subtracting  $\zeta(\bphit_{k,i-1}-\bpsit_{k,i}+\bz_{k,i-1})$ to the RHS of the above equation, we obtain:
\begin{equation}
\label{eq: bphit_k in terms of the quantization noise}
\bphit_{k,i}=(1-\zeta)\bphit_{k,i-1}+\zeta\bpsit_{k,i}+\zeta(\bz_{k,i}-\bz_{k,i-1}),
\end{equation}
in terms of the compression error vector $\bz_{k,i}$ defined in~\eqref{eq: quantization error}. By combining~\eqref{eq: network error vector psi},~\eqref{eq: wt in terms of phi}, and~\eqref{eq: bphit_k in terms of the quantization noise}, we  conclude that the network  error vector $\bphit_i$ in~\eqref{eq: collection of phi tilde} evolves according to the following dynamics
:
\begin{equation}
\label{eq: network weight error vector}
\bphit_i=\bcB_{i-1}\bphit_{i-1}-\mu\bs_i+\mu b+\left(\bz_i-\bz_{i-1}\right),
\end{equation} 
where
\begin{align}
\bcB_{i-1}&\triangleq (1-\zeta)I_M+\zeta\left(I_M-\frac{\mu}{\zeta}\bcH_{i-1}\right)\cA',\label{eq: equation for B}\\
\bz_i&\triangleq\zeta\col\left\{ \bz_{k,i}\right\}_{k=1}^K.\label{eq: equation for z}
\end{align}

\noindent\textbf{Remark 1 (Temporal filtering of the compression error):} A direct consequence of feeding back the error in the compression step~\eqref{eq: step 2} is {to subtract} the compression error from previous instants in recursion~\eqref{eq: network weight error vector}, thereby allowing for a correction mechanism\footnote{To see this, we can simply remove the error feedback mechanism from the approach~\eqref{eq: compact algorithm form} by replacing the compression step~\eqref{eq: step 2} by $\bphi_{k,i}=\bphi_{k,i-1}+\zeta\bcC_{k}(\bpsi_{k,i}-\bphi_{k,i-1})$ and derive the network  error vector $\bphit_i$ in~\eqref{eq: collection of phi tilde} by following similar arguments as in~\eqref{eq: stochastic gradient vector-mean-value}--\eqref{eq: equation for z}. Instead of~\eqref{eq: network weight error vector}, we {would arrive at the} following dynamics:
\begin{equation}
\label{eq: network weight error vector without error feedback}
\bphit_i=\bcB_{i-1}\bphit_{i-1}-\mu\bs_i+\mu b+\bz_i,
\end{equation}
where {we obtain in~\eqref{eq: network weight error vector without error feedback} the instantaneous noise vector $\bz_i$ instead of the difference vector $\bz_i-\bz_{i-1}$ as in~\eqref{eq: network weight error vector}.}}. This correction helps {mitigate} the accumulation of errors over time, leading to  improved network performance. \hfill $\blacksquare$

As the presentation will reveal, the study of the network behavior in the presence of error feedback is a challenging task since, in addition to analyzing the dynamics of the network error vector $\bphit_i$, we need to examine how the compression error~\eqref{eq: equation for z}, which is fed back into the network through the compression step~\eqref{eq: step 2}, affects its behavior. When all is said and done, the results will help clarify the effect of network topology, step-size $\mu$, damping coefficient~$\zeta$, mixing parameter $\gamma$, gradient (through  $\{\beta^2_{s,k},\sigma^2_{s,k}\}$) and compression   (through  $\{\beta^2_{c,k},\sigma^2_{c,k}\}$) noise processes on the network mean-square-error stability and performance, and will provide insights into the design of effective compression operators for decentralized learning. 

\subsection{Mean-square-error stability}
\label{subsec: Mean-square-error stability}
The mean-square-error analysis {will be carried out} by first establishing the boundedness of $\limsup_{i\rightarrow\infty}\expec\|\bphit_i-\bz_i\|^2$, and then using {relation~\eqref{eq: wt in terms of phi} and Holder's and  Jensen's inequalities} to deduce boundedness of $\limsup_{i\rightarrow\infty}\expec\|\bcwt_i\|^2$. Therefore, in the following, we first study the stability of the network error vector $\bphit_i^z$ defined as $\bphit_i^z\triangleq\bphit_i-\bz_i$ and evolving according to:
\begin{equation}
\label{eq: network weight error vector 2}
\boxed{\bphit_i^z=\bcB_{i-1}\bphit^z_{i-1}-\mu\bs_i+\mu b-(I_M-\bcB_{i-1})\bz_{i-1}}
\end{equation} 
The above identity can be found by adding and subtracting the term $\bcB_{i-1}\bz_{i-1}$ to the RHS of~\eqref{eq: network weight error vector}. {We analyze the stability of recursion~\eqref{eq: network weight error vector 2} by first transforming it into a more convenient form (shown later in~\eqref{eq: evolution of the centralized recursion 1} and~\eqref{eq: evolution of the centralized recursion 3})} using the Jordan canonical {decomposition~\cite{horn2012matrix}} of the matrix $\cA'$ defined in~\eqref{eq: equation for A'}. To exploit the eigen-structure of  $\cA'$, we first recall that a matrix $\cA$ satisfying the conditions in~\eqref{eq: condition A} (for a full-column rank semi-unitary matrix $\cU$) has a Jordan decomposition of the form $\cA=\cV_{\epsilon}\Lambda_{\epsilon}\cV_{\epsilon}^{-1}$ with~\cite[Lemma~2]{nassif2020adaptation}:
\begin{equation}
\label{eq: jordan decomposition of A}
\cV_{\epsilon}=\left[\begin{array}{c|c}
\cU&\cV_{R,\epsilon}
\end{array}\right],~
\Lambda_{\epsilon}=\left[\begin{array}{c|c}
I_P&0\\
\hline
0&\cJ_{\epsilon}
\end{array}\right],~
\cV_{\epsilon}^{-1}=\left[\begin{array}{c}
\cU^\top\\
\hline
\cV_{L,\epsilon}^{\top}
\end{array}\right],
\end{equation} 
where $\cJ_{\epsilon}$ is a Jordan matrix with  eigenvalues (which may be complex but have magnitude less than one) on the diagonal and $\epsilon>0$  on the super-diagonal~\cite[Lemma~2]{nassif2020adaptation},\cite[pp.~510]{sayed2014adaptation}. The parameter $\epsilon$ is chosen small enough to ensure  $\rho(\cJ_{\epsilon})+\epsilon\in(0,1)$~\cite{nassif2020adaptation}. Consequently, the matrix $\cA'$ in~\eqref{eq: equation for A'} has a Jordan decomposition of the form $\cA'=\cV_{\epsilon}\Lambda'_{\epsilon}\cV_{\epsilon}^{-1}$ where:
\begin{equation}
\label{eq: jordan decomposition of A'}
\Lambda'_{\epsilon}=\left[\begin{array}{c|c}
I_P&0\\
\hline
0&\cJ'_{\epsilon}
\end{array}\right], \quad{\text{with }}\cJ'_{\epsilon}\triangleq (1-\gamma) I_{M-P}+\gamma\cJ_{\epsilon}.
\end{equation} 

By multiplying both sides of~\eqref{eq: network weight error vector 2} from the left by $\cV_{\epsilon}^{-1}$ in~\eqref{eq: jordan decomposition of A},  we obtain the transformed iterates and variables:
\begin{align}
\cV_{\epsilon}^{-1}\bphit_i^z&=\left[
\begin{array}{c}
\cU^\top\bphit^z_i\\
\cV_{L,\epsilon}^{\top}\bphit^z_i
\end{array}
\right]\triangleq\left[
\begin{array}{c}
\bphib^z_i\\
\bphic^z_i
\end{array}
\right],\label{eq: transformed variable phi definition}\\
\cV_{\epsilon}^{-1}\bs_i&=\left[
\begin{array}{c}
\cU^\top\bs_i\\
\cV_{L,\epsilon}^{\top}\bs_i
\end{array}
\right]\triangleq\left[
\begin{array}{c}
\bsb_i\\
\bsc_i
\end{array}
\right],\label{eq: transformed variable s definition}\\
\cV_{\epsilon}^{-1}b&=\left[
\begin{array}{c}
\cU^\top b\\
\cV_{L,\epsilon}^{\top}b
\end{array}
\right]\triangleq\left[
\begin{array}{c}
0\\
\widecheck{b}
\end{array}
\right],\label{eq: bias transformed vector}\\
\cV_{\epsilon}^{-1}\bz_{i-1}&=\left[
\begin{array}{c}
\cU^\top \bz_{i-1}\\
\cV_{L,\epsilon}^{\top}\bz_{i-1}
\end{array}
\right]\triangleq\left[
\begin{array}{c}
\bzb_{i-1}\\
\bzc_{i-1}
\end{array}
\right],\label{eq: quantization transformed vector}
\end{align}
where  in~\eqref{eq: bias transformed vector}  we used the fact that $\cU^\top b=0$ as shown in~\cite[Sec. III-B]{nassif2020adaptation}. In particular, the transformed components $\bphib^z_{i}$ and ${\bphic}^z_{i}$ evolve according to the recursions:
\begin{align}
\bphib^z_{i}&=(I_P-\mu\bcD_{11,i-1})\bphib^z_{i-1}-\mu\bcD_{12,i-1}\bphic^z_{i-1}-\mu\bsb_{i}-\mu\bcD_{11,i-1}\bzb_{i-1}-\mu\bcD_{12,i-1}\bzc_{i-1}\label{eq: evolution of the centralized recursion 1}\\
{\bphic}^z_{i}&=(\cJ''_{\epsilon}-\mu\bcD_{22,i-1})\bphic^z_{i-1}-\mu\bcD_{21,i-1}\bphib^z_{i-1}+\mu\widecheck{b}-\mu\bsc_{i}-\mu\bcD_{21,i-1}\bzb_{i-1}-\left(\zeta(I-\cJ_{\epsilon}')+\mu\bcD_{22,i-1}\right)\bzc_{i-1}\label{eq: evolution of the centralized recursion 3}
\end{align}
where 
\begin{align}
\bcD_{11,i-1}&\triangleq\cU^{\top}\bcH_{i-1}\cU,\label{eq: definition of D11 1}\\
 \bcD_{12,i-1}&\triangleq\cU^{\top}\bcH_{i-1}\cV_{R,\epsilon}\cJ_{\epsilon}',\\
\bcD_{21,i-1}&\triangleq \cV_{L,\epsilon}^{\top}\bcH_{i-1}\cU,\\
\bcD_{22,i-1}&\triangleq\cV_{L,\epsilon}^{\top}\bcH_{i-1}\cV_{R,\epsilon}\cJ_{\epsilon}',\label{eq: definition of D22 1}\\
\cJ''_{\epsilon}&\triangleq(1-\zeta)I_{M-P}+\zeta\cJ'_{\epsilon}.\label{eq: definition of J'' epsilon}
\end{align}

\begin{theorem}{\emph{\textbf{(Mean-square-error stability)}.}} 
\label{theorem: Network mean-square-error stability}Consider a network of $K$ agents running the \emph{DEF-ATC} {diffusion approach} (listed in Algorithm~\ref{alg: Quantized decentralized approach}) {to solve problem~\eqref{eq: network constrained problem}} 
under Assumptions~\ref{assump: assumption of the individual costs},~\ref{assump: gradient noise}, and~\ref{assump: quantization noise}, with a matrix $\cA$ satisfying~\eqref{eq: condition A}. 
In the absence of the relative compression noise term (i.e., $\beta^2_{c,k}=0$, $\forall k$), let the damping and mixing parameters be such that $\zeta=\gamma=1$. In the presence of the relative compression noise {(i.e., at least one $\beta^2_{c,k}$ is positive for some agent $k$)}, let $\zeta\in(0,1]$ and $\gamma\in(0,1]$ be such that the two following conditions are satisfied:
\begin{equation}
\label{eq: mixing parameter condition theorem}
{0<\gamma\zeta<\frac{1-(\rho(\cJ_{\epsilon})+\epsilon)}{4v_1^2v_2^2\beta^2_{c,\max}(\rho(I-\cJ_{\epsilon})+\epsilon)^2},}
\end{equation}
and
\begin{equation}
\label{eq: mixing parameter condition theorem 2}
\begin{split}
{\gamma\zeta\frac{(\rho(I-\cJ_{\epsilon})+\epsilon)^2}{1-(\rho(\cJ_{\epsilon})+\epsilon)}  +\zeta^2\beta_{c,\max}^2v_1^2v_2^2\left(1+\left((1+\gamma)-\gamma(\rho(\cJ_{\epsilon})+\epsilon)\right)^2\right)<\frac{1}{2},}
\end{split}
\end{equation}
where $v_1=\|\cV_{\epsilon}^{-1}\|$, $v_2=\|\cV_{\epsilon}\|$, and ${\beta}_{c,\max}^2\triangleq\max_{1\leq k\leq K}\{{\beta}^2_{c,k}\}$. Then, {for sufficiently small step-size~$\mu$, the network is mean-square-error stable, and it holds that:}
\begin{equation}
\limsup_{i\rightarrow\infty}\expec\|\bphit^z_i\|^2=\kappa \mu+\overline{\sigma}^2_c O(1),\label{eq: steady state mean square error}
\end{equation}
where $\overline{\sigma}^2_c=\sum_{k=1}^K{\sigma}^2_{c,k}$. The constant $\kappa$ is positive, independent of the step-size $\mu$ and the compression noise terms $\{\beta^2_{c,k},\sigma^2_{c,k}\}$, and is given by $\kappa=v_1^2v_2^2\frac{\overline{\sigma}^2_s}{\sigma_{11}}$ with $\overline{\sigma}^2_s=\sum_{k=1}^K\sigma_{s,k}^2$, and $\sigma_{11}$ {is} some positive constant {resulting from the {derivation} of inequality~\eqref{eq: centralized error vector recursion inequality biased} in Appendix~\ref{app: Mean-square-error analysis}}. 
Moreover, by choosing compression schemes with $\sigma^2_{c,k}\propto \mu^{1+\varepsilon}$ (where the symbol $\propto$ hides a proportionality constant independent of~$\mu$) and $\varepsilon\in(0,1]$, {we obtain}:
\begin{equation}
{\limsup_{i\rightarrow\infty}\expec\|\bphit^z_i\|^2=\kappa \mu+O(\mu^{1+{\varepsilon}}).}\label{eq: steady state mean square error phitilde}
\end{equation} 
{It then holds that:}
\begin{equation}
\limsup_{i\rightarrow\infty}\expec\|\bcwt_i\|^2=\kappa\mu+O(\mu^{1+\frac{\varepsilon}{2}}),\label{eq: steady state mean square error w}
\end{equation} 
from which we  conclude that:
\begin{equation}
\label{eq: steady state mean square error w 1}
{\lim_{\mu\rightarrow 0}\limsup_{i\rightarrow\infty}\frac{1}{\mu}\expec\|w^o_k-\bw_{k,i}\|^2}=\kappa,
\end{equation} 
for $k=1,\ldots,K$.
\end{theorem}
\begin{proof} See Appendix~\ref{app: Mean-square-error analysis}.
\end{proof}

While expressions~\eqref{eq: steady state mean square error}--\eqref{eq: steady state mean square error w 1} in Theorem~\ref{theorem: Network mean-square-error stability} {reveal} the influence of the \emph{step-size} $\mu$, the \emph{compression noise} (captured by $\{\overline{\sigma}^2_{c},\beta^2_{c,\max}\}$), and the \emph{gradient noise} (captured by $\overline{\sigma}^2_s$) on the steady-state mean-square error, expressions~\eqref{eq: mixing parameter condition theorem} and~\eqref{eq: mixing parameter condition theorem 2} reveal the influence of the \emph{relative compression noise} term (captured by $\beta^2_{c,\max}$) on the network stability, and how this influence can be mitigated by properly choosing the damping coefficient $\zeta$ and the mixing parameter~$\gamma$.  One main conclusion stemming from Theorem~\ref{theorem: Network mean-square-error stability} (expression~\eqref{eq: steady state mean square error}) is that the mean-square-error contains two terms. {The first term is $\kappa \mu$ where $\kappa$ is a constant independent of the compression noise $\{\beta^2_{c,k},\sigma^2_{c,k}\}$, but depends on the gradient noise $\{\sigma^2_{s,k}\}$. This term, which we refer to as the \emph{gradient noise term}, is classically encountered in the uncompressed case~\cite{nassif2020adaptation}. The second {factor}  is an $O(1)$ term that is proportional to the quantizers' absolute noise components $\{\sigma^2_{c,k}\}$.} 
 \emph{Interestingly}, by choosing compression schemes with $\sigma^2_{c,k}\propto \mu^{1+\varepsilon}$ and $\varepsilon\in(0,1]$, {for sufficiently small step-sizes $\mu$ we} obtain {$\limsup_{i\rightarrow\infty}\expec\|w^o_k-\bw_{k,i}\|^2\approx\kappa \mu$}. This result is reassuring since it implies that  the impact of the compression noise can be minimized to the point where it only affects higher-order terms of the step-size. Consequently, the primary noise influencing the learning process will be the gradient noise, which is consistent with the classical {results} observed in the uncompressed case studied in~\cite{nassif2020adaptation}.

{While result~\eqref{eq: steady state mean square error w 1} is appealing, it is not sufficient to characterize the DEF-ATC {diffusion approach}. To fully characterize a decentralized strategy endowed with a compression mechanism, it is essential to consider the \emph{learning-communication tradeoff}. In other words, we need {to assess also} how the design choice $\sigma^2_{c,k}\propto \mu^{1+\varepsilon}$ impacts the amount of communication {resources} (e.g., quantization bits). For instance, consider} 
the probabilistic uniform quantizer from Table~\ref{table: examples of quantizers}. For this scheme, setting $\sigma^2_{c,k}\propto \mu^{1+\varepsilon}$ is equivalent to requiring the quantization {step}  $\Delta$ to be proportional to $ \mu^{\frac{1+\varepsilon}{2}}$. Thus, while small values of $\sigma^2_{c,k}$  imply small compression {errors in view of~\eqref{eq: bounded quantization noise variance condition}}, they  might in principle require large bit rates. Moreover, as $\mu\rightarrow0$, the quantization {step  $\Delta$}  becomes very small, potentially leading to an unbounded bit rate increase. It becomes therefore important to find a quantization scheme that achieves the same performance as the uncompressed case, i.e., {$\limsup_{i\rightarrow\infty}\expec\|w^o_k-\bw_{k,i}\|^2\approx\kappa \mu$}, while guaranteeing a finite   bit rate  as $\mu\rightarrow 0$. {In the next theorem, we show that the DEF-ATC {diffusion approach} equipped with the variable-rate coding scheme from~\cite{lee2021finite},~\cite[Sec.~II]{nassif2023quantization} achieves both objectives.}
\subsection{Bit rate stability}
\label{subsec: bit rate}
We first assume that the top-$c_k$ quantizer  in Definition~\ref{def: definition of the top-c quantizer} (with a subscript $k$ added to $c$ to highlight the fact that the compression characteristics can vary {across agents}) is used at each iteration $i$ and agent $k$. We then recall that the quantizer input  is given by the error compensated difference $\bchi_{k,i}=\bpsi_{k,i}-\bphi_{k,i-1}+\bz_{k,i-1}$, and assume that the probabilistic ANQ scheme\footnote{The probabilistic uniform rule (Table~\ref{table: examples of quantizers}, row~$2$) can be obtained from the ANQ rule by letting $\omega\rightarrow 0$~\cite{nassif2023quantization}.} (Table~\ref{table: examples of quantizers}, row~$3$) is employed at the output of the  top-$c_k$ sparsifier. Consequently, from~\eqref{eq: bit rate general formula}, the bit rate at agent $k$ and iteration $i$ is given by:
\begin{equation}
\label{eq: bit rate general at agent k}
{{r}_{k,i}}=\log_2(3)\sum_{j\in\cI_{c_{k,i}}}\Big(1+\expec\Big[\left\lceil\log_2(|\boldsymbol{n}([\bchi_{k,i}]_j)|+1)\right\rceil\Big]\Big)+(M_k-c_k)\log_2(3),
\end{equation}
where $[\bchi_{k,i}]_j$ denotes the $j$-th entry of $\bchi_{k,i}$, $\cI_{c_{k,i}}$ is the set containing the indices of the $c_k$ largest-magnitude components of $\bchi_{k,i}$, and $(M_k-c_k)\log_2(3)$ is the number of bits required to encode the $M_k-c_k$ zero components at the output of the top-$c_k$ sparsifier\footnote{Encoding the $0$ input using the variable-rate coding scheme from~\cite{lee2021finite},~\cite[Sec.~II]{nassif2023quantization} requires {only a} parsing symbol
, i.e., $\log_2(3)$ bits. {Instead of encoding the zero components, agent $k$ can send the location of the $c_k$ largest-magnitude components. In this case, the term $(M_k-c_k)\log_2(3)$ is replaced by $c_k\lceil \log_2(M_k)\rceil$. This alternative solution does not affect the main conclusions of Sec.~\ref{subsec: bit rate}.}}. 
\begin{theorem}{\emph{\textbf{(Bit rate stability)}.}} 
\label{theorem: rate stability} Assume that each agent $k$ employs the top-$c_k$ quantizer (see Definition~\ref{def: definition of the top-c quantizer}) with the probabilistic ANQ scheme. Assume further that the design parameters of the compression operators are chosen such that:
\begin{equation}
\label{eq: condition of the theorem 2}
\omega_k=t, \qquad \eta_k\propto \mu^{\frac{1+\varepsilon}{2}},
\end{equation}
where $t$ is a constant independent of $\mu$ and $0<\varepsilon\leq 1$. First, under conditions~\eqref{eq: condition of the theorem 2}, we have:
\begin{equation}
\label{eq: sigma square}
\sigma^2_{c,k}\propto \mu^{1+\varepsilon}.
\end{equation}
Second, {in the steady state,} the average number of bits at agent $k$ stays bounded as $\mu\rightarrow 0$, namely, 
\begin{equation}
\label{eq: rate stability result}
\limsup_{i\rightarrow\infty}\,{r}_{k,i}=O(1).
\end{equation}
\end{theorem}
\begin{proof}
See Appendix~\ref{app: rate stability}.
\end{proof}

The bit rate stability result~\eqref{eq: rate stability result} can be explained by considering again the uniform quantization rule (Table~\ref{table: examples of quantizers}, row~$2$) which, as explained previously, requires setting $\Delta\propto \mu^{\frac{1+\varepsilon}{2}}$ in order to guarantee {that} $\sigma^2_{c,k}\propto \mu^{1+\varepsilon}$. The result~\eqref{eq: appendix c} reveals {that the input to the compression operators of the DEF-ATC strategy in~\eqref{eq: step 2}, namely, the}  \emph{error compensated difference} $\bchi_{k,i}=\bpsi_{k,i}-\bphi_{k,i-1}+\bz_{k,i-1}${, under~\eqref{eq: sigma square}}  is on the order of $\mu^{\frac{1+\varepsilon}{2}}$ at steady-state. This means that as  {$\mu\rightarrow 0$}, the quantizer resolution $\Delta\propto \mu^{\frac{1+\varepsilon}{2}}$ decreases, but in proportion to the  \emph{effective} range of the quantizers' inputs. Theorem~\ref{theorem: rate stability}  reveals the adaptability of  the variable-rate scheme, ensuring that even as the quantization becomes increasingly precise (as $\mu\rightarrow 0$), the {DEF-ATC strategy} can still maintain a finite expected bit rate, which is crucial for efficient data transmission.
\section{Simulation results}
\label{sec: simulation results}
In this section, we first illustrate the theoretical results of Theorems~\ref{theorem: Network mean-square-error stability} and~\ref{theorem: rate stability}. Then, we illustrate the benefit of the top-$c$ quantizer over other quantizers, particularly those that quantize a vector element-wise, without prioritizing the $c$ most important components.   In the third part, we compare DEF-ATC to state-of-the-art baselines in various scenarios, including those beyond the small step-size regime. {The first three parts} focus on solving single-task optimization problems of the form~\eqref{eq: consensus formulation}. {The} last part illustrates the performance of the DEF-ATC approach when used to solve multitask estimation problems with overlapping parameter vectors.
\begin{figure}
\begin{center}
\includegraphics[scale=0.22]{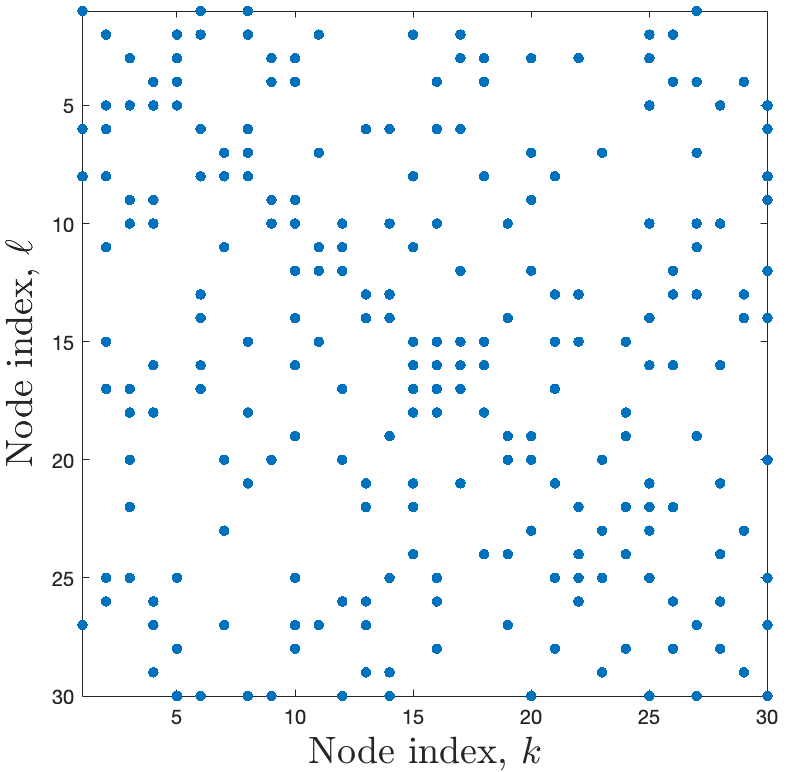}\qquad
\includegraphics[scale=0.22]{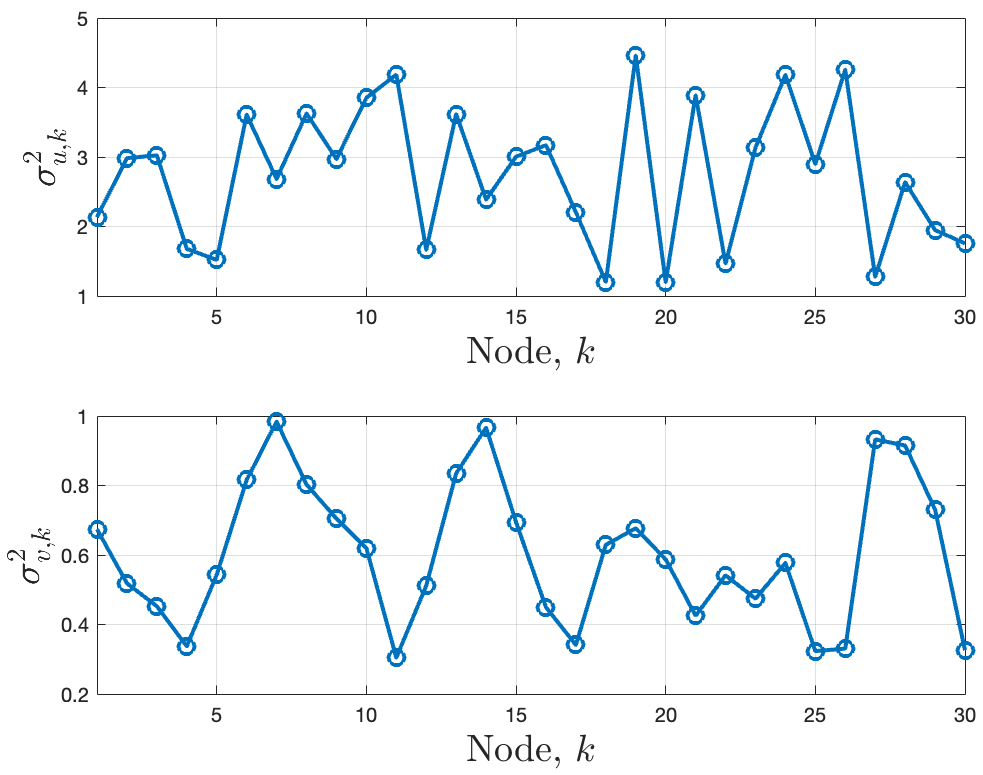}
\caption{Experimental setup. \emph{(Left)} Communication link matrix.  \emph{(Right)} Regression and noise variances.
}
\label{fig: network settings}
\end{center}
\end{figure}

We consider a network of $K=30$ nodes  with the communication link matrix shown in Fig.~\ref{fig: network settings} \emph{(left)}, where the $(k,\ell)$-th entry is equal to $1$ if there is a link between $k$ and $\ell$ and is 0 otherwise. 
Each agent is subjected to streaming data $\{\bd_k(i),\bu_{k,i}\}$ assumed to satisfy a linear regression model of the form $\bd_k(i)=\bu_{k,i}^\top w^\star_k+\bv_k(i)$ for some ${{M_c}}\times 1$ vector $w^\star_k$ with $\bv_k(i)$ denoting a zero-mean measurement noise and ${{M_c}}=10$. A mean-square-error {risk} of the form $J_k(w_k)=\expec|\bd_k(i)-\bu_{k,i}^{\top}w_k|^2$ is associated with each agent $k$. The processes $\{\bu_{k,i},\bv_k(i)\}$ are assumed to be zero-mean Gaussian with: $i)$ $\expec\bu_{k,i}\bu_{\ell,i}^\top=R_{u,k}=\sigma^2_{u,k}I_{{M_c}}$ if $k=\ell$ and $0$ otherwise; $ii)$ $\expec\bv_{k}(i)\bv_{\ell}(i)=\sigma^2_{v,k}$ if $k=\ell$ and $0$ otherwise; and $iii)$  $\bu_{k,i}$ and $\bv_{k}(i)$ are independent of each other. The variances  $\sigma^2_{u,k}$ and $\sigma^2_{v,k}$ are shown in Fig.~\ref{fig: network settings} \emph{(right)}. 
{Throughout Sec.~\ref{sec: simulation results}, we assume that all agents employ the same compression rule, i.e., $\bcC_k=\bcC$ $\forall k$. We use the terminology ``top-$c$ \textsf{{\footnotesize{quantizer-name}}}" to refer to the top-$c$ quantizer of Definition~\ref{def: definition of the top-c quantizer} where,  as compression scheme $\bcQ$, we use \textsf{{\footnotesize{quantizer-name}}}.  For instance, ``top-$4$ probabilistic ANQ" is the quantizer obtained by applying the probabilistic ANQ scheme at the output of the top-$4$ sparsifier.} 

\subsection{Illustrating the theoretical findings}
\label{subsec: Illustrating the theoretical findings}
{In this {section and the following Secs.~\ref{subsec: Top-$k$ with quantization outperforms other compression rules} and~\ref{subsec: Performance w.r.t. state-of-the-art baselines}}, we assume that agents have a common model parameter   $w^\star_k=w^o$ $\forall k$.}  The model $w^o$ is generated by normalizing to unit norm a randomly generated Gaussian vector, with zero mean and unit variance. To promote consensus (i.e., {to} solve problem~\eqref{eq: consensus formulation} or, equivalently,~\eqref{eq: network constrained problem} with $\cU=\frac{1}{\sqrt{K}}(\mathds{1}_K\otimes I_{{M_c}})$), we run Alg.~\ref{alg: Quantized decentralized approach} using a combination matrix of the form $\cA=A\otimes I_{{M_c}}$, where $A$ is generated according to the Metropolis rule~\cite[Chap. 8]{sayed2014adaptation}. 

In Fig.~\ref{fig: variable step-size} \emph{(left)}, we report the network mean-square-deviation (MSD) learning curves:
\begin{equation}
\label{eq: definition of instanteneous MSD}
\text{MSD}(i)=\frac{1}{K}\sum_{k=1}^K\expec\|w^o_k-\bw_{k,i}\|^2,
\end{equation}
for $3$ different values of the step-size $\mu$.  The results are averaged over $100$ Monte-Carlo runs.  For each value of the step-size, we run  Alg.~\ref{alg: Quantized decentralized approach}  for $4$ different choices of the compression operator $\bcC$: $i)$ top-$4$ sparsifier,  $ii)$ top-$4$ QSGD  (Table~\ref{table: examples of quantizers}, row $6$, $s=2$), $iii)$ top-$4$ probabilistic uniform (Table~\ref{table: examples of quantizers}, row $2$, $\Delta=\mu$), and $iv)$  top-$4$ probabilistic ANQ (Table~\ref{table: examples of quantizers}, row $3$, $\omega=0.5$, $\eta=\mu$). We set $\gamma=\zeta=0.9$. As it can be observed, despite compression, the DEF-ATC approach achieves a performance that is almost identical to the uncompressed ATC approach (which can be obtained from Alg.~\ref{alg: Quantized decentralized approach} by setting $\gamma=\zeta=1$ and replacing the compression operator by identity). We further observe that, in steady-state, the network MSD increases by approximately 3 dB when $\mu$ goes from $\mu_0$ to $2\mu_0$. This means that the performance is on the order of $\mu$, as expected from Theorem~\ref{theorem: Network mean-square-error stability} since in the
simulations the absolute noise component is such that $\sigma^2_{c,k}\propto \mu^2$. For the top-$4$ probabilistic uniform and ANQ quantizers, we report  in Fig.~\ref{fig: variable step-size} \emph{(right)} the average number of bits per node, per component, computed according~to:
\begin{equation}
\label{eq: definition of instanteneous rate}
R(i)=\frac{1}{K}\sum_{k=1}^K\frac{1}{M_k}r_{k,i},
\end{equation}
where $r_{k,i}$ is the bit rate given by~\eqref{eq: bit rate general at agent k}, which is associated with the encoding of the error compensated difference vector $\bchi_{k,i}=\bpsi_{k,i}-\bphi_{k,i-1}+\bz_{k,i-1}$ transmitted by agent $k$ at iteration $i$. As it can be observed, for the three different values of the step-size, we approximately obtain the same finite average number of bits in steady-state (approximately 2.4 bits/component/iteration are required on average in steady-state when the top-$4$ probabilistic ANQ quantizer is used). From Table~\ref{table: examples of quantizers} (row $7$), the top-$4$ sparsifier would require an average of $\frac{4(32+4)}{10}= 14.4$ bits/node/component/iteration\footnote{Note that we replaced $B_{\text{HP}}$ by 32 since we are performing the experiments on {MATLAB 2022a} which uses $32$ bits to represent a floating number in single-precision.}, which is almost {six times} higher than the one obtained in steady-state when the probabilistic ANQ is used.  This is expected since the top-$4$ sparsifier  requires encoding the $4$ largest magnitude components of the input with very high precision. On the other hand, the  top-$4$ QSGD {(Table~\ref{table: examples of quantizers}, row $6$, $s=2$)}, which requires encoding the norm of the input with high precision, would need an average of $\frac{32+10+10}{10}=5.2$ bits/node/component/iteration, which is almost {two times} higher than the one obtained for the probabilistic ANQ. 

\begin{figure}
\begin{center}
\includegraphics[scale=0.22]{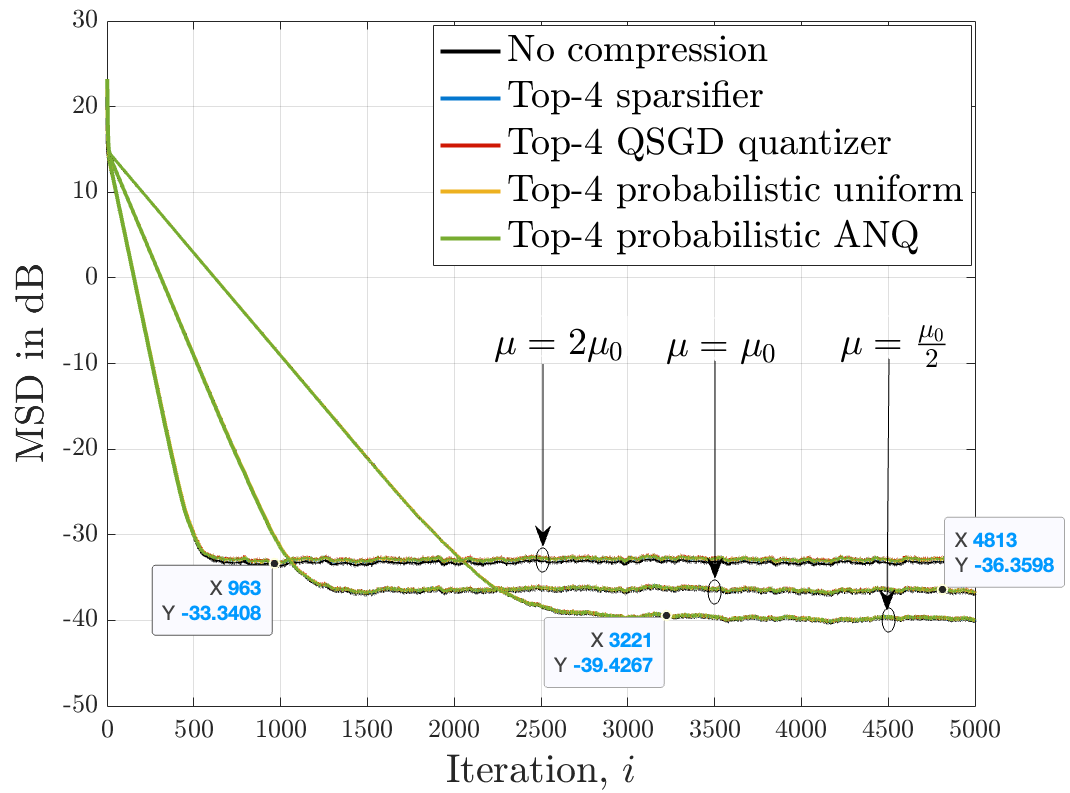}\quad
\includegraphics[scale=0.22]{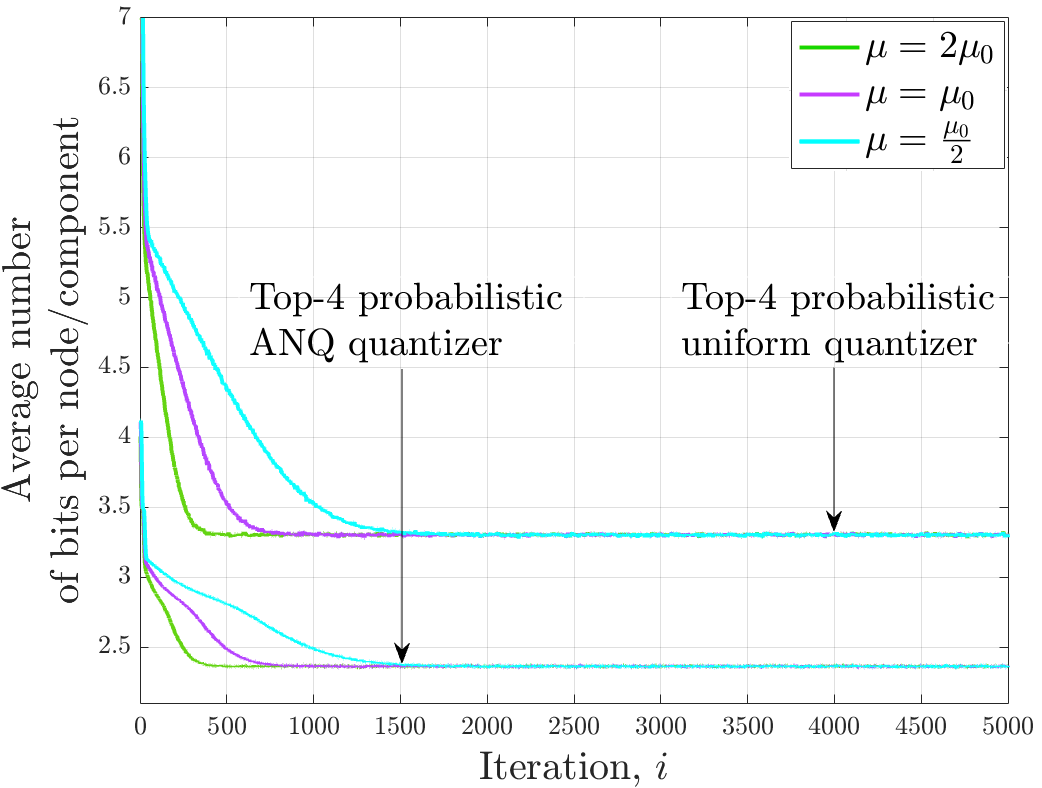}
\caption{Network performance w.r.t. $w^o$ for three different values of the step-size ($\mu_0=0.001$). \textit{(Left)}  Evolution of the MSD learning curves. \textit{(Right)} Evolution of the average number of bits per node, per component, when  the variable-rate probabilistic uniform and ANQ schemes are used at the output of the top-$4$ sparsifier to encode the error compensated difference $\bchi_{k,i}=\bpsi_{k,i}-\bphi_{k,i-1}+\bz_{k,i-1}$ in~\eqref{eq: step 2}.}
\label{fig: variable step-size}
\end{center}
\end{figure}

Evaluating the performance of a learning approach requires considering both the \emph{attained learning error} (MSD)  and the associated \emph{bit expense}. Therefore, in the following, we focus on reporting rate-distortion (RD) curves, where {the bit budget quantifies the rate} and the MSD quantifies the distortion. 

\subsection{Top-$c$ quantization outperforms other compression rules}
\label{subsec: Top-$k$ with quantization outperforms other compression rules}

We report in Fig.~\ref{fig: rate-distortion curve for probabilistic uniform} the RD curves of the DEF-ATC approach with probabilistic uniform and top-$4$ probabilistic uniform quantization. We set $\mu=0.001$ and $\gamma=\zeta=0.9$. Each point of the rate-distortion curve 
corresponds to one value of the parameter $\varepsilon$, which determines the quantization step $\Delta=\mu^{\frac{1+\varepsilon}{2}}$. In the example, we selected $25$ values of $\varepsilon$ linearly spaced in the interval $[10^{-3},1]$. For each value of $\Delta$ (i.e., each point of the curve), the resulting MSD (distortion) and average number of bits/node/component (rate) were obtained by averaging the instantaneous mean-square-deviation $\text{MSD}(i)$ in~\eqref{eq: definition of instanteneous MSD} and averaging the number of bits $R(i)$ in~\eqref{eq: definition of instanteneous rate} over $100$ samples after convergence of the algorithm (the expectations in~\eqref{eq: definition of instanteneous MSD}  and~\eqref{eq: definition of instanteneous rate} are estimated empirically over $100$ Monte Carlo runs). The trade-off between rate and distortion can be observed from Fig.~\ref{fig: rate-distortion curve for probabilistic uniform}, namely, as the rate decreases, the distortion increases, and vice versa. For comparison purposes, we illustrate in Fig.~\ref{fig: rate-distortion curve for probabilistic uniform}  the  distortion (horizontal dashed line) of the uncompressed ATC approach {obtained by averaging $\text{MSD}(i)$ in~\eqref{eq: definition of instanteneous MSD} over $100$ samples after convergence. We also illustrate the specific $\log_2(3)$ bit rate (vertical dashed line) corresponding to minimum number of bits possible for the considered scheme, namely, the variable-rate coding scheme from~\cite{lee2021finite},~\cite[Sec.~II]{nassif2023quantization}. Under this scheme, each component is appended with a parsing symbol, and the $0$ value is encoded as an empty element. Thus, the minimum number of bits/component would correspond to sending only one symbol per component.} As it can be observed, top-$4$ {probabilistic uniform} is more efficient than probabilistic uniform, namely,  it approaches the uncompressed performance (low distortion) at a lower bit rate compared to probabilistic uniform quantization.

\begin{figure}
\begin{center}
\includegraphics[scale=0.25]{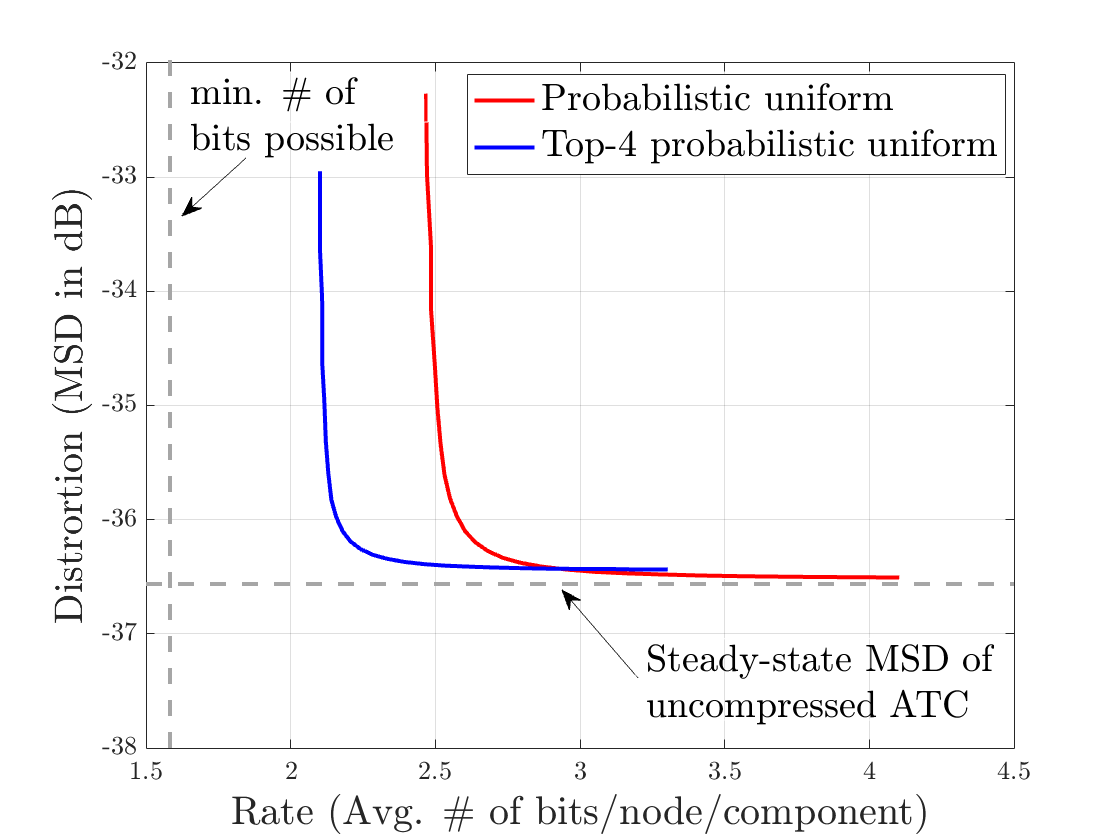}
\caption{Rate-distortion curves of the DEF-ATC approach with probabilistic uniform and top-$4$ probabilistic uniform {quantization}.}
\label{fig: rate-distortion curve for probabilistic uniform}
\end{center}
\end{figure}

\subsection{Performance w.r.t. state-of-the-art baselines}
\label{subsec: Performance w.r.t. state-of-the-art baselines}
 \begin{figure}
\begin{center}
\includegraphics[scale=0.22]{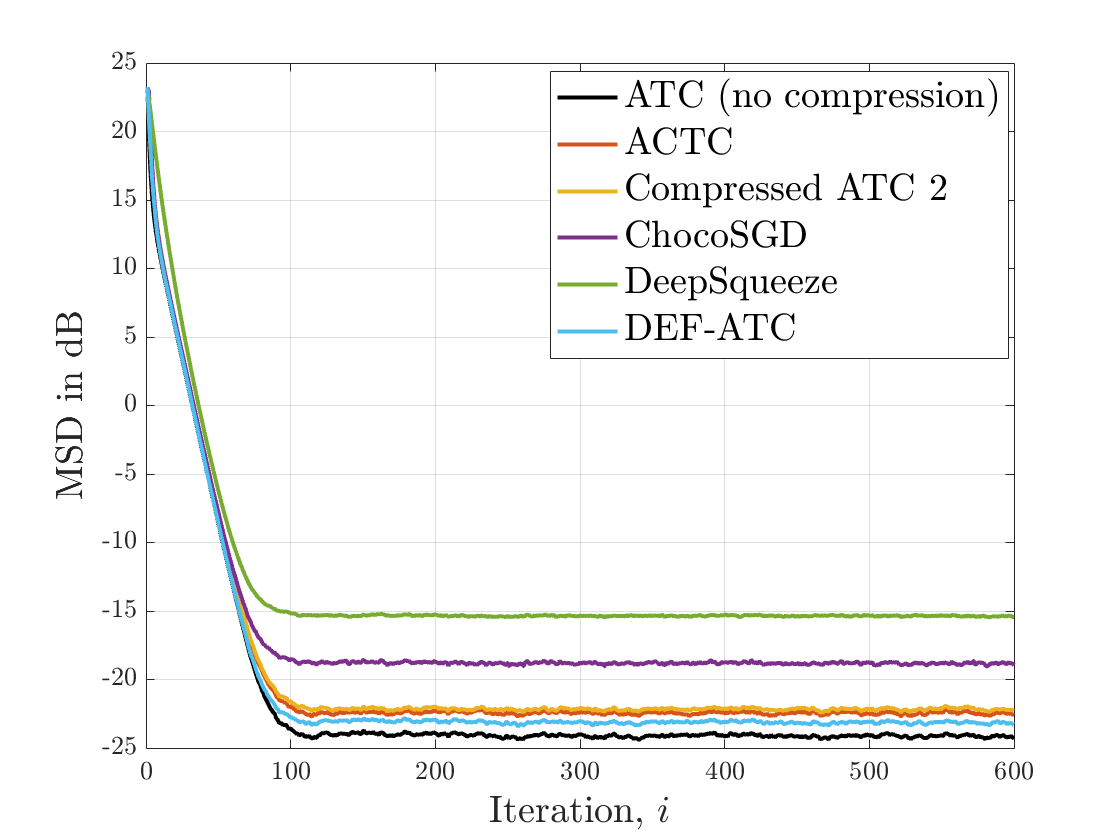}\quad
\includegraphics[scale=0.22]{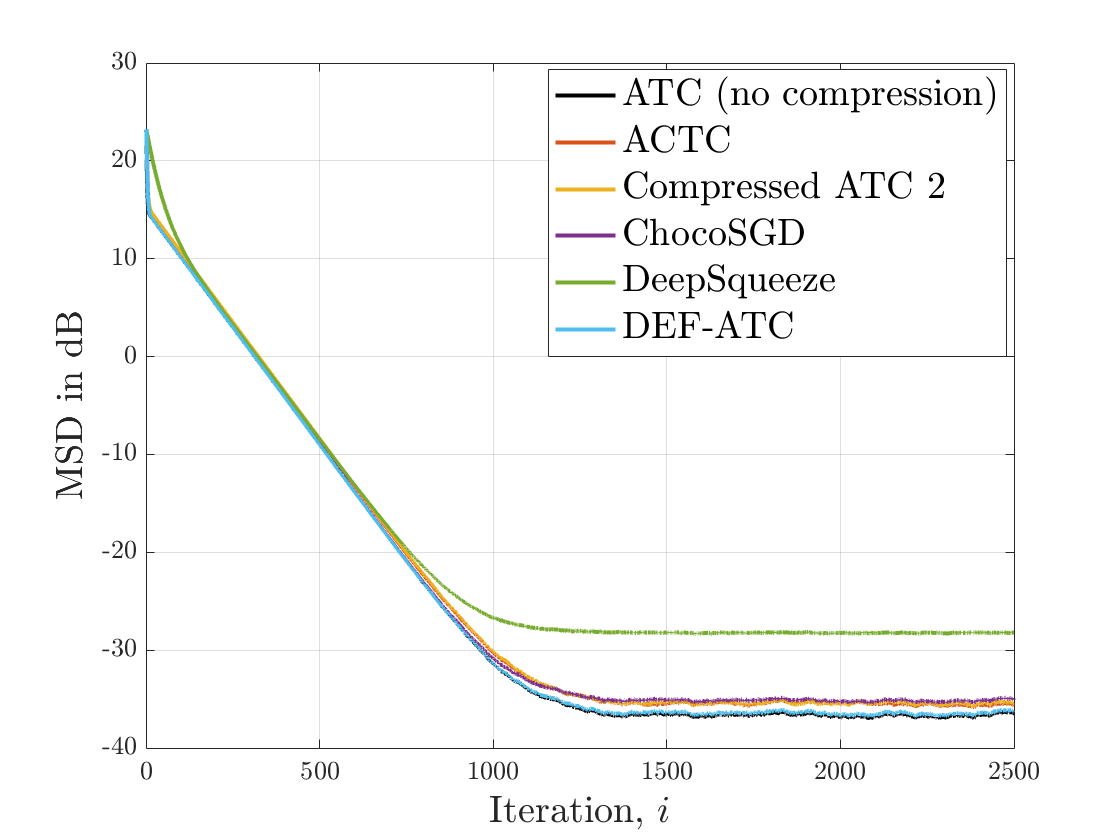}\\
\includegraphics[scale=0.22]{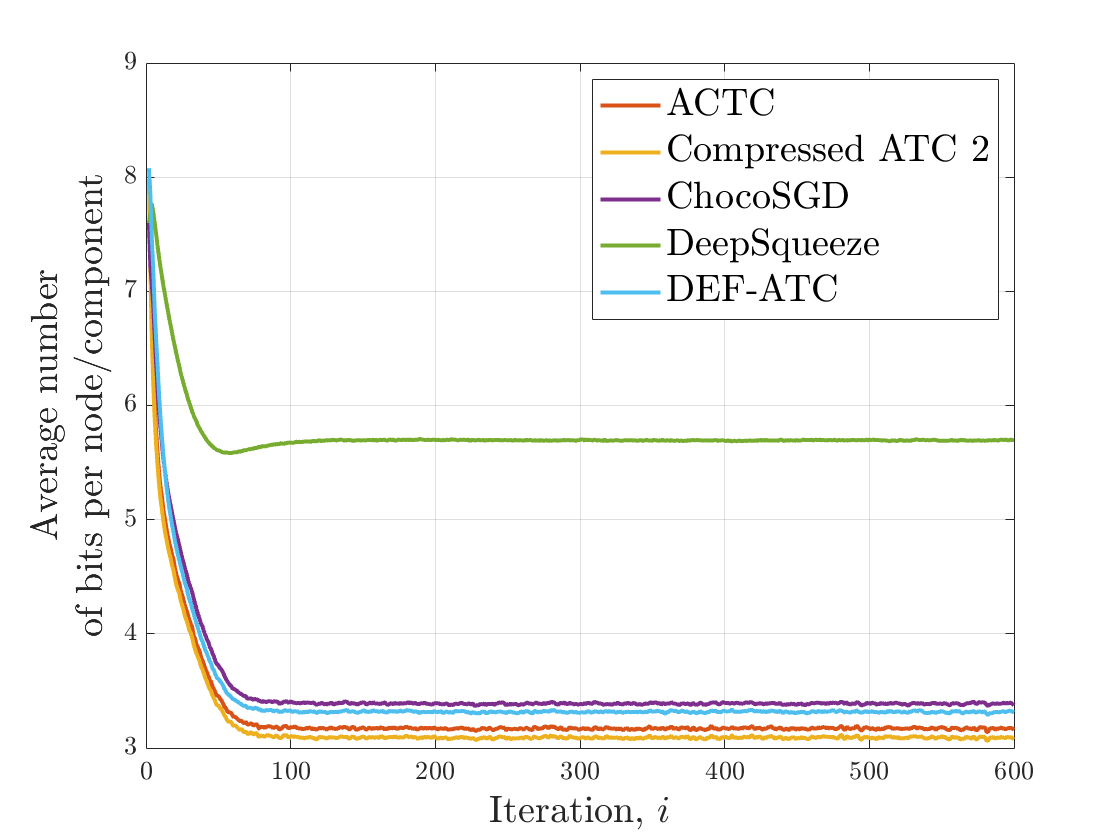}\quad
\includegraphics[scale=0.22]{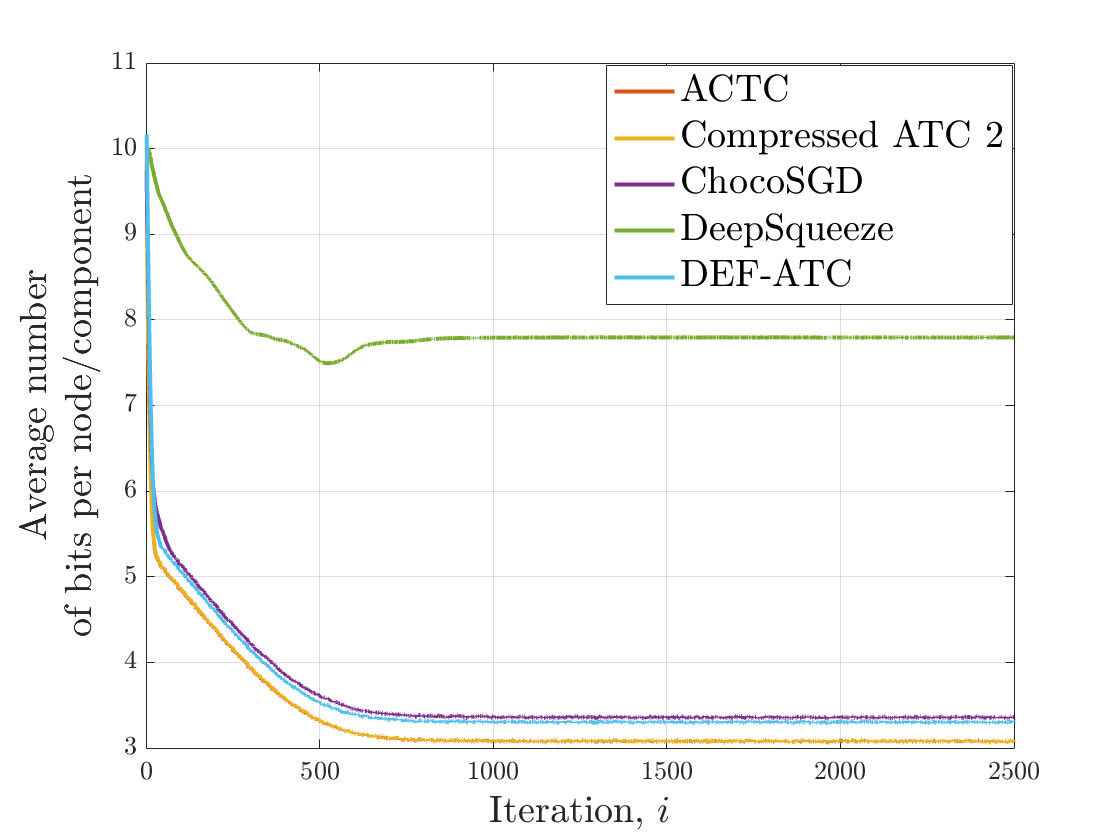}
\caption{DEF-ATC performance w.r.t. to state-of-the-art baselines. Evolution of the  MSD  and average number of bits/node/component (when top-$4$ probabilistic uniform  with variable-rate encoding scheme is used) for different values of the step-size. \textit{(Left)} $\mu=0.01$. \textit{(Right)} $\mu=0.001$.}
\label{fig: MSD sota comparisons}
\end{center}
\end{figure}
In this part, we compare the DEF-ATC {diffusion} to the following approaches: $i)$ ChocoSGD~\cite{koloskova2019decentralized}, $ii)$ DeepSqueeze~\cite{tang2019deepsqueeze}, $iii)$ diffusion ACTC~\cite{carpentiero2022adaptive}, and $iv)$ compressed diffusion ATC approach  (which we refer to as the ``compressed ATC $2$'')~\cite{nassif2023quantization}. We assume that all agents employ the top-$4$ probabilistic uniform quantizer with the variable rate encoding scheme from~\cite{lee2021finite},~\cite[Sec.~II]{nassif2023quantization}. In Fig.~\ref{fig: MSD sota comparisons}, we report the network MSD learning curves with the corresponding bit rates, for 2 different values of the step-size $\mu=0.01$ (\emph{left}) and $\mu=0.001$ (\emph{right}). The results are averaged over $100$ Monte-Carlo runs. We set the quantization parameter $\Delta=\mu$. For the ACTC~\cite{carpentiero2022adaptive} and compressed ATC $2$~\cite{nassif2023quantization} approaches (which were originally designed to handle unbiased probabilistic compression), {we used a step-size $\mu=0.0125$ and  $\mu=0.00125$, in place of $\mu=0.01$ and $\mu=0.001$, respectively, in order to ensure that they have the same learning rate as the uncompressed ATC approach. This configuration ensures that all approaches are compared at the same learning rate}. 
The other parameters of the baselines approaches are set  as follows: $i)$ ChocoSGD: $\gamma=0.9$,  $ii)$ DeepSqueeze: (consensus parameter) $\eta=0.05$ when $\mu=0.01$ and  $\eta=0.2$ when $\mu=0.001$, $iii)$ ACTC: $\zeta=0.9$, $iv)$ compressed ATC $2$: $\gamma=0.9$, and $v)$ DEF-ATC: $\gamma=\zeta=0.9$. {While the space of algorithms' hyperparameters  ($\gamma$, $\zeta$, etc.)  is explored in the next experiment, it is worth noting that the values chosen in the current experiment ensure a stable compressed strategy with the lowest MSD level.} As it can be observed from the MSD learning curves, the DEF-ATC approach tends to outperform state-of-the-art baselines in various step-size regimes. In order to identify which method provides better compression efficiency for a given level of distortion, we report in Fig.~\ref{fig: rate-distortion sota comparisons} the RD curves of the different approaches. Before analyzing the results, it is noteworthy that the bit rate curves reported in Fig.~\ref{fig: MSD sota comparisons} indicate that the DeepSqueeze approach {tends to require a larger number of bits as the step-size decreases. Thus, for very fine quantization (i.e., small step-sizes since $\Delta=\mu$), the number of bits required tends to grow, making it impractical for use in such scenarios. By noting that DeepSqueeze does not employ differential quantization, the increase in bit rate with decreasing quantization step is expected. In fact, in this case, the input values at the compressor remain large (in particular, their effective range does not scale with the step-size $\mu$), leading to an  increase in the bit rate with decreasing quantization step.} 


 \begin{figure}
\begin{center}
\includegraphics[scale=0.22]{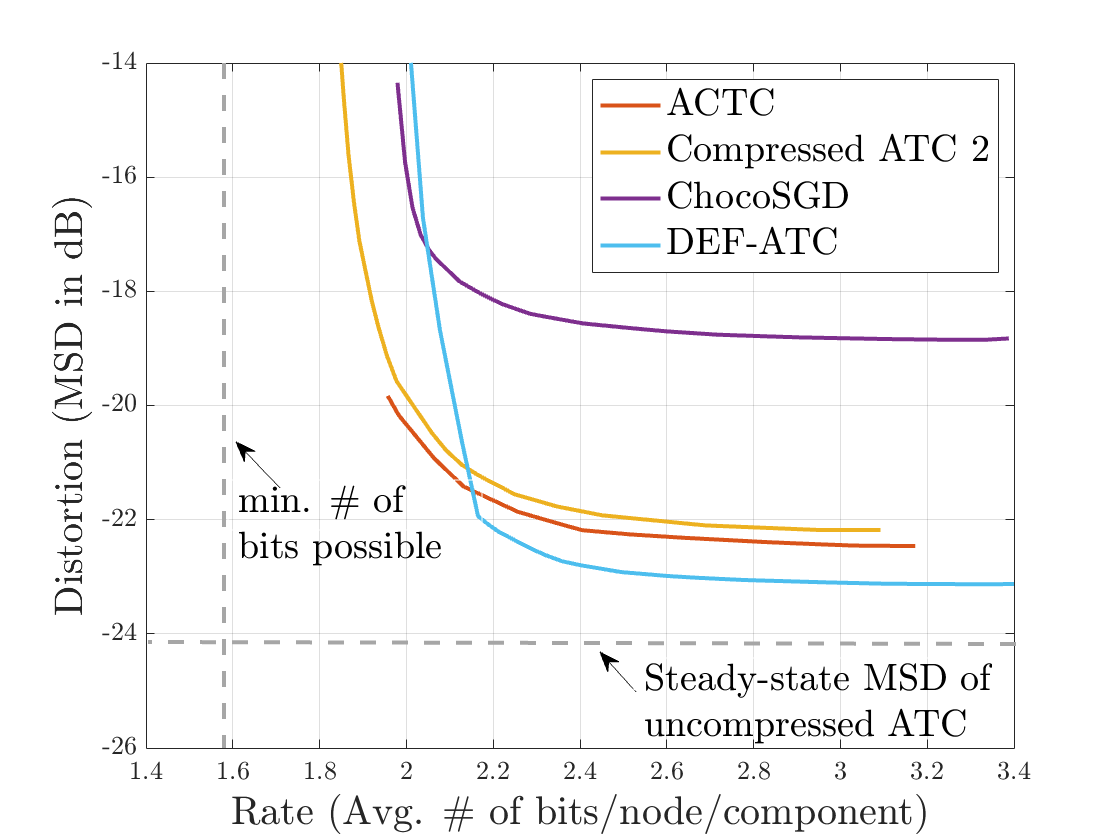}\quad
\includegraphics[scale=0.22]{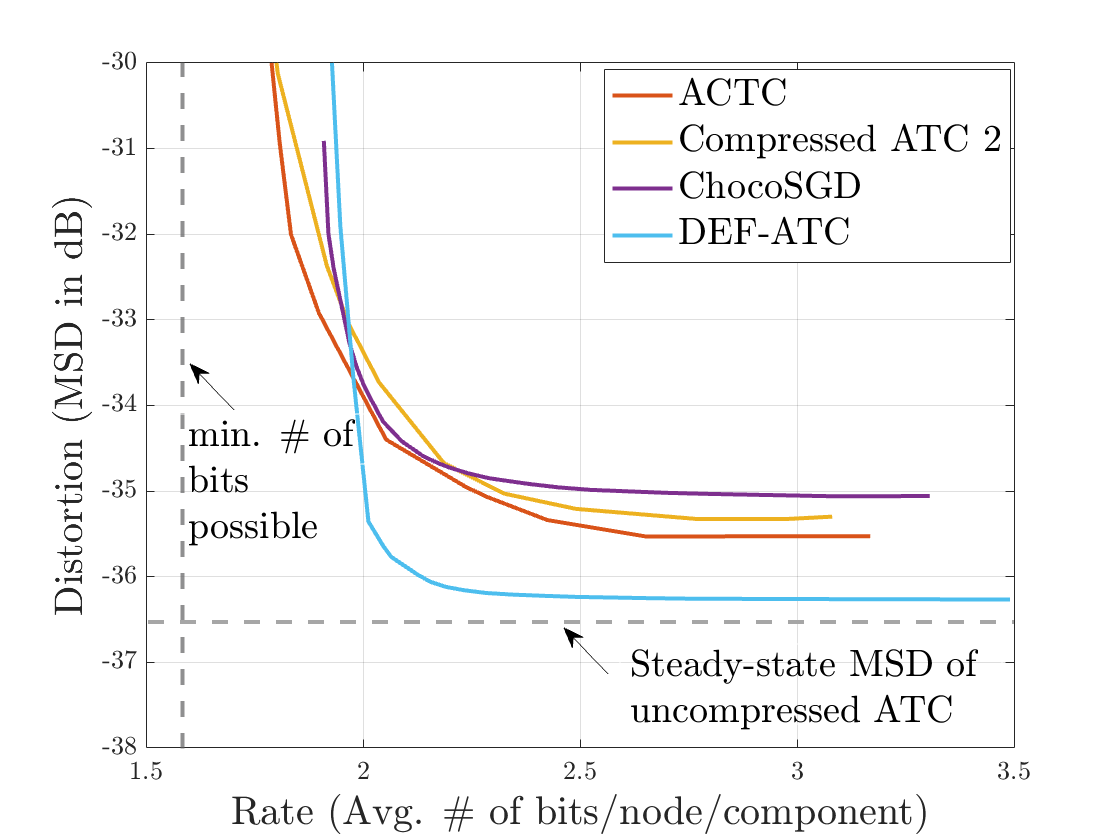}
\caption{Rate-distortion curves of the DEF-ATC and state-of-the-art baselines approaches for different values of the step-size. The top-$4$ probabilistic uniform quantizer with variable-rate encoding scheme is used. \textit{(Left)} $\mu=0.01$. \textit{(Right)} $\mu=0.001$.  }
\label{fig: rate-distortion sota comparisons}
\end{center}
\end{figure}
For each approach, the process for generating the rate-distortion curves in Fig.~\ref{fig: rate-distortion sota comparisons} consists of three main steps. First, we select  a \emph{set of algorithm's hyperparameters} ($\gamma$, $\zeta$, etc.). In particular, for ACTC~\cite{carpentiero2022adaptive}, we select $10$ values of the damping coefficient $\zeta$ uniformly spaced in the interval $[0.1,0.9]$. For compressed ATC $2$~\cite{nassif2023quantization} (and ChocoSGD~\cite{koloskova2019decentralized}), we select $10$ values of the mixing parameter $\gamma$ uniformly spaced in the interval $[0.1,0.9]$. For the DEF-ATC approach, we create two sets of $5$ uniformly spaced values in the interval $[0.1,0.9]$ for the coefficients $\zeta$ and $\gamma$, and then consider all possible pairs from these sets. In the second step, and for each hyperparameter setting (i.e., each element in the sets of step~1), we generate the RD curve by following the same method as in Sec.~\ref{subsec: Top-$k$ with quantization outperforms other compression rules}, namely, we vary the quantization step according to $\Delta=\mu^{\frac{1+\varepsilon}{2}}$, where $25$ values of $\varepsilon$ linearly spaced in the interval $[10^{-3},1]$ are chosen. For each value of $\Delta$, distortion and rate are evaluated by averaging over $100$ samples after convergence (expectations are computed empirically over $50$ Monte Carlo runs). Each RD curve then represents the performance of the algorithm under a specific choice of hyperparameters. In the last step, we generate and report in Fig.~\ref{fig: rate-distortion sota comparisons} the optimal RD curve given by the convex hull of the empirical curves collected in step $2$. This process allows us to identify the best possible performance trade-offs by varying the algorithms' hyperparameters ($\gamma$, $\zeta$, etc.) and the compression parameters, namely, $\varepsilon$. Two learning step-size regimes are considered, namely, $\mu=0.01$ (\emph{left} plot) and $\mu=0.001$ (\emph{right} plot). By exploring the hyperparameter space and by considering different step-size regimes, the results show that the DEF-ATC approach can still achieve the closest performance to the uncompressed approach with a relatively small number of bits (approximately $2.2$--$2.4$ bits/component/iteration are required on average in steady-state), outperforming state-of-the-art baselines.

\subsection{Beyond single-task estimation}
\label{subsec: Beyond classical consensus or single-task estimation}
To illustrate the effectiveness of the DEF-ATC approach in solving general optimization problems of the form~\eqref{eq: network constrained problem}, we conduct an experiment in which agents seek consensus on certain components of their estimates while seeking partial consensus on others. {In particular, we} assume that we have $5$ {connected\footnote{A group of nodes is said to be connected if there is a path between every pair of nodes.}} groups of agents, namely, $\cG_1=\{1,\ldots, 15\}$, $\cG_2=\{16,\ldots, 30\}$, $\cG_3=\{1,\ldots,10\}$, $\cG_4=\{11,\ldots,20\}$, and $\cG_5=\{21,\ldots,30\}${, and that}   
 the model parameter vector $w^\star_k$ at agent $k$ is of the form $w^\star_k=w^{\bullet}_k+\Delta_{k,i}$, where $\Delta_{k,i}$ is a $10\times 1$ vector with each component randomly generated from the Gaussian distribution, with zero mean and variance $0.1$. The vectors $\{w^{\bullet}_k\}$ are generated in such a way that the first $5$ components are common across the network, the components  $6,7,8$ are separately common for agents in $\cG_1$ and $\cG_2$, and the last two components are separately common for agents in groups $\cG_3,\cG_4$, and $\cG_5$. Then, we choose the constraints in~\eqref{eq: network constrained problem} (i.e., the matrix $\cU$) in order to enforce \emph{global} consensus on the first $5$ components of the estimates and \emph{partial} consensus on the remaining components. The partial consensus is as follows. Agents in $\cG_1$ should converge to a consensus on components $6$--$8$, and agents in $\cG_2$ should also converge to a consensus on components $6$--$8$, independently from the first group. For the remaining two components $9$ and $10$, consensus is enforced within each of the groups $\cG_3,\cG_4$, and $\cG_5$. The matrix $\cA$ satisfying the conditions in~\eqref{eq: condition A} and having the same {sparsity structure as the link matrix in Fig.~\ref{fig: network settings} \emph{(left)}} is found by following the same approach as in
~\cite{nassif2020adaptation}.We assume that all agents employ the top-$4$ probabilistic uniform quantizer. In Fig.~\ref{fig: variable step-size overlapping} \emph{(left)} and \emph{(middle)}, we report the network MSD~\eqref{eq: definition of instanteneous MSD} and average number of bits/node/component~\eqref{eq: definition of instanteneous rate} for $3$ different values of the step-size $\mu$.  The results are averaged over $100$ Monte-Carlo runs. We set $\gamma=\zeta=0.9$ and the quantizer parameter $\Delta=\mu$.  As in Sec.~\ref{subsec: Illustrating the theoretical findings}, we observe, in the small step-size regime, that the DEF-ATC achieves the same performance (which is on the order of $\mu$) as the uncompressed ATC approach, and  is able to maintain a finite bit rate when the step-size approaches zero. This illustrates the effectiveness of DEF-ATC in handling different problem settings, beyond  traditional single-task estimation. In Fig.~\ref{fig: variable step-size overlapping} \emph{(right)}, we report the RD curves  of DEF-ATC with probabilistic uniform and top-$4$ probabilistic uniform quantization. These curves have been generated in the same way as those in Fig.~\ref{fig: rate-distortion curve for probabilistic uniform}. The results show that quantizing only the highest magnitude components of a vector, as opposed to the entire vector, can reduce the number of bits required while maintaining a low level of distortion.

\begin{figure}
\begin{center}
\includegraphics[scale=0.17]{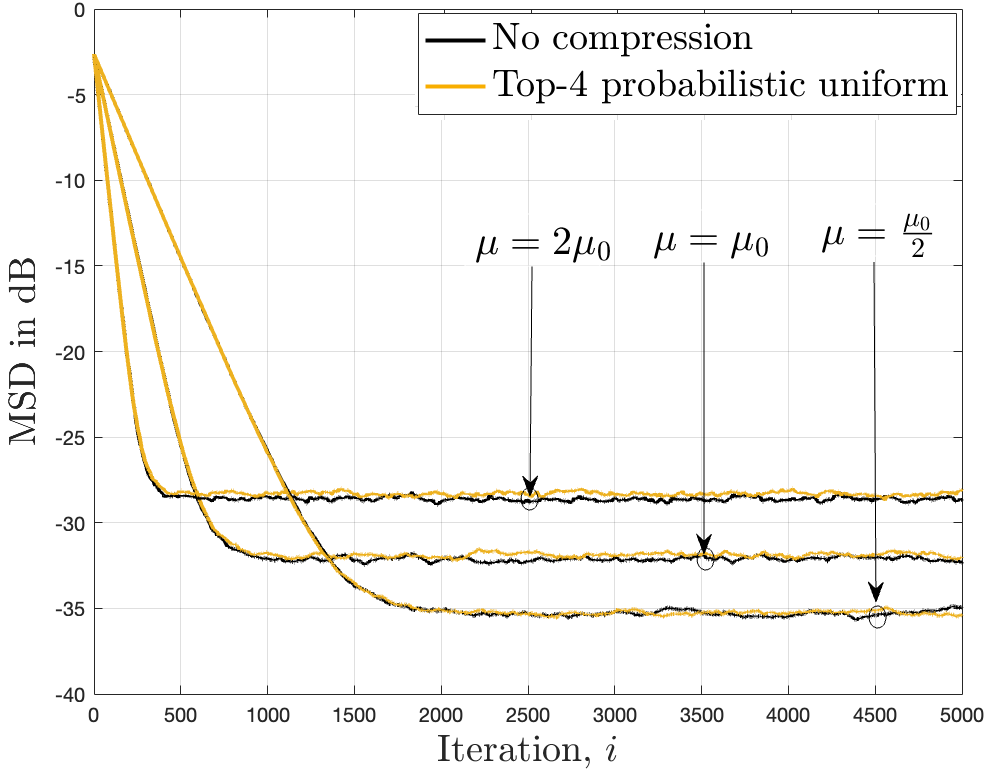}
\includegraphics[scale=0.17]{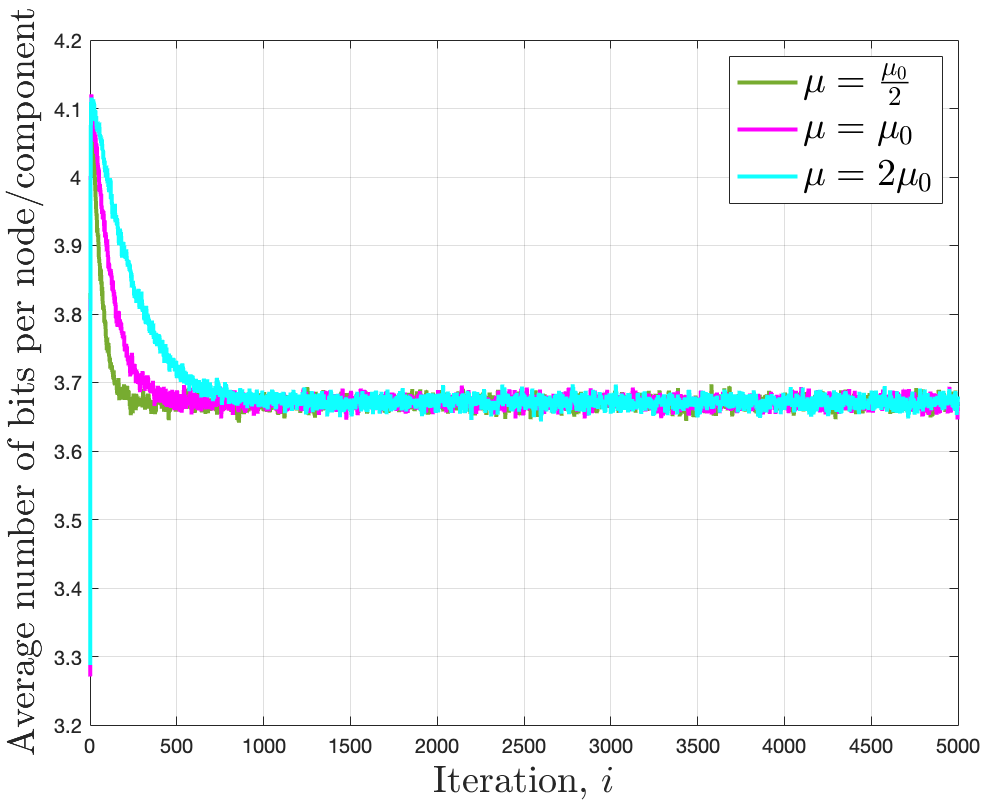}
\includegraphics[scale=0.17]{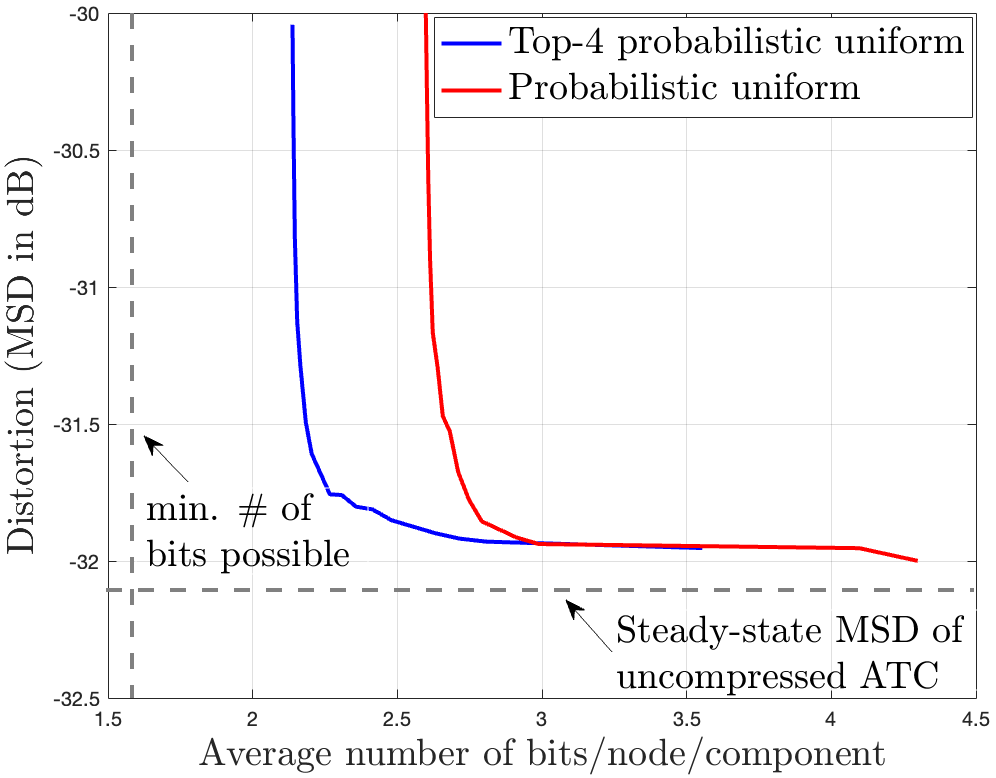}
\caption{Beyond single-task estimation. \textit{(Left)} Network MSD learning curves w.r.t. $w^o_k$ in~\eqref{eq: network constrained problem} for three different values of the step-size ($\mu_0=0.001$). \textit{(Middle)} Evolution of the average number of bits/node/component when  the variable-rate scheme is used. \textit{(Right)} Rate-distortion curves of the DEF-ATC (when $\mu=\mu_0$) with probabilistic uniform and top-$4$ probabilistic uniform quantization. }
\label{fig: variable step-size overlapping}
\end{center}
\end{figure}

\section{Conclusion}
In this work, we presented an approach for solving decentralized learning problems where agents have individual {risks} to minimize subject to subspace constraints that require the minimizers across the network to lie in low-dimensional subspaces. This constrained formulation includes consensus or single-task optimization as {a} special  case, and allows for more general task relatedness models such as multitask smoothness and  coupled optimization. To reduce the communication cost among agents, we incorporated compression into the decentralized approach by employing \emph{differential quantization} at the agent level to compress the iterates before communicating them to neighbors. {In addition},  we implemented in the learning approach an \emph{error-feedback} mechanism
, which consists of incorporating  the compression  error into subsequent steps. We then showed that, under some general conditions on the compression noise, and for sufficiently small step-sizes~$\mu$, the resulting communication-efficient strategy is stable both in  terms of  mean-square error and  average bit rate. The results {showed} that, in the small step-size regime, the iterates generated by the decentralized communication-efficient approach achieve the same  performance as the decentralized  baseline full-precision approach where no communication compression is performed. Simulations illustrated   the theoretical findings and the effectiveness of the approach.

\begin{appendices}
\section{Proof of Lemma~\ref{lemma: Property of the top-$c$ quantizer}}
\label{app: Property of the top-c quantizer}
Let $\cI'_c$ denote the complement set of $\cI_c$. 
Since all the components $\{x_j,j\in\cI_c\}$ are in magnitude greater than or equal to the components $\{x_j,j\in\cI'_c\}$, we can write:
\begin{equation}
\frac{1}{c}\sum_{j\in\cI_c}x_j^2\geq \frac{1}{L}\sum_{j=1}^Lx_j^2,
\end{equation}
from which we conclude the following useful identity:
\begin{equation}
\label{eq: intermediate equation for the sparsifier}
\sum_{j\in\cI_c}x_j^2\geq  \frac{c}{L}\|x\|^2.
\end{equation}
Now, to show that the top-$c$ quantizer is a bounded-distortion compression operator with parameters given by~\eqref{eq: bounded distortion compression operator parameters}, we can manipulate the compression error as follows:
\begin{align}
\expec\|x-\bcC(x)\|^2=\expec\|x-\alpha\bcQ(\cS(x))\|^2&=\expec\|x-\alpha\cS(x)+\alpha(\cS(x)-\bcQ(\cS(x)))\|^2\notag\\
&\overset{\text{(a)}}{=}\|x-\alpha\cS(x)\|^2+\alpha^2\expec\|\cS(x)-\bcQ(\cS(x))\|^2\notag\\
&\leq\|x-\alpha\cS(x)\|^2+\alpha^2\beta^2_q\|\cS(x)\|^2+\alpha^2\sigma^2_q\notag\\
&=(1-\alpha)^2\sum_{j\in\cI_c}x_j^2+\sum_{j\in\cI'_c}\|x_j\|^2+\alpha^2\beta^2_q\|\cS(x)\|^2+\alpha^2\sigma^2_q\notag\\
&=\left((1-\alpha)^2+\alpha^2\beta^2_q\right)\sum_{j\in\cI_c}x_j^2+\sum_{j\in\cI'_c}\|x_j\|^2+\alpha^2\sigma^2_q\notag\\
&=\|x\|^2-\left(1-\left((1-\alpha)^2+\alpha^2\beta^2_q\right)\right)\sum_{j\in\cI_c}\|x_j\|^2+\alpha^2\sigma^2_q\notag\\
&\overset{\eqref{eq: intermediate equation for the sparsifier}}\leq\left(1-\left(1-\left((1-\alpha)^2+\alpha^2\beta^2_q\right)\right)\frac{c}{L}\right)\|x\|^2+\alpha^2\sigma^2_q,
\end{align}
where in (a) we used the fact that $\alpha\langle x-\alpha\cS(x),\expec[\cS(x)-\bcQ(\cS(x))]\rangle=0$ in both cases: biased ($\alpha=1$) and unbiased ($\alpha=\frac{1}{1+\beta^2_q}$) quantizer $\bcQ(\cdot)$.


\section{Mean-square-error analysis}
\label{app: Mean-square-error analysis}
We consider the transformed iterates $\bphib^z_{i}$ and $\bphic^z_{i}$ in~\eqref{eq: evolution of the centralized recursion 1} and~\eqref{eq: evolution of the centralized recursion 3}, respectively.  Computing the second-order moment of both sides of~\eqref{eq: evolution of the centralized recursion 1}, we get:
\begin{equation}
\expec\|\bphib^z_i\|^2=\expec\|(I_P-\mu\bcD_{11,i-1})\bphib^z_{i-1}-\mu\bcD_{12,i-1}\bphic^z_{i-1}-\mu\bcD_{11,i-1}\bzb_{i-1}-\mu\bcD_{12,i-1}\bzc_{i-1}\|^2+\mu^2\expec\|\bsb_{i}\|^2,\label{eq: centralized error recursion biased}
\end{equation}
where,  from Assumption~\ref{assump: gradient noise} on the gradient noise processes, we used the fact that:
\begin{align}
\expec[\bx_{i-1}^{\top}\bsb_{i}]&=\expec\left[\expec\left[\bx_{i-1}^{\top}\bsb_{i}\Big|\{{\bphi_{\ell,i-1},\bz_{\ell,i-1}}\}_{\ell=1}^K\right]\right]=\expec\left[\bx_{i-1}^{\top}\expec\left[\bsb_{i}\Big|\{{\bphi_{\ell,i-1},\bz_{\ell,i-1}}\}_{\ell=1}^K\right]\right]=0
\end{align} 
with $\bx_{i-1}=(I_P-\mu\bcD_{11,i-1})\bphib^z_{i-1}-\mu\bcD_{12,i-1}\bphic^z_{i-1}-\mu\bcD_{11,i-1}\bzb_{i-1}-\mu\bcD_{12,i-1}\bzc_{i-1}$. Using similar arguments, we can also show that:
\begin{align}
\expec\|\bphic^z_i\|^2=&\expec\|(\cJ''_{\epsilon}-\mu\bcD_{22,i-1})\bphic^z_{i-1}-\mu\bcD_{21,i-1}\bphib^z_{i-1}+\mu\widecheck{b}-\mu\bcD_{21,i-1}\bzb_{i-1}-\left(\zeta(I-\cJ_{\epsilon}')+\mu\bcD_{22,i-1}\right)\bzc_{i-1}\|^2\notag\\
&\qquad\qquad\qquad+\mu^2\expec\|\bsc_{i}\|^2.\label{eq: non centralized error  recursion biased}
\end{align}
By using similar arguments as those used to  establish inequalities~(119) and~(124) in~\cite[Appendix D]{nassif2020adaptation}, we can show that:
\begin{equation}
\label{eq: centralized error vector recursion inequality biased}
\begin{split}
\expec\|\bphib^z_i\|^2&\leq(1-\mu\sigma_{11})\expec\|\bphib^z_{i-1}\|^2+\frac{\mu}{\sigma_{11}}\expec\|\bcD_{12,i-1}\bphic^z_{i-1}+\bcD_{11,i-1}\bzb_{i-1}+\bcD_{12,i-1}\bzc_{i-1}\|^2+\mu^2\expec\|\bsb_{i}\|^2\\
&\leq(1-\mu\sigma_{11})\expec\|\bphib^z_{i-1}\|^2+\frac{3\mu\sigma^2_{12}}{\sigma_{11}}\expec\|\bphic^z_{i-1}\|^2+3\mu\sigma_{11}\expec\|\bzb_{i-1}\|^2+\frac{3\mu\sigma^2_{12}}{\sigma_{11}}\expec\|\bzc_{i-1}\|^2+\mu^2\expec\|\bsb_{i}\|^2\\
&\leq(1-\mu\sigma_{11})\expec\|\bphib^z_{i-1}\|^2+\frac{3\mu\sigma^2_{12}}{\sigma_{11}}\expec\|\bphic^z_{i-1}\|^2+\left(3\mu\sigma_{11}+\frac{3\mu\sigma^2_{12}}{\sigma_{11}}\right)\expec\|\cV_{\epsilon}^{-1}\bz_{i-1}\|^2+{\mu^2}\expec\|\bsb_{i}\|^2,
\end{split}
\end{equation}
and 
\begin{equation}
\label{eq: centralized error vector recursion inequality 2 biased}
\begin{split}
\expec\|\bphic^z_i\|^2&\leq\|\cJ''_{\epsilon}\|\expec\|\bphic^z_{i-1}\|^2+\frac{2\mu^2}{1-\|\cJ''_{\epsilon}\|}\expec\|\bcD_{22,i-1}\bphic^z_{i-1}+\bcD_{21,i-1}\bphib^z_{i-1}-\widecheck{b}+\bcD_{21,i-1}\bzb_{i-1}+\bcD_{22,i-1}\bzc_{i-1}\|^2+\\
&\qquad \frac{2\zeta^2\|I-\cJ_{\epsilon}'\|^2}{1-\|\cJ''_{\epsilon}\|}\expec\|\bzc_{i-1}\|^2+\mu^2\expec\|\bsc_{i}\|^2\\
&\leq\left(\|\cJ''_{\epsilon}\|+\frac{10\mu^2\sigma_{22}^2}{1-\|\cJ''_{\epsilon}\|}\right)\expec\|\bphic^z_{i-1}\|^2+\left(\frac{10\mu^2\sigma_{21}^2}{1-\|\cJ''_{\epsilon}\|}\right)\expec\|\bphib^z_{i-1}\|^2+\left(\frac{10\mu^2}{1-\|\cJ''_{\epsilon}\|}\right)\|\widecheck{b}\|^2+\\
&\qquad\left(\frac{10\mu^2\sigma_{21}^2}{1-\|\cJ''_{\epsilon}\|}\right)\expec\|\bzb_{i-1}\|^2+\left(\frac{2\zeta^2\|I-\cJ_{\epsilon}'\|^2}{1-\|\cJ''_{\epsilon}\|}+\frac{10\mu^2\sigma_{22}^2}{1-\|\cJ''_{\epsilon}\|}\right)\expec\|\bzc_{i-1}\|^2+\mu^2\expec\|\bsc_{i}\|^2\\
&\leq\left(\|\cJ''_{\epsilon}\|+\frac{10\mu^2\sigma_{22}^2}{1-\|\cJ''_{\epsilon}\|}\right)\expec\|\bphic^z_{i-1}\|^2+\left(\frac{10\mu^2\sigma_{21}^2}{1-\|\cJ''_{\epsilon}\|}\right)\expec\|\bphib^z_{i-1}\|^2+\left(\frac{10\mu^2}{1-\|\cJ''_{\epsilon}\|}\right)\|\widecheck{b}\|^2+\\
&\qquad\left(\frac{2{\zeta^2}\|I-\cJ_{\epsilon}'\|^2}{1-\|\cJ''_{\epsilon}\|}+\frac{10\mu^2(\sigma_{22}^2+\sigma_{21}^2)}{1-\|\cJ''_{\epsilon}\|}\right)\expec\|\cV_{\epsilon}^{-1}\bz_{i-1}\|^2+\mu^2\expec\|\bsc_{i}\|^2,
\end{split}
\end{equation}
for some positive constant $\sigma_{11}$ and non-negative constants $\sigma_{12},\sigma_{21}$, and $\sigma_{22}$ independent of $\mu$, and where we used the fact {that the} $2-$induced matrix norm of the block diagonal matrix $\cJ''_{\epsilon}$ in~\eqref{eq: definition of J'' epsilon} satisfies $\|\cJ''_{\epsilon}\|\in(0,1)$. In fact, from~\eqref{eq: jordan decomposition of A'}, we can re-write $\cJ''_{\epsilon}$ in the following form: 
\begin{equation}
\label{eq: jordan decomposition of A"}
\cJ''_{\epsilon}=(1-\gamma\zeta)I_{M-P}+\gamma\zeta\cJ_{\epsilon}.
\end{equation}
By following similar arguments as in~\cite[pp.~516--517]{sayed2014adaptation}, we can first show that  $\cJ''_{\epsilon}$ in~\eqref{eq: jordan decomposition of A"}  satisfies:
\begin{equation}
\label{eq: bounding square of j prime biased}
\|\cJ_{\epsilon}''\|^2\leq(\rho(\cJ''_{\epsilon})+\gamma\zeta\epsilon)^2.
\end{equation}
From~\eqref{eq: jordan decomposition of A"}, we can also show that: 
\begin{equation}
\label{eq: equation to bound rho J prime biased}
\rho(\cJ''_{\epsilon})\leq(1-\gamma\zeta)+\gamma\zeta\rho(\cJ_{\epsilon}).
\end{equation} 
Using the fact that $\rho(\cJ_{\epsilon})\in(0,1)$ from~\cite[Lemma~2]{nassif2020adaptation} and the fact that $\gamma\zeta\in(0,1]$, we obtain $\rho(\cJ''_{\epsilon})\in(0,1)$.  Moreover, since $\rho(\cJ''_{\epsilon})+\gamma\zeta\epsilon$ is non-negative, by replacing~\eqref{eq: equation to bound rho J prime biased} into~\eqref{eq: bounding square of j prime biased},  we obtain:
\begin{equation}
\label{eq: spectral norm of J''}
\begin{split}
\|\cJ_{\epsilon}''\|&\leq(1-\gamma\zeta)+\gamma\zeta\rho(\cJ_{\epsilon})+\gamma\zeta\epsilon\\
&=1-\gamma\zeta(1-\rho(\cJ_{\epsilon})-\epsilon).
\end{split}
\end{equation}
This identity will be used in the subsequent analysis. Returning to the result~\eqref{eq: centralized error vector recursion inequality 2 biased}, and as it can be seen from~\eqref{eq: bias transformed vector}, $\widecheck{b}=\cV_{L,\epsilon}^\top b$ depends on $b$ in~\eqref{eq: equation for b}, which is defined in terms of the gradients $\{\nabla_{w_k}J_k(w^o_k)\}$.  Since the costs $J_k(w_k)$ are twice differentiable, then $\|b\|^2$ is bounded and we obtain $\|\widecheck{b}\|^2=O(1)$.

For the gradient noise terms, by following similar arguments as in~\cite[Chapter 9]{sayed2014adaptation},~\cite[Appendix D]{nassif2020adaptation} and by using Assumption~\ref{assump: gradient noise}, we can show that:
\begin{equation}
\label{eq: gradient noise term bound 0 biased}
\expec\|\bsb_i\|^2+\expec\|\bsc_i\|^2=\expec\|\cV_{\epsilon}^{-1}\bs_i\|^2\leq v_1^2\beta_{s,\max}^2 \expec\|\bcwt_{i-1}\|^2+ v_1^2\overline{\sigma}^2_s,
\end{equation}
where $v_1=\|\cV_{\epsilon}^{-1}\|$, $\beta_{s,\max}^2=\max_{1\leq k\leq K}\beta_{s,k}^2$, and $\overline{\sigma}^2_s=\sum_{k=1}^K\sigma^2_{s,k}$. 
 Using expression~\eqref{eq: wt in terms of phi}, the fact that $\bphit_{i-1}=\bphit^z_{i-1}+\bz_{i-1}$, and the Jordan decomposition of the matrix $\cA'$ in~\eqref{eq: jordan decomposition of A'}, we obtain:
\begin{align}
\expec\|\bsb_i\|^2+\expec\|\bsc_i\|^2&\leq v_1^2\beta_{s,\max}^2 \expec\|\cA'(\bphit^z_{i-1}+\bz_{i-1})\|^2+ v_1^2\overline{\sigma}^2_s\notag\\
&\leq v_1^2\beta_{s,\max}^2 \expec\|\cV_{\epsilon}\Lambda'(\cV_{\epsilon}^{-1}\bphit^z_{i-1}+\cV_{\epsilon}^{-1}\bz_{i-1})\|^2+ v_1^2\overline{\sigma}^2_s\notag\\
&\overset{\text{(a)}}\leq v_1^2\beta_{s,\max}^2 v_2^2(\expec\|\cV_{\epsilon}^{-1}\bphit^z_{i-1}+\cV_{\epsilon}^{-1}\bz_{i-1}\|^2)+ v_1^2\overline{\sigma}^2_s\notag\\
&\leq 2v_1^2\beta_{s,\max}^2 v_2^2(\expec\|\bphib^z_{i-1}\|^2+\expec\|\bphic^z_{i-1}\|^2)+2v_1^2\beta_{s,\max}^2 v_2^2\expec\|\cV_{\epsilon}^{-1}\bz_{i-1}\|^2+ v_1^2\overline{\sigma}^2_s,\label{eq: gradient noise term bound biased}
\end{align}
where $v_2=\|\cV_{\epsilon}\|$. {In step (a) we used the sub-multiplicative property of norms and the fact that the $2-$induced matrix norm of the block diagonal matrix $\Lambda'$ in~\eqref{eq: jordan decomposition of A'} is equal to $1$}. Using the bound~\eqref{eq: gradient noise term bound biased} into~\eqref{eq: centralized error vector recursion inequality biased} and~\eqref{eq: centralized error vector recursion inequality 2 biased}, we obtain:
\begin{equation}
\label{eq: centralized error vector recursion inequality with bounding gradient noise biased}
\begin{split}
\expec\|\bphib^z_i\|^2\leq&\left(1-\mu\sigma_{11}+2\mu^2v_1^2\beta_{s,\max}^2 v_2^2\right)\expec\|\bphib^z_{i-1}\|^2+\left(\frac{3\mu\sigma^2_{12}}{\sigma_{11}}+2\mu^2v_1^2\beta_{s,\max}^2 v_2^2\right)\expec\|\bphic^z_{i-1}\|^2+\\
& \left(3\mu\sigma_{11}+\frac{3\mu\sigma^2_{12}}{\sigma_{11}}+2\mu^2v_1^2\beta_{s,\max}^2 v_2^2\right)\expec\|\cV_{\epsilon}^{-1}\bz_{i-1}\|^2+\mu^2v_1^2\overline{\sigma}^2_s,
\end{split}
\end{equation}
and
\begin{equation}
\label{eq: centralized error vector recursion inequality 2 with bounding gradient noise biased}
\begin{split}
\expec\|\bphic^z_i\|^2\leq&\left(\|\cJ''_{\epsilon}\|+\frac{10\mu^2\sigma_{22}^2}{1-\|\cJ''_{\epsilon}\|}+2\mu^2v_1^2\beta_{s,\max}^2 v_2^2\right)\expec\|\bphic^z_{i-1}\|^2+\left(\frac{10\mu^2\sigma_{21}^2}{1-\|\cJ''_{\epsilon}\|}+2\mu^2v_1^2\beta_{s,\max}^2 v_2^2\right)\expec\|\bphib^z_{i-1}\|^2+\\
&\left(\frac{2\zeta^2\|I-\cJ_{\epsilon}'\|^2}{1-\|\cJ''_{\epsilon}\|}+\frac{10\mu^2(\sigma_{22}^2+\sigma_{21}^2)}{1-\|\cJ''_{\epsilon}\|}+2\mu^2v_1^2\beta_{s,\max}^2 v_2^2\right)\expec\|\cV_{\epsilon}^{-1}\bz_{i-1}\|^2+\left(\frac{10\mu^2}{1-\|\cJ''_{\epsilon}\|}\right)\|\widecheck{b}\|^2+\mu^2v_1^2\overline{\sigma}^2_s.
\end{split}
\end{equation}

Now, for the quantization noise vector $\bz_i$ in~\eqref{eq: equation for z}, we have:
\begin{equation}
\label{eq: expression 1 biased}
\begin{split}
\expec\|\cV_{\epsilon}^{-1}\bz_i\|^2\leq \zeta^2v_1^2\left(\sum_{k=1}^K\expec\|\bz_{k,i}\|^2\right).
\end{split}
\end{equation}
From~\eqref{eq: quantization error} and Assumption~\ref{assump: quantization noise}, and since $\bpsi_{k,i}-\bphi_{k,i-1}=\bphit_{k,i-1}-\bpsit_{k,i}$, we can write:
\begin{equation}
\label{eq: expression 3 biased}
\begin{split}
\expec\|\bz_{k,i}\|^2\leq{\beta}^2_{c,k}\expec\|\bphit_{k,i-1}-\bpsit_{k,i}+\bz_{k,i-1}\|^2+{\sigma}^2_{c,k},
\end{split}
\end{equation}
and, therefore,
\begin{equation}
\label{eq: expression 4 biased}
\begin{split}
{\sum_{k=1}^K\expec\|\bz_{k,i}\|^2}\leq {\beta}_{c,\max}^2\expec\|\bphit_{i-1}-\bpsit_{i}+\bz_{i-1}\|^2+\overline{\sigma}^2_c,
\end{split}
\end{equation}
where ${\beta}_{c,\max}^2=\max_{1\leq k\leq K}\{{\beta}^2_{c,k}\}$ and $\overline{\sigma}^2_c=\sum_{k=1}^K{\sigma}^2_{c,k}$. Since the analysis is facilitated by transforming the network vectors into the Jordan decomposition basis of the matrix $\cA'$, we proceed by noting that the term ${\sum_{k=1}^K\expec\|\bz_{k,i}\|^2}$ can be bounded as follows:
\begin{equation}
\label{eq: expression 6 biased}
\begin{split}
{\sum_{k=1}^K\expec\|\bz_{k,i}\|^2}&\overset{\eqref{eq: expression 4 biased}}\leq  {\beta}_{c,\max}^2\expec\|\cV_{\epsilon}\cV_{\epsilon}^{-1}(\bphit_{i-1}-\bpsit_{i}+\bz_{i-1})\|^2+ \overline{\sigma}^2_c\\
&~\leq {\beta}_{c,\max}^2\|\cV_{\epsilon}\|^2\expec\|\cV_{\epsilon}^{-1}(\bphit_{i-1}-\bpsit_{i}+\bz_{i-1})\|^2+ \overline{\sigma}^2_c\\
&~= v_2^2 {\beta}_{c,\max}^2\expec\|\cV_{\epsilon}^{-1}(\bphit_{i-1}-\bpsit_{i}+\bz_{i-1})\|^2+\overline{\sigma}^2_c\\
&~= v_2^2 {\beta}_{c,\max}^2\left(\expec\|\overline{\bchi}_{i}\|^2+\expec\|\widecheck{\bchi}_i\|^2\right)+\overline{\sigma}^2_c,
\end{split}
\end{equation}
where
\begin{align}
\overline{\bchi}_{i}&\triangleq\cU^{\top}(\bphit_{i-1}-\bpsit_{i}+\bz_{i-1}),\\
\widecheck{\bchi}_i&\triangleq\cV_{L,\epsilon}^{\top}(\bphit_{i-1}-\bpsit_{i}+\bz_{i-1}).
\end{align}
Therefore, by combining~\eqref{eq: expression 1 biased} and~\eqref{eq: expression 6 biased}, we obtain:
\begin{equation}
\label{eq: expression 10 biased}
\begin{split}
\expec\|\cV_{\epsilon}^{-1}\bz_i\|^2\leq \zeta^2v_1^2v_2^2{\beta}_{c,\max}^2[\expec\|\overline{\bchi}_{i}\|^2+\expec\|\widecheck{\bchi}_i\|^2]+\zeta^2v_1^2\overline{\sigma}^2_c.
\end{split}
\end{equation}
We focus now on deriving the recursions that allow to examine the time-evolution of the transformed vectors $\overline{\bchi}_{i}$ and $\widecheck{\bchi}_i$. Subtracting $\bphit_{i-1}$  from both sides of~\eqref{eq: network error vector psi}, {adding $\bz_{i-1}$,} and using~\eqref{eq: wt in terms of phi}, we can write:
\begin{align}
\bphit_{i-1}-\bpsit_{i}+\bz_{i-1}&=\left(I_M-\left(I_M-\frac{\mu}{\zeta}\bcH_{i-1}\right)\cA'\right)\bphit_{i-1}+\frac{\mu}{\zeta}\bs_i-\frac{\mu}{\zeta} b+\bz_{i-1}\notag\\
&\overset{\text{(a)}}=\left(I_M-\left(I_M-\frac{\mu}{\zeta}\bcH_{i-1}\right)\cA'\right)\bphit^z_{i-1}+\frac{\mu}{\zeta}\bs_i-\frac{\mu}{\zeta} b+\left(2I_M-\left(I_M-\frac{\mu}{\zeta}\bcH_{i-1}\right)\cA'\right)\bz_{i-1}, \label{eq: expression 5 biased}
\end{align}
where in (a) we used the fact that $\bphit_{i-1}=\bphit^z_{i-1}+\bz_{i-1}$. By multiplying both sides of~\eqref{eq: expression 5 biased} by $\cV_{\epsilon}^{-1}$ and by using~\eqref{eq: transformed variable phi definition}--\eqref{eq: quantization transformed vector},~\eqref{eq: definition of D11 1}--\eqref{eq: definition of D22 1}, and the Jordan decomposition of the matrix $\cA'$, we obtain:
\begin{equation}
\label{eq: expression 8 biased}
\begin{split}
\left[\begin{array}{c}
\overline{\bchi}_{i}\\
\widecheck{\bchi}_i
\end{array}
\right]=&\left[\begin{array}{cc}
\frac{\mu}{\zeta}\bcD_{11,i-1}&\frac{\mu}{\zeta}\bcD_{12,i-1}\\
\frac{\mu}{\zeta}\bcD_{21,i-1}&I_{M-P}-\cJ'_{\epsilon}+\frac{\mu}{\zeta}\bcD_{22,i-1}
\end{array}\right]\left[\begin{array}{c}
\bphib^z_{i-1}\\
\bphic^z_{i-1}
\end{array}
\right]+\frac{\mu}{\zeta}\left[\begin{array}{c}
\bsb_i\\
\bsc_{i}
\end{array}
\right]-\frac{\mu}{\zeta}\left[\begin{array}{c}
0\\
\widecheck{b}
\end{array}
\right]+\\
&\left[\begin{array}{cc}
I_P+\frac{\mu}{\zeta}\bcD_{11,i-1}&\frac{\mu}{\zeta}\bcD_{12,i-1}\\
\frac{\mu}{\zeta}\bcD_{21,i-1}&2I_{M-P}-\cJ'_{\epsilon}+\frac{\mu}{\zeta}\bcD_{22,i-1}
\end{array}\right]\left[\begin{array}{c}
\bzb_{i-1}\\
\bzc_{i-1}
\end{array}
\right].
\end{split}
\end{equation}
Again, by using similar arguments as those used to  establish inequalities (119) and (124) in~\cite[Appendix D]{nassif2020adaptation}, we can {verify} that:
\begin{equation}
\label{eq: centralized error vector recursion inequality on delta biased}
\begin{split}
\expec\|\overline{\bchi}_{i}\|^2&=\expec\left\|(I_P+\frac{\mu}{\zeta}\bcD_{11,i-1})\bzb_{i-1}+\frac{\mu}{\zeta}\bcD_{12,i-1}\bzc_{i-1}+\frac{\mu}{\zeta}\bcD_{11,i-1}\bphib^z_{i-1}+\frac{\mu}{\zeta}\bcD_{12,i-1}\bphic^z_{i-1}\right\|^2+\frac{\mu^2}{\zeta^2}\expec\|\bsb_{i}\|^2\\
&\leq2(1+\frac{\mu}{\zeta}\sigma_{11})^2\expec\|\bzb_{i-1}\|^2+\frac{2\mu^2}{\zeta^2}\expec\|\bcD_{12,i-1}\bzc_{i-1}+\bcD_{11,i-1}\bphib^z_{i-1}+\bcD_{12,i-1}\bphic^z_{i-1}\|^2+\frac{\mu^2}{\zeta^2}\expec\|\bsb_{i}\|^2\\
&\leq2(1+\frac{\mu}{\zeta}\sigma_{11})^2\expec\|\bzb_{i-1}\|^2+\frac{6\mu^2\sigma^2_{12}}{\zeta^2}\expec\|\bzc_{i-1}\|^2+\frac{6\mu^2\sigma^2_{11}}{\zeta^2}\expec\|\bphib^z_{i-1}\|^2+\frac{6\mu^2\sigma^2_{12}}{\zeta^2}\expec\|\bphic^z_{i-1}\|^2+\frac{\mu^2}{\zeta^2}\expec\|\bsb_{i}\|^2\\
&\overset{\text{(a)}}\leq\left(2(1+\frac{\mu}{\zeta}\sigma_{11})^2+\frac{6\mu^2\sigma^2_{12}}{\zeta^2}\right)\expec\|\cV_{\epsilon}^{-1}\bz_{i-1}\|^2+\frac{6\mu^2\sigma^2_{11}}{\zeta^2}\expec\|\bphib^z_{i-1}\|^2+\frac{6\mu^2\sigma^2_{12}}{\zeta^2}\expec\|\bphic^z_{i-1}\|^2+\frac{\mu^2}{\zeta^2}\expec\|\bsb_{i}\|^2,
\end{split}
\end{equation}
and
\begin{equation}
\label{eq: centralized error vector recursion inequality on delta check biased}
\begin{split}
\expec\|\widecheck{\bchi}_{i}\|^2&=\expec\left\|(2I-\cJ'_{\epsilon})\bzc_{i-1}+\frac{\mu}{\zeta}\bcD_{22,i-1}\bzc_{i-1}+\frac{\mu}{\zeta}\bcD_{21,i-1}\bzb_{i-1}+\frac{\mu}{\zeta}\bcD_{21,i-1}\bphib^z_{i-1}+(I-\cJ'_{\epsilon})\bphic^z_{i-1}+\frac{\mu}{\zeta}\bcD_{22,i-1}\bphic^z_{i-1}-\frac{\mu}{\zeta}\widecheck{b}\right\|^2\\
&\qquad+\frac{\mu^2}{\zeta^2}\expec\|\bsc_{i}\|^2\\
&\leq2\|2I-\cJ'_{\epsilon}\|^2\expec\|\bzc_{i-1}\|^2+\frac{4\mu^2}{\zeta^2}\expec\|\bcD_{22,i-1}\bzc_{i-1}+\bcD_{21,i-1}\bzb_{i-1}+\bcD_{21,i-1}\bphib^z_{i-1}+\bcD_{22,i-1}\bphic^z_{i-1}-\widecheck{b}\|^2+\\
&\qquad 4\|I-\cJ'_{\epsilon}\|^2\expec\|\bphic^z_{i-1}\|^2+\frac{\mu^2}{\zeta^2}\expec\|\bsc_{i}\|^2\\
&\leq\left(2\|2I-\cJ'_{\epsilon}\|^2+\frac{20\mu^2\sigma^2_{22}}{\zeta^2}\right)\expec\|\bzc_{i-1}\|^2+\frac{20\mu^2\sigma^2_{21}}{\zeta^2}\expec\|\bzb_{i-1}\|^2+\frac{20\mu^2\sigma^2_{21}}{\zeta^2}\expec\|\bphib^z_{i-1}\|^2+\frac{20\mu^2}{\zeta^2}\|\widecheck{b}\|^2+\\
&\qquad\left(4\|I-\cJ'_{\epsilon}\|^2+\frac{20\mu^2\sigma^2_{22}}{\zeta^2}\right)\expec\|\bphic^z_{i-1}\|^2+\frac{\mu^2}{\zeta^2}\expec\|\bsc_{i}\|^2\\
&\overset{\text{(b)}}\leq\left(2\|2I-\cJ'_{\epsilon}\|^2+\frac{20\mu^2(\sigma^2_{22}+\sigma^2_{21})}{\zeta^2}\right)\expec\|\cV_{\epsilon}^{-1}\bz_{i-1}\|^2+\frac{20\mu^2\sigma^2_{21}}{\zeta^2}\expec\|\bphib^z_{i-1}\|^2+\frac{20\mu^2}{\zeta^2}\|\widecheck{b}\|^2+\\
&\qquad\left(4\|I-\cJ'_{\epsilon}\|^2+\frac{20\mu^2\sigma^2_{22}}{\zeta^2}\right)\expec\|\bphic^z_{i-1}\|^2+\frac{\mu^2}{\zeta^2}\expec\|\bsc_{i}\|^2.
\end{split}
\end{equation}
{In steps (a) and (b) we used the fact that the norm of the components $\{\bzb_{i-1},\bzc_{i-1}\}$ is smaller than the norm of the transformed vector $\cV_{\epsilon}^{-1}\bz_{i-1}$.} By combining expressions~\eqref{eq: centralized error vector recursion inequality on delta biased} and~\eqref{eq: centralized error vector recursion inequality on delta check biased}, we obtain:
\begin{equation}
\label{eq: centralized error vector recursion inequality on sum of  delta biased}
\begin{split}
\expec\|\overline{\bchi}_{i}\|^2+\expec\|\widecheck{\bchi}_{i}\|^2\leq&\left(2\left(1+\frac{\mu}{\zeta}\sigma_{11}\right)^2+2\|2I-\cJ'_{\epsilon}\|^2+\frac{6\mu^2\sigma^2_{12}}{\zeta^2}+\frac{20\mu^2(\sigma^2_{22}+\sigma^2_{21})}{\zeta^2}\right)\expec\|\cV_{\epsilon}^{-1}\bz_{i-1}\|^2+\\
&\left(\frac{6\mu^2\sigma^2_{11}}{\zeta^2}+\frac{20\mu^2\sigma^2_{21}}{\zeta^2}\right)\expec\|\bphib^z_{i-1}\|^2+\left(4\|I-\cJ'_{\epsilon}\|^2+\frac{6\mu^2\sigma^2_{12}}{\zeta^2}+\frac{20\mu^2\sigma^2_{22}}{\zeta^2}\right)\expec\|\bphic^z_{i-1}\|^2+\\
&\frac{20\mu^2}{\zeta^2}\|\widecheck{b}\|^2+\frac{\mu^2}{\zeta^2}\left(\expec\|\bsb_{i}\|^2+\expec\|\bsc_{i}\|^2\right),
\end{split}
\end{equation}
and by using the bound~\eqref{eq: gradient noise term bound biased} in~\eqref{eq: centralized error vector recursion inequality on sum of  delta biased}, we  get:
\begin{equation}
\label{eq: centralized error vector recursion inequality on sum of  delta 1 biased}
\begin{split}
\expec\|\overline{\bchi}_{i}\|^2+\expec\|\widecheck{\bchi}_{i}\|^2\leq&\left(2(1+\frac{\mu}{\zeta}\sigma_{11})^2+2\|2I-\cJ'_{\epsilon}\|^2+\frac{6\mu^2\sigma^2_{12}}{\zeta^2}+\frac{20\mu^2(\sigma^2_{22}+\sigma^2_{21})}{\zeta^2}+2\frac{\mu^2}{\zeta^2}v_1^2\beta_{s,\max}^2 v_2^2\right)\expec\|\cV_{\epsilon}^{-1}\bz_{i-1}\|^2\\
&+\left(\frac{6\mu^2\sigma^2_{11}}{\zeta^2}+\frac{20\mu^2\sigma^2_{21}}{\zeta^2}+2\frac{\mu^2}{\zeta^2}v_1^2\beta_{s,\max}^2 v_2^2\right)\expec\|\bphib^z_{i-1}\|^2+\\
&\left(4\|I-\cJ'_{\epsilon}\|^2+\frac{6\mu^2\sigma^2_{12}}{\zeta^2}+\frac{20\mu^2\sigma^2_{22}}{\zeta^2}+2\frac{\mu^2}{\zeta^2}v_1^2\beta_{s,\max}^2 v_2^2\right)\expec\|\bphic^z_{i-1}\|^2+\frac{20\mu^2}{\zeta^2}\|\widecheck{b}\|^2+\frac{\mu^2}{\zeta^2} v_1^2\overline{\sigma}^2_s.
\end{split}
\end{equation}
Finally, by using~\eqref{eq: centralized error vector recursion inequality on sum of  delta 1 biased} in~\eqref{eq: expression 10 biased}, we find the following inequality that describes the evolution of the compression error vector~$\bz_i$:
\begin{equation}
\label{eq: expression 20 biased}
\begin{split}
&\expec\|\cV_{\epsilon}^{-1}\bz_i\|^2\leq\\
&{\beta}_{c,\max}^2v_1^2v_2^2\left(2\zeta^2(1+\frac{\mu}{\zeta}\sigma_{11})^2+2\zeta^2\|2I-\cJ'_{\epsilon}\|^2+6\mu^2\sigma^2_{12}+20\mu^2(\sigma^2_{22}+\sigma^2_{21})+2\mu^2v_1^2\beta_{s,\max}^2 v_2^2\right)\expec\|\cV_{\epsilon}^{-1}\bz_{i-1}\|^2\\
&+{\beta}_{c,\max}^2v_1^2v_2^2\left(6\mu^2\sigma^2_{11}+20\mu^2\sigma^2_{21}+2\mu^2v_1^2\beta_{s,\max}^2 v_2^2\right)\expec\|\bphib^z_{i-1}\|^2\\
&+{\beta}_{c,\max}^2v_1^2v_2^2\left(4\zeta^2\|I-\cJ'_{\epsilon}\|^2+6\mu^2\sigma^2_{12}+20\mu^2\sigma^2_{22}+2\mu^2v_1^2\beta_{s,\max}^2 v_2^2\right)\expec\|\bphic^z_{i-1}\|^2\\
&+20{\beta}_{c,\max}^2v_1^2v_2^2\mu^2\|\widecheck{b}\|^2+\mu^2 {\beta}_{c,\max}^2v_2^2v_1^4\overline{\sigma}^2_s+\zeta^2v_1^2\overline{\sigma}^2_c.
\end{split}
\end{equation}

From~\eqref{eq: centralized error vector recursion inequality with bounding gradient noise biased},~\eqref{eq: centralized error vector recursion inequality 2 with bounding gradient noise biased}, {and~\eqref{eq: expression 20 biased}}, we finally find that {$\expec\|\bphib^z_i\|^2$, $\expec\|\bphic^z_i\|^2$, and $\expec\|\cV_{\epsilon}^{-1}\bz_i\|^2$} are coupled and recursively bounded as:
\begin{equation}
\label{eq: single inequality recursion text new biased}
\left[\begin{array}{c}
\expec\|\bphib^z_i\|^2\\
\expec\|\bphic^z_i\|^2\\
\expec\|\cV_{\epsilon}^{-1}\bz_i\|^2
\end{array}\right]\preceq\Gamma\left[\begin{array}{c}
\expec\|\bphib^z_{i-1}\|^2\\
\expec\|\bphic^z_{i-1}\|^2\\
\expec\|\cV_{\epsilon}^{-1}\bz_{i-1}\|^2
\end{array}\right]+\left[\begin{array}{c}
l\\
m\\
n
\end{array}\right],
\end{equation}
where $\Gamma$ is the  $3\times 3$ matrix given by:
\begin{equation}
\Gamma=\left[\begin{array}{ccc}
a&b&c\\
d&e&f\\
g&h&j
\end{array}\right],\label{eq: general form for gamma biased}
\end{equation}
with
\begin{align}
a&\triangleq 1-\mu\sigma_{11}+2\mu^2 v_1^2\beta_{s,\max}^2 v_2^2=1-\mu\sigma_{11}+O(\mu^2),\label{eq: constant a definition new biased}\\
b&\triangleq \frac{3\mu\sigma_{12}^2}{\sigma_{11}}+2\mu^2v_1^2\beta_{s,\max}^2 v_2^2=O(\mu),\\
c&\triangleq 3\mu\sigma_{11}+\frac{3\mu\sigma_{12}^2}{\sigma_{11}}+2\mu^2v_1^2\beta_{s,\max}^2 v_2^2=O(\mu),\\
d&\triangleq \frac{10\mu^2\sigma_{21}^2}{1-\|\cJ''_{\epsilon}\|}+2\mu^2v_1^2\beta_{s,\max}^2 v_2^2=O(\mu^2),\\
e&\triangleq \|\cJ''_{\epsilon}\|+ \frac{10\mu^2\sigma_{22}^2}{1-\|\cJ''_{\epsilon}\|}+2\mu^2v_1^2\beta_{s,\max}^2 v_2^2=\|\cJ''_{\epsilon}\|+O(\mu^2),\\
f&\triangleq \frac{2\zeta^2 \|I-\cJ'_{\epsilon}\|^2}{1-\|\cJ''_{\epsilon}\|}+\frac{10\mu^2(\sigma_{22}^2+\sigma_{21}^2)}{1-\|\cJ''_{\epsilon}\|}+2\mu^2v_1^2\beta_{s,\max}^2 v_2^2=\frac{2\zeta^2 \|I-\cJ'_{\epsilon}\|^2}{1-\|\cJ''_{\epsilon}\|}+O(\mu^2),\\
g&\triangleq {\beta}_{c,\max}^2v_1^2v_2^2\left(6\mu^2 \sigma^2_{11}+20\mu^2 \sigma^2_{21}+2\mu^2 v_1^2\beta_{s,\max}^2 v_2^2\right)=O(\mu^2 ),\\
h&\triangleq 
{\beta}_{c,\max}^2v_1^2v_2^2\left(4\zeta^2\|I-\cJ'_{\epsilon}\|^2+6\mu^2 \sigma^2_{12}+20\mu^2 \sigma^2_{22}+2\mu^2 v_1^2\beta_{s,\max}^2 v_2^2\right)=4\zeta^2{\beta}_{c,\max}^2v_1^2v_2^2\|I-\cJ'_{\epsilon}\|^2+O(\mu^2 ),\\
j&\triangleq  2{\beta}_{c,\max}^2v_1^2v_2^2\left(\zeta^2(1+\frac{\mu}{\zeta}\sigma_{11})^2+\zeta^2\|2I-\cJ'_{\epsilon}\|^2+3\mu^2\sigma^2_{12}+10\mu^2 (\sigma^2_{22}+\sigma^2_{21})+\mu^2v_1^2\beta_{s,\max}^2 v_2^2\right)\notag\\
&= 2\zeta^2 {\beta}_{c,\max}^2v_1^2v_2^2\left((1+\frac{\mu}{\zeta}\sigma_{11})^2+\|2I-\cJ'_{\epsilon}\|^2\right)+O(\mu^2),\\
l&\triangleq\mu^2 v_1^2\overline{\sigma}^2_s=O(\mu^2),\label{eq: constant j definition new biased}\\
m&\triangleq \mu^2  v_1^2\overline{\sigma}^2_s+\frac{10\mu^2 \|\widecheck{b}\|^2}{1-\|\cJ''_{\epsilon}\|}=O(\mu^2),\label{eq: constant k definition new biased}\\
n&\triangleq \zeta^2v_1^2\overline{\sigma}^2_c +\mu^2  v_1^4v_2^2\beta_{c,\max}^2\overline{\sigma}^2_s+20\beta_{c,\max}^2v_1^2v_2^2\mu^2 \|\widecheck{b}\|^2= \zeta^2v_1^2\overline{\sigma}^2_c+O(\mu^2).\label{eq: constant l definition new biased}
\end{align}
If the matrix $\Gamma$ is stable, i.e., $\rho(\Gamma)<1$, then by iterating~\eqref{eq: single inequality recursion text new biased}, we arrive at:
\begin{align}
\label{eq: single inequality recursion steady state 2 new biased}
\limsup_{i\rightarrow\infty}\left[\begin{array}{c}
\expec\|\bphib_i^z\|^2\\
\expec\|\bphic_i^z\|^2\\
\expec\|\cV_{\epsilon}^{-1}\bz_i\|^2
\end{array}\right]\preceq(I_3-\Gamma)^{-1}\left[\begin{array}{c}
l\\
m\\
n
\end{array}\right].
\end{align}
As we will see in the following, for some given {learning problem} settings (captured by $\{\sigma_{11}^2,\sigma_{12}^2,\sigma_{21}^2,\sigma_{22}^2,\beta^2_{s,\max},\overline{\sigma}^2_s\}$), small step-size parameter $\mu$, network topology (captured by the matrix $\cA$ {and the variables  $\{v_1^2,v_2^2,\cJ_{\epsilon}\}$ resulting from its eigendecomposition}), and quantizer settings (captured by $\{\beta_{c,\max}^2,\overline{\sigma}^2_{c}\}$), the stability of $\Gamma$ can be controlled by the damping coefficient $\zeta$ and the mixing parameter $\gamma$ used in steps~\eqref{eq: step 2} and~\eqref{eq: step 3}, respectively. 
Generally speaking, and since the spectral radius of a matrix is upper bounded by its $1-$norm, the matrix $\Gamma$ is stable if:
\begin{equation}
\label{eq: condition ensuring stability of gamma new biased}
\rho(\Gamma)\leq\max\{|a|+|d|+|g|,|b|+|e|+|h|,|c|+|f|+|j|\}<1.
\end{equation}
Since $\sigma_{11}>0$
, a sufficiently small $\mu$ can make $|a|+|d|+|g|$ strictly smaller than 1. For $|b|+|e|+|h|$, observe that if the damping coefficient $\zeta$ and the mixing parameter $\gamma$ are chosen such that:
\begin{equation}
\label{eq: condition on the mixing parameter new biased}
\|\cJ''_{\epsilon}\|+4\zeta^2v_1^2v_2^2\beta^2_{c,\max}\|I-\cJ'_{\epsilon}\|^2<1,
\end{equation}
then $|b|+|e|+|h|$ can be made strictly smaller than  {$1$} for sufficiently small $\mu$. Finally, for $|c|+|f|+|j|$, observe that if the parameters $\gamma$ and $\zeta$ are chosen such that:
\begin{equation}
\label{eq: condition on the mixing parameter new biased 1}
\frac{2\zeta^2\|I-\cJ'_{\epsilon}\|^2}{1-\|\cJ''_{\epsilon}\|}+2 {\beta}_{c,\max}^2{\zeta^2}v_1^2v_2^2\left(1+\|2I-\cJ'_{\epsilon}\|^2\right)<1,
\end{equation}
then $|c|+|f|+|j|$ can be made strictly smaller than {$1$} for sufficiently small $\mu$. It is therefore clear that the RHS of~\eqref{eq: condition ensuring stability of gamma new biased} can be made strictly smaller than  {$1$} for sufficiently small $\mu$ and for a damping coefficient $\zeta\in(0,1]$ and mixing parameter $\gamma\in(0,1]$ satisfying conditions~\eqref{eq: condition on the mixing parameter new biased} and~\eqref{eq: condition on the mixing parameter new biased 1}. In the following, we analyze in details conditions~\eqref{eq: condition on the mixing parameter new biased} and~\eqref{eq: condition on the mixing parameter new biased 1}. By following similar arguments as in~\cite[pp.~516--517]{sayed2014adaptation}, we can establish the following identities on the block diagonal matrices $I-\cJ_{\epsilon}'$ and $2I-\cJ_{\epsilon}'$ appearing in conditions~\eqref{eq: condition on the mixing parameter new biased} and~\eqref{eq: condition on the mixing parameter new biased 1}:
\begin{align}
\|I-\cJ_{\epsilon}'\|^2&\overset{\eqref{eq: jordan decomposition of A'}}=\|\gamma(I-\cJ_{\epsilon})\|^2 =\gamma^2\|I-\cJ_{\epsilon}\|^2 \leq\gamma^2(\rho(I-\cJ_{\epsilon})+\epsilon)^2,\label{eq: bound on I-J epsilon' biased}\\
\|2I-\cJ_{\epsilon}'\|&\overset{\eqref{eq: jordan decomposition of A'}}=\|(1+\gamma)I-\gamma\cJ_{\epsilon}\|\leq(1+\gamma)-\gamma\left(\rho(\cJ_{\epsilon})+\epsilon\right)\in(1,1+\gamma),\label{eq: bound on 2I-J epsilon' biased}
\end{align}
where $\rho(I-\cJ_{\epsilon})\in(0,2)$ since $\rho(\cJ_{\epsilon})\in(0,1)$. By using the bounds~\eqref{eq: spectral norm of J''} and~\eqref{eq: bound on I-J epsilon' biased} into~\eqref{eq: condition on the mixing parameter new biased}, we can upper bound the LHS of~\eqref{eq: condition on the mixing parameter new biased} by:
\begin{equation}
\|\cJ''_{\epsilon}\|+4v_1^2v_2^2\beta^2_{c,\max}\zeta^2\|I-\cJ'_{\epsilon}\|^2\leq 1-\gamma\zeta(1-\rho(\cJ_{\epsilon})-\epsilon)+4v_1^2v_2^2\beta^2_{c,\max}(\gamma\zeta)^2(\rho(I-\cJ_{\epsilon})+\epsilon)^2.
\end{equation}
The upper bound in the above inequality is guaranteed to be strictly smaller than $1$ if:
\begin{equation}
\label{eq: condition on 1-zeta to solve}
4v_1^2v_2^2\beta^2_{c,\max}(\gamma\zeta)^2(\rho(I-\cJ_{\epsilon})+\epsilon)^2-\gamma\zeta(1-\rho(\cJ_{\epsilon})-\epsilon)<0.
\end{equation}
Now, by using the above condition and the fact that $\gamma\zeta$ must be in $(0,1]$, we  obtain {condition~\eqref{eq: mixing parameter condition theorem} on} $\gamma\zeta$.
For the second condition~\eqref{eq: condition on the mixing parameter new biased 1},  we start by noting that its LHS can be upper bounded by:
\begin{equation}
\label{eq: mixing parameter condition theorem biased 1}
\begin{split}
&\frac{2\zeta^2\|I-\cJ'_{\epsilon}\|^2}{1-\|\cJ''_{\epsilon}\|}+ 2{\beta}_{c,\max}^2\zeta^2v_1^2v_2^2\left(1+\|2I-\cJ'_{\epsilon}\|^2\right)\\
&\overset{\eqref{eq: spectral norm of J''},\eqref{eq: bound on I-J epsilon' biased},\eqref{eq: bound on 2I-J epsilon' biased}}\leq 2\gamma\zeta\frac{(\rho(I-\cJ_{\epsilon})+\epsilon)^2}{1-(\rho(\cJ_{\epsilon})+\epsilon)}  +{2{\beta}_{c,\max}^2\zeta^2v_1^2v_2^2\left(1+\left((1+\gamma)-\gamma(\rho(\cJ_{\epsilon})+\epsilon)\right)^2\right)}
\end{split}
\end{equation}
{Thus,~\eqref{eq: condition on the mixing parameter new biased 1} is guaranteed to be satisfied under condition~\eqref{eq: mixing parameter condition theorem 2}.}

Under conditions~\eqref{eq: mixing parameter condition theorem} and~\eqref{eq: mixing parameter condition theorem 2}, $\rho(\Gamma)<1$, and consequently, the matrix $\Gamma$ is stable. Moreover, it holds that\footnote{{While constants crucial to understanding the algorithm's behavior are written explicitly  in~\eqref{eq: I-Gamma new biased} and~\eqref{eq: inverse of I-Gamma new biased}, the other constants that are less significant are encapsulated in the Big $O$ notation.}}:
\begin{equation}
\label{eq: I-Gamma new biased}
(I-\Gamma)=\left[\begin{array}{ccc}
\mu\sigma_{11}&O(\mu)&O(\mu)\\
O(\mu^2)&1-e&f\\
O(\mu^2)&h&1-j
\end{array}\right],
\end{equation}
and:
\begin{equation}
\label{eq: inverse of I-Gamma new biased}
(I-\Gamma)^{-1}=\left[\begin{array}{ccc}
\frac{1}{\mu\sigma_{11}}&O(1)&O(1)\\
O(\mu)&O(1)&O(1)\\
O(\mu)&O(1)&O(1)
\end{array}\right].
\end{equation}
Now, using~\eqref{eq: constant j definition new biased},~\eqref{eq: constant k definition new biased},~\eqref{eq: constant l definition new biased}, and~\eqref{eq: inverse of I-Gamma new biased} into~\eqref{eq: single inequality recursion steady state 2 new biased}, we arrive at:
\begin{align}
\label{eq: single inequality recursion steady state 3 new biased}
&\limsup_{i\rightarrow\infty}\left[\begin{array}{c}
\expec\|\bphib^z_i\|^2\\
\expec\|\bphic^z_i\|^2\\
\expec\|\cV_{\epsilon}^{-1}\bz_i\|^2
\end{array}\right]\preceq\left[\begin{array}{c}
\mu v_1^2\frac{\overline{\sigma}^2_s}{\sigma_{11}}+O(\mu^2)+\overline{\sigma}^2_{c} O(1)\\
O(\mu^2)+\overline{\sigma}^2_{c}  O(1)\\
O(\mu^2)+\overline{\sigma}^2_{c} O(1)
\end{array}\right].
\end{align}
By noting that:
\begin{align}
\limsup_{i\rightarrow\infty}\expec\|\bphit^z_i\|^2=\limsup_{i\rightarrow\infty}\expec\|\cV_{\epsilon}(\cV_{\epsilon}^{-1}\bphit^z_i)\|^2&\leq\limsup_{i\rightarrow\infty}\|\cV_{\epsilon}\|^2\left[\expec\|\bphib^z_i\|^2+\expec\|\bphic^z_i\|^2\right]\notag\\
&\hspace{-1mm}\overset{\eqref{eq: single inequality recursion steady state 3 new biased}}={\kappa\mu+O(\mu^2)+\overline{\sigma}^2_{c} O(1)},\label{eq: equation relating w tilde to phi tilde biased}
\end{align}
we can finally conclude {that~\eqref{eq: steady state mean square error} holds}. {Result~\eqref{eq: steady state mean square error phitilde} follows from~\eqref{eq: equation relating w tilde to phi tilde biased} by replacing $\overline{\sigma}^2_cO(1)$ by $O(\mu^{1+\varepsilon})$.}

To establish~\eqref{eq: steady state mean square error w}, we first show that the mean-square difference between the trajectories $\{\bphit^z_i,\bcwt_i\}$ is asymptotically bounded by $O(\mu^{1+\varepsilon})$. By subtracting $\bcwt_i$ and $\bphit^z_i$, we can write:
\begin{align}
\limsup_{i\rightarrow\infty}\expec\|\bcwt_i-\bphit^z_i\|^2&\overset{\eqref{eq: wt in terms of phi}}=\limsup_{i\rightarrow\infty}\expec\|\cA'\bphit_i-\bphit^z_i\|^2\notag\\
&=\limsup_{i\rightarrow\infty}\expec\|\cA'(\bphit^z_i+\bz_i)-\bphit^z_i\|^2\notag\\
&=\limsup_{i\rightarrow\infty}\expec\left\|\cV_{\epsilon}(\Lambda_{\epsilon}'-I_M)\cV_{\epsilon}^{-1}\bphit^z_i+\cV_{\epsilon}\Lambda_{\epsilon}'\cV_{\epsilon}^{-1}\bz_i\right\|^2\notag\\
&\hspace{-3.3mm}\overset{{\eqref{eq: jordan decomposition of A},}\eqref{eq: jordan decomposition of A'}}=\limsup_{i\rightarrow\infty}\expec\left\|\cV_{R,\epsilon}(\cJ'_{\epsilon}-I)\bphic^z_i+\cV_{\epsilon}\Lambda_{\epsilon}'\cV_{\epsilon}^{-1}\bz_i\right\|^2\notag\\
&\leq\limsup_{i\rightarrow\infty}\left[2\|\cV_{R,\epsilon}(\cJ'_{\epsilon}-I)\|^2\expec\|\bphic^z_i\|^2+2\|\cV_{\epsilon}\Lambda'_{\epsilon}\|^2\expec\left\|\cV_{\epsilon}^{-1}\bz_i\right\|^2\right].\label{eq: equation relating w tilde to phi tilde biased and w tilde}
\end{align}
Now, using~\eqref{eq: single inequality recursion steady state 3 new biased} with $\overline{\sigma}^2_cO(1)$ replaced by $O(\mu^{1+\varepsilon})$, we can conclude that:
\begin{equation}
\limsup_{i\rightarrow\infty}\expec\|\bcwt_i-\bphit^z_i\|^2=O(\mu^{1+\varepsilon}).\label{eq: intermediate result on the difference of the trajectories}
\end{equation}
Finally, note that:
\begin{align}
\expec\|\bcwt_i\|^2&=\expec\|\bcwt_i-\bphit^z_i+\bphit^z_i\|^2\notag\\
&\leq\expec\|\bcwt_i-\bphit^z_i\|^2+\expec\|\bphit^z_i\|^2+2|\expec(\bcwt_i-\bphit^z_i)^\top\bphit^z_i|\notag\\
&\overset{\text{(a)}}\leq\expec\|\bcwt_i-\bphit^z_i\|^2+\expec\|\bphit^z_i\|^2+2\sqrt{\expec\|\bcwt_i-\bphit^z_i\|^2\expec\|\bphit^z_i\|^2},\label{eq: equation relating w tilde to phi tilde biased and w tilde}
\end{align}
and, hence,  {from the  sub-additivity property of the limit superior, and} from~\eqref{eq: steady state mean square error phitilde} and~\eqref{eq: intermediate result on the difference of the trajectories}, we get:
\begin{equation}
\label{eq: limsup relation on the vector difference}
{\limsup_{i\rightarrow\infty}\expec\|\bcwt_i\|^2\leq \limsup_{i\rightarrow\infty}\expec\|\bphit^z_i\|^2+O(\mu^{1+\frac{\varepsilon}{2}}),}
\end{equation}
which establishes~\eqref{eq: steady state mean square error w}. In step (a), we used $|\expec\bx|\leq \expec|\bx|$ from Jensen's inequality and we applied Holder's {inequality, namely, $\expec|\bx^\top\by|\leq(\expec|\bx|^p)^\frac{1}{p}(\expec|\by|^q)^\frac{1}{q}$ when $1/p+1/q=1$, with $p=q=2$}.

The analysis can be simplified in settings where the compression operators $\{\bcC_k\}$ are such that their relative compression noise terms $\beta^2_{c,k}=0$, $\forall k$. In fact, in such settings, we can replace $\beta^2_{c,\max}$ in~\eqref{eq: expression 20 biased} by $0$, and use the resulting inequality $\expec\|\cV_{\epsilon}^{-1}\bz_{i}\|^2\leq \zeta^2v_1^2\overline{\sigma}^2_c$ 
directly into~\eqref{eq: centralized error vector recursion inequality with bounding gradient noise biased} and~\eqref{eq: centralized error vector recursion inequality 2 with bounding gradient noise biased}, without the need to study the evolution of $\expec\|\cV_{\epsilon}^{-1}\bz_{i}\|^2$ as in~\eqref{eq: single inequality recursion text new biased}. By doing so, we find that the variances of $\bphib^z_i$ and $\bphic^z_i$ are coupled and recursively bounded as:
 \begin{equation}
\label{eq: single inequality recursion text new biased 2 times 2}
\left[\begin{array}{c}
\expec\|\bphib^z_i\|^2\\
\expec\|\bphic^z_i\|^2
\end{array}\right]\preceq\left[\begin{array}{cc}
a&b\\
d&e
\end{array}\right]\left[\begin{array}{c}
\expec\|\bphib^z_{i-1}\|^2\\
\expec\|\bphic^z_{i-1}\|^2
\end{array}\right]+\left[\begin{array}{c}
l'\\
m'
\end{array}\right],
\end{equation}
where $l' = l+\overline{\sigma}^2_cO(\mu)$, and $m'= m+\overline{\sigma}^2_c O(1)$. By setting $\gamma=\zeta=1$ (so that $\cJ''_{\epsilon}=\cJ_{\epsilon}$), and by using similar arguments as in~\eqref{eq: single inequality recursion steady state 2 new biased}--\eqref{eq: equation relating w tilde to phi tilde biased},  we can  establish the mean-square-error stability for a sufficiently small step-size $\mu$ and obtain the performance result~\eqref{eq: steady state mean square error}.

\section{Bit Rate stability}
\label{app: rate stability}
Equation~\eqref{eq: sigma square} follows  from Lemma~\ref{lemma: Property of the top-$c$ quantizer} and Table~\ref{table: examples of quantizers} (row~$3$, col. $3$--$4$). Invoking similar arguments to the ones used to establish Theorem 2 in~\cite{nassif2023quantization}, the individual summands in~\eqref{eq: bit rate general at agent k} can be upper bounded by:
\begin{equation}
2+\log_2\left(\frac{\ln\left(1+\displaystyle{\frac{\omega}{\eta}}\sqrt{\expec\|\bchi_{k,i}\|^2}\right)}{2\ln\left(\omega+\sqrt{1+\omega^2}\right)}+2\right).
\label{eq:finalrateformula}
\end{equation}
By taking the limit superior of~\eqref{eq: centralized error vector recursion inequality on sum of  delta 1 biased} as $i\rightarrow \infty$ and by using~\eqref{eq: single inequality recursion steady state 3 new biased} and~\eqref{eq: sigma square}, we obtain:
\begin{equation}
\label{eq: appendix c}
\limsup_{i\rightarrow\infty}\expec\|\bchi_{k,i}\|^2\leq \kappa_1\mu^{1+\varepsilon}+\kappa_2\mu^2,
\end{equation}
for sufficiently small $\mu$, and where $\kappa_1$ and $\kappa_2$ are some positive constants independent of $\mu$. Applying the limit superior to~\eqref{eq:finalrateformula},  using the fact that~\eqref{eq:finalrateformula} is a continuous and increasing function in the argument $\expec\|\bchi_{k,i}\|^2$, and using~\eqref{eq: appendix c}, in view of~\eqref{eq: sigma square}, we find that each summand in~\eqref{eq: bit rate general at agent k} is $O(1)$, which in turn implies~\eqref{eq: rate stability result}.

\end{appendices}
\bibliographystyle{IEEEbib}
{\balance{
\bibliography{reference}}}

\end{document}